\documentclass[acmsmall,nonacm]{acmart}

\setcopyright{none}             
\usepackage{xcolor}
\usepackage{subfigure}
\usepackage{enumerate}
\usepackage{oz}
\usepackage{refcal}
\usepackage{url}
\usepackage{rcoms}

\usepackage{datetime}
\usepackage{xspace}
\usepackage{soundness}
\usepackage{counters}

\usepackage{opsem}
\usepackage{wmm}
\usepackage{rcp}

\usepackage{stmaryrd}
\usepackage{anyfontsize}
\usepackage{relsize}

\usepackage[strings]{underscore}
\usepackage{nameref}

\numberwithin{equation}{section}

\renewcommand{\asgn}{\mathop{{:}{=}}}

\newcommand{\vecaca}{\vec{\aca}}

\renewcommand{\AllStores}{\mathbb{S}}
\renewcommand{\AllLoads}{\mathbb{L}}

\renewcommand{\fencep}[1]{\mathbf{br}(#1)}
\renewcommand{\fencep}[1]{#1}
\renewcommand{\fullfencespec}{\mathsf{full_{fnc}}}
\renewcommand{\loadfencespec}{\mathsf{load_{fnc}}}
\renewcommand{\storefencespec}{\mathsf{store_{fnc}}}

\renewcommand{\fwd}[2]{{}_{#1\mbox{\raisebox{-0.5pt}{\guillemetright}}}#2}
\newcommand{\fwda}[1]{\fwd{\aca}{#1}}

\newcommand{\MP}{\T{mp}}
\newcommand{\MPC}{\MP}
\newcommand{\MPra}{\MP^{RA}}

\newcommand{\plain}[1]{{#1}^{-}}
\newcommand{\plainc}{\plain{\cmdc}}
\newcommand{\plaincp}{\plain{\cmdc'}}

\renewcommand{\refeqn}[1]{(\ref{eqn:#1})}
\renewcommand{\refeqns}[2]{\refeqn{#1} and~\refeqn{#2}}

\renewcommand{\reffig}[1]{Fig.~\ref{fig:#1}}

\renewcommand{\refsect}[1]{Sect.~\ref{sect:#1}}

\newcommand{\labeldefn}[1]{\label{defn:#1}}

\newcommand{\arefeq}{&\refeq&}
\newcommand{\arefsto}{&\refsto&}
\newcommand{\seps}[2]{#2#1#2} 
\newcommand{\impseps}[1]{\seps{\imp}{#1}} 
\newcommand{\impS}{\impseps{\sspace}}

\newcommand{\atmexchexpl}[1]{\T{atomic\_exhange\_explicit}(#1)}
\newcommand{\atmexch}[1]{\T{atomic\_exhange}(#1)}
\newcommand{\atmcmpexchstrexpl}[1]{\T{atomic\_compare\_exhange\_strong\_explicit}(#1)}
\newcommand{\atmcmpexchstr}[1]{\T{atomic\_compare\_exhange\_strong}(#1)}
\newcommand{\atmfetchaddexpl}[1]{\T{atomic\_fetch\_add\_explicit}(#1)}
\newcommand{\atmfetchadd}[1]{\T{atomic\_fetch\_add}(#1)}

\newcommand{\ocsucc}{\oc_{succ}}
\newcommand{\ocfail}{\oc_{fail}}

\newcommand {\improc}{\impro{}${}^\Clang$}

\renewcommand{\id}[1]{{\sf id}_{#1}}
\newcommand{\idn}{{\sf id}}

\renewcommand{\ro}{\roM{\Cmm}}
\renewcommand{\nro}{\nroM{\Cmm}}

\newcommand{\OCmm}{\mmname{oc}}
\newcommand{\roOC}{\roM{\OCmm}}
\newcommand{\nroOC}{\nroM{\OCmm}}

\newcommand{\OCsmm}{\mmname{ocs}}
\newcommand{\roOCs}{\roM{\OCsmm}}

\definecolor{dkgreen}{rgb}{0,0.6,0}
\definecolor{gray}{rgb}{0.5,0.5,0.5}
\definecolor{mauve}{rgb}{0.58,0,0.82}

\usepackage{scalerel}

\newcommand{\Clang}{{\tt C}\xspace}
\newcommand{\Clangctee}{\Clang committee}
\newcommand{\Cpp}{{\tt C++}\xspace}
\newcommand{\rgq}[5]{\quint{#1}{#2}{#5}{#3}{#4}}
\newcommand{\rgqp}[4]{\rgq{p}{#1}{#2}{#3}{#4}}
\newcommand{\rgqpr}[3]{\rgqp{r}{#1}{#2}{#3}}
\newcommand{\rgqprg}[2]{\rgqpr{g}{#1}{#2}}
\newcommand{\rgqprgq}[1]{\rgqprg{q}{#1}}
\newcommand{\rgqprgqc}{\rgqprgq{\cmdc}}

\renewcommand{\refdefn}[1]{Defn.~(\ref{defn:#1})}
\renewcommand{\refeqn}[1]{(\ref{eqn:#1})}
\renewcommand{\reffig}[1]{Fig.~\ref{fig:#1}}
\renewcommand{\refrule}[1]{Rule~\ref{rule:#1}}

\newcommand{\lless}{\scaleobj{0.8}{<\!\!<}}
\newcommand{\roc}[3]{#1\, \lless\, #2\, \lless\, #3}

\newcommand{\indivisword}{\mathsf{indivis}}
\newcommand{\indivisact}[1]{\indivisword(#1)}
\newcommand{\indivisa}{\indivisact{\aca}}
\newcommand{\indivisb}{\indivisact{\acb}}

\renewcommand{\True}{\mathsf{True}}
\renewcommand{\False}{\mathsf{False}}

\newcommand{\tick}{\green{\checkmark}}

\newcommand{\ntick}{\red{\cross}}

\newcommand{\ceiling}[1]{\lceil#1\rceil}
\newcommand{\getocs}[1]{\ceiling{#1}}
\newcommand{\getocse}{\getocs{e}}
\newcommand{\getocsep}{\getocs{e'}}

\newcommand{\OCtype}{OC}

\newcommand{\binop}{\oplus}
\newcommand{\unop}{\ominus}

\newcommand{\relfence}{\mathbf{rel\_fence}}
\newcommand{\acqfence}{\mathbf{acq\_fence}}
\newcommand{\scfence}{\mathbf{sc\_fence}}

\newcommand{\fencemm}{\mmname{fnc}}

\newcommand{\roFNC}{\roM{\fencemm}}
\newcommand{\nroFNC}{\nroM{\fencemm}}
\title[A simpler basis for working with the \Clang\Cpp memory model (extended version)]{%
Separation of concerning things: 
a simpler basis for defining and programming with the \Clang/\Cpp memory model (extended version)
}

\begin{document}


\OMIT{
\author{
\me
\inst{1,2}
}
\institute{
Defence Science and Technology Group, Australia
\and
The University of Queensland, Australia
\email{r.colvin@uq.edu.au}
}
\date{}
}

\author{Robert J. Colvin}

\affiliation{
    \institution{
    Defence Science and Technology Group, and
	The School of Information Technology and Electrical Engineering,
    The University of Queensland
    }
    \department{
    The University of Queensland
    }
	\city{Brisbane}
	\country{Australia}
}
\email{r.colvin@uq.edu.au}


\begin{abstract}

The \Clang/\Cpp memory model provides an interface and execution model for programmers of concurrent (shared-variable) code. 
It provides a range of mechanisms that
abstract from underlying hardware memory models -- that govern how multicore architectures handle concurrent accesses to main memory --
as well as abstracting from compiler transformations.
The \Clang standard describes the memory model in terms of 
cross-thread relationships between events, and has been influenced by several research works that 
are similarly based.
In this paper we provide a thread-local definition of the fundamental principles of the \Clang memory model,
which, for concise concurrent code, serves as a basis for relatively straightforward reasoning about the effects of the \Clang ordering mechanisms.
We argue that this definition is more practical from a programming perspective and is amenable to analysis by already established techniques for concurrent code.
The key aspect is that the memory model definition is separate to other considerations of a rich programming language such as \Clang,
in particular, expression evaluation and optimisations, though we show how to reason about those considerations in the presence of \Clang concurrency.
A major simplification of our framework compared to the description in the \Clang standard and related work in the literature
is separating out considerations around the ``lack of multicopy atomicity'',
a concept that is in any case irrelevant to developers of code for x86, Arm, RISC-V or SPARC architectures.
We show how the framework is convenient for reasoning about well-structured code,
and for formally addressing unintuitive behaviours such as ``out-of-thin-air'' writes.

\end{abstract}

\acmBooktitle{Proceedings}

\maketitle

\renewcommand{\notin}{\mathbin{/\!\!\!\in}}


\section{Introduction}

\Clang/\Cpp is one of the most widely used programming languages,
including for low-level concurrent code with a high imperative for efficiency.
The \Clang (weak) memory model, governed by an ISO standard, provides an interface (\T{atomics.h}) for instrumenting shared-variable concurrency that abstracts from
the multitude of multicore architectures it may be compiled to, each with its own guarantees about and mechanisms for controlling accesses to shared variables.

The \Clang memory model standard is described in terms of a cross-thread ``happens before'' relationship, relating stores and loads within and between threads,
and ``release sequences''.  The fundamentals of this
approach were established by Boehm \& Adve \cite{BoehmAdveC++Concurrency},
and the standard has been influenced and improved by various research works (\eg, \cite{MathematizingC++,VafeiadisC11}).
However, because it is cross-thread, verification techniques are often complex
to apply, and the resulting formal semantics are often highly specialised and involve global data structures capturing a partial order on events, and rarely cover 
the full range of features available in \Clang's \T{atomics.h} \cite{VafGroundingThinAir,RepairingSCinC11}.
In many cases this is because such formalisms attempt to explain how reordering can occur by mixing considerations such as expression optimisations
and Power's cache coherence system, alongside local compiler-induced and processor pipeline reorderings.
We instead take a separation-of-concerns approach, where the three fundamental principles involved in \Clang concurrent code --
data dependencies, fences, and memory ordering constraints --
are specified separately to other aspects of \Clang which may complicate reasoning,
such as expression evaluation, expression optimisations, and arbitrary compiler transformations.
Provided a programmer steers clear of concurrent code that is subject to these extra factors, 
reasoning about their code is relatively straightforward in our framework,
but if the full richness of \Clang is insisted upon, our framework is also applicable.

Syntactically and semantically the key aspect is the
\emph{\pseqc} operator, formalising a processor-pipeline concept that
has been around since the 1960s, and which has been previously shown to have good explanatory power for most behaviours observed on hardware weak memory models
\cite{SEFM21,pseqc-arXiv}.  
Reasoning in this setting involves making explicit at the program level 
what effects the memory model has, either reducing to a sequential form where the use of \T{atomics.h} features prevents the compiler or hardware
from making problematic reorderings, or making the residual parallelism explicit; in either case, standard techniques (such as Owicki-Gries or rely/guarantee)
apply to the reduced program.
We cover a significant portion of the \Clang weak memory model, including release/acquire and release/consume synchronisation, and sequentially consistent accesses and fences.
We demonstrate verification of
some simple behaviours (litmus tests), a spin lock implementation, 
and also explain how two types of related problematic behaviours -- out-of-thin-air stores and read-from-untaken-branch -- 
can be analysed and addressed. 
We argue that this foundation provides a simpler and more direct basis for discussion of the consequences of choices within the \Clang memory model,
and in particular for analysing the soundness of compiler transformations for concurrent code.

We explain the main aspects of the \Clang memory model using a simple list-of-instructions language
in \refsect{simple-syntax}, covering all relevant aspects of the \Clang memory model's principles.
We then give the syntax and semantics of a more realistic imperative language (with conditionals and loops) in \refsect{impro}.
We give a set of reduction rules for reasoning about the effects of the memory model in \refsect{reduction},
and explain how standard reasoning techniques can be applied in \refsect{reasoning}.
We show some illustrative examples in \refsect{examples}, 
including the often-discussed ``out-of-thin-air'' behaviours,
showing how in our framework an allowed version of the pattern arises naturally,
and a disallowed version is similarly disallowed. 
In \refsectsC{incremental-evaluation}{expression-optimisations}{forwarding}
we extend the language of \refsect{impro} with other features of programming in \Clang
such as incremental (usually called non-atomic) expression evaluation and instruction execution,
expression optimisations, and forwarding of values from earlier instructions to later ones.
Crucially, despite the complexities these features introduce, the fundamental 
principles of the \Clang memory model from \refsect{simple-syntax} do not change.
In \refsect{sfp} we give a formal discussion of the ``read-from-untaken branch behaviour''
which exposes the often problematic interactions between standard compiler optimisations and 
controlling shared-variable concurrency.

\newcommand{\mca}{mca}

\section{A simple language with thread-local reordering}
\labelsect{simple-syntax}

In this section we give a simple language formed from primitive actions and ``parallelized sequential prefixing'', 
which serves to explain the crucial parts of reordering due to the \Clang memory model.  In \refsect{impro} we extend the
language to include standard imperative constructs such as conditionals, loops, and composite actions.

\subsection{Syntax and semantics of a language with instruction reordering}
\newcommand{\FenceType}{Fence}

\newcommand{\prefixop}[1]{\overset{#1}{\triangleright}}
\newcommand{\prefixopm}{\prefixop{\mm}}
\newcommand{\prefixopsc}{\mathop{\raisebox{1pt}{$\mathsmaller{\mathsmaller{\mathsmaller{\blacktriangleright}}}$}}}
\newcommand{\prefixM}[3]{#2 \prefixop{#1} #3}
\newcommand{\prefixm}[2]{\prefixM{\mm}{#1}{#2}}
\newcommand{\prefixsc}[2]{#1 \prefixopsc {#2}}
\newcommand{\prefixC}[2]{\prefixM{\Cmm}{#1}{#2}}
\newcommand{\prefixmac}{\prefixm{\aca}{\cmdc}}

To focus on the semantic point of reorderings we introduce a very basic language formed from primitive instructions representing assignments, branches, and fences, which are composed
solely by a prefix operator that allows reordering (early execution) of later instructions.  
\begin{eqnarray*}
\labeleqn{impro-simple}
	e \attdef v \csep x \csep \unop e \csep e_1 \binop e_2 
	\\
	\fencegen \attdef \storefencespec \csep \loadfencespec \csep \fullfencespec \csep \ldots
	\\
    \aca & \ttdef & x \asgn e \csep \guarde 
		\csep \fencepf
    \\
    \cmdc & \ttdef & 
		\Skip \csep 
		\prefixm{\aca}{c}
\end{eqnarray*}
Expressions $e$ can be base values $v$, variables $x$, or involve the usual unary ($\unop$) or binary ($\binop$) operators.
A primitive action $\aca$ in this language is either an assignment $x \asgn e$, where $x$ is a variable and $e$ an expression, guard $\guarde$, where $e$ is an expression,
or fence $\fencepf$, where $\fencegen$ is some fence type, described below.
A command can be the terminated command $\Nil$, or a simple prefixing of a primitive instruction $\aca$ before command $\cmdc$, 
parameterised by some memory model $\mm$,
written $\prefixmac$.
As we show elsewhere \cite{FM18,pseqc-arXiv,SEFM21}, the parameter $\mm$ can be instantiated to give the behaviour of hardware weak memory models, but in this paper we focus mostly on \Clang's
memory model, denoted formally by `\Cmm', and
the special cases of sequential and parallel composition.
The $\prefixopm$ operator essentially allows the construction of a sequence of instructions that may be reordered under some circumstances, similar to a hardware pipeline.

The semantics of the language is given operationally below.
All primitive actions are executed in a single indivisible step.%
\footnote{Normally called \emph{atomic} but we avoid that term to keep the notion separate from \Clang's notion of \T{atomics}.}
\begin{gather*}
	(i) \quad
	\prefixmac \tra{\aca} \cmdc
\qquad
\qquad
	(ii) \quad
	\Rule{
        c \tra{\acb} c'
        \qquad
        \aca \rom \acb
    }{
		\prefixmac
		\tra{\acb}
		\prefixm{\aca}{c'}
    }
\end{gather*}
A command $\prefixmac$ may either immediately execute $\aca$, (rule $(i)$) 
or it may execute some action $\acb$ of $\cmdc$ provided that $\acb$ may \emph{reorder} with $\aca$ with respect to $\mm$, written
$\aca \rom \acb$ (rule $(ii)$).
The conditions under which reordering may occur are specific to the memory model under consideration, 
and we define these below for \Clang.
The memory model parameter is defined pointwise on instruction types.  This is a relatively convenient way to express reorderings, especially as it is agnostic about global
traces and behaviours.  As shown empirically in \cite{pseqc-arXiv} it is suitable for explaining observed behaviours on architectures such as Arm, x86 and RISC-V.  
It specialises to the notions of sequential 
and parallel conjunction  straightforwardly.

\newcommand{\prefixmn}[1]{\prefixm{#1}{\Nil}}
\newcommand{\prefixscn}[1]{\prefixsc{#1}{\Nil}}

As an example, consider a memory model $\mm$ such that assignments of values to different variables can be executed in either order 
(as on Arm, but not on x86, for instance), that is, $x \asgn v \rom y \asgn w$ for $x \neq y$.
Then we have two possible terminating traces (traces ending in $\Nil$) for the program $\prefixm{x \asgn 1}{\prefixmn{y \asgn 1}}$.
\begin{gather*}
	\prefixm{x \asgn 1}{\prefixmn{y \asgn 1}}
	\ttra{x \asgn 1}
	\prefixmn{y \asgn 1}
	\ttra{y \asgn 1}
	\Nil 
	\\
	\prefixm{x \asgn 1}{\prefixmn{y \asgn 1}}
	\ttra{y \asgn 1}
	\prefixmn{x \asgn 1}
	\ttra{x \asgn 1}
	\Nil 
\end{gather*}
The first behaviour results from two applications of rule $(i)$ above (as in prefixing in CSP \cite{CSP} or CCS \cite{CCS}).
The second behaviour results from applying rule $(ii)$, noting that by assumption $x \asgn 1 \rom y \asgn 1$, and then rule $(i)$.

We define $\ffence \sdef \fencep{\fullfencespec}$, and subsequently treat it is a full fence in $\mm$ by defining
$\ffence \nrom \aca$ and $\aca \nrom \ffence$ for all $\aca$ (where $\aca \nrom \acb$ abbreviates $\neg (\aca \rom \acb)$).  
Then we have exactly one possible trace in the following circumstance.
For convenience below we omit the trailing $\Nil$.
\renewcommand{\prefixmn}[1]{#1}
\renewcommand{\prefixscn}[1]{#1}
\begin{gather}
	\prefixm{x \asgn 1}{\prefixm{\ffence}{\prefixmn{y \asgn 1}}}
	\ttra{x \asgn 1}
	\prefixm{\ffence}{\prefixmn{y \asgn 1}}
	\ttra{\ffence}
	\prefixmn{y \asgn 1}
	\ttra{y \asgn 1}
	\Nil 
\end{gather}
The fence has prevented application of the second rule (since by definition both $x \asgn 1 \nro \ffence$ and $\ffence \nro y \asgn 1$) and hence restored sequential order on the instructions.

The framework admits the definition of standard sequential and parallel composition as (extreme) memory models.
Let the ``sequential consistency'' model $\SCmm$ be the model that prevents all reordering,
and introduce a special operator for that case.
We define the complement of $\SCmm$ to be $\PARmm$, \ie, the memory model that allows all reordering, 
which corresponds to parallel execution.
\begin{eqnarray}
	\aca \roSC \acb
	\aiff
	\False \quad \mbox{for all $\aca, \acb$}
	\labeleqndefn{mm-sc}
	\\
	\aca \roPAR \acb
	\aiff
	\True \quad \mbox{for all $\aca, \acb$}
	\labeleqndefn{mm-par}
	\\
	\prefixsc{\aca}{\cmdc}
	\asdef
	\prefixM{\SCmm}{\aca}{\cmdc}
\end{eqnarray}
Using $\refeq$ for trace equality (defined formally later),
\begin{eqnarray}
	\prefixm{x \asgn 1}{\prefixm{\ffence}{\prefixmn{y \asgn 1}}}
	\arefeq
	\prefixsc{x \asgn 1}{\prefixsc{\ffence}{\prefixscn{y \asgn 1}}}
\end{eqnarray}
Without the fence $\aca$ and $\acb$ effectively execute in parallel (under $\mm$), that is,
\begin{equation}
	\prefixm{x \asgn 1}{\prefixmn{y \asgn 1}}
	\sspace
	\refeq
	\sspace
	x \asgn 1 \pl y \asgn 1
\end{equation}
Whether $\prefixm{x \asgn 1}{\prefixmn{y \asgn 1}}$ satisfies some property depends exactly on whether
the parallel execution satisfies the property.

\subsection{Reordering in \Clang}

\newcommand{\effword}{{\sf eff}}
\newcommand{\eff}[1]{\effword(#1)}
\newcommand{\effa}{\eff{\aca}}
\newcommand{\effc}{\eff{\cmdc}}
\newcommand{\State}{\Sigma}

\newcommand{\Visible}[1]{{\sf visible}~#1}
\newcommand{\SilentWord}{{\sf silent}}
\newcommand{\Silent}[1]{\SilentWord~#1}
\newcommand{\Infeasible}[1]{{\sf infeasible}~#1}
\newcommand{\Visa}{\Visible{\aca}}
\newcommand{\Silenta}{\Silent{\aca}}
\newcommand{\Infeasa}{\Infeasible{\aca}}

\newcommand{\datadep}[2]{#1 \rightsquigarrow #2}
\newcommand{\ndatadepop}{\,\,\mathop{\not\!\hspace{0.0pt}\rightsquigarrow}\,}
\newcommand{\ndatadep}[2]{#1 \ndatadepop #2}
\newcommand{\nddep}[2]{\ndatadep{#1}{#2}}
\newcommand{\nointf}[2]{interference\_free(#1, #2)}
\newcommand{\loadindep}[2]{load\_indep(#1, #2)}
\newcommand{\orderindep}[2]{order\_indep(#1, #2)}

\newcommand{\ddepab}{\datadep{\aca}{\acb}}
\newcommand{\nddepab}{\ndatadep{\aca}{\acb}}

\newcommand{\nointfab}{\nointf{\aca}{\acb}}
\newcommand{\loadindepab}{\loadindep{\aca}{\acb}}
\newcommand{\orderindepab}{\orderindep{\aca}{\acb}}

We now consider the specifics of reordering in the $\Clang$ memory model, which considers three aspects: (i) variable (data) dependencies; (ii)  fences;
and (iii) ``memory ordering constraints'', that can be used to annotate variables or fences.
We cover each of these three aspects in turn.

\subsubsection{Data dependencies/respecting sequential semantics}
\labelsect{datadeps}

A key concept underpinning both processor pipelines and compiler transformations is that
of data dependencies, where one instruction depends on a value being calculated by another.
To capture this we write $\ddepab$ if instruction $\acb$ depends on 
a value that instruction $\aca$ writes to.
We define a range of foundational syntactic and semantic concepts below.
In a concurrent setting we distinguish between local and shared variables, that is, the set $\AllVars$ is divided into two
mutually exclusive and exhaustive sets $\LocalVars$ and $\SharedVars$.
By convention we let $x, y, z$ be shared variables and $r, r_1, r_2 \ldots$ be local variables.
For convenience we introduce the syntax $\mx{s_1}{s_2}$ to mean that sets
$s_1$ and $s_2$ are mutually exclusive.
\begin{eqnarray}
	\mx{s_1}{s_2}
	\asdef
	s_1 \int s_2 = \ess
	\notag
	\\
	\AllVars \aeq \LocalVars \union \SharedVars \qquad \mbox{(with $\mx{\LocalVars}{\SharedVars}$)}
	\\
	\wv{x \asgn e} \sseq \{x\} 
	\qquad
	\wv{\guarde} \aeq \ess
	\qquad
	\quad
	\wv{\fencepf} \sseq \ess
	\labeleqn{defn-wv}
	\\
	\rv{x \asgn e} \sseq \fve
	\qquad
	\rv{\guarde} \aeq \fve
	\qquad
	\rv{\fencepf} \sseq \ess
	\labeleqn{defn-rv}
	\\
	\fva \aeq \wva \union \rva 
	\labeleqn{defn-fva}
	\\
	\sva = \fva \int\SharedVars
	\qquad
	\rsva \aeq \rva \int \SharedVars
	\qquad
	\wsva = \wva \int \SharedVars
	\labeleqn{defn-sva}
	\\
	\AllStores
	\sssdef
	\{\aca | \wsva \neq \ess\}
	&&
	\AllLoads 
	\sssdef
	\{\aca | \rsva \neq \ess\}
	\labeleqn{defn-stores}
	\\
	\Also
	\datadep{\aca}{\acb}
	\asdef
		\wv{\aca} \int \rv{\acb} \neq \ess
	\labeleqn{defn-datadep}
	\\
	\ndatadep{\aca}{\acb}
	\asdef
		\neg (\datadep{\aca}{\acb})
	\labeleqn{defn-ndatadep}
	\\
	\nointf{\aca}{\acb}
	\asdef
		\mx{\wv{\aca}}{\fv{\acb}} 
		\sspace \land \sspace 
		\mx{\wv{\acb}}{\fv{\aca}}
	\labeleqn{defn-nointf}
	\\
	\aeq
		\ndatadep{\aca}{\acb} \land \ndatadep{\acb}{\aca}
		\land \mx{\wv{\aca}}{\wvb}
	\notag
	\\
	\loadindep{\aca}{\acb}
	\asdef
		\mx{\rsv{\aca}}{\rsv{\acb}}
	\labeleqn{defn-loadindep}
	\\
	\orderindep{\aca}{\acb}
	\asdef
	\eff{\aca} \comp \eff{\acb} = \eff{\acb} \comp \eff{\aca}
	\labeleqn{defn-orderindep}
\end{eqnarray}
The write variables of instructions (written $\wva$) are collected syntactically \refeqn{defn-wv}, as are the read variables (written
$\rva$) \refeqn{defn-rv}, which depend on the usual notion of the free variables in an expression (written $\fve$, defined
straightforwardly over the syntax of expressions).  The free variables of an instruction are the union of the write and read variables
\refeqn{defn-fva}.  Shared and local variables have different requirements from a reordering perspective and so we introduce specialisations
of these concepts to just the shared variables \refeqn{defn-sva}.
We define ``store'' instructions ($\AllStores$) as those that write to a shared variable,
and ``load'' instructions ($\AllLoads$) as those that read shared variables (and hence an instruction such as $x \asgn y$
is both a store and a load).

Using these definitions we can describe various relationships between actions.
One of the key notions is that of ``data dependence'', where 
we write $\ddepab$ if 
instruction $\acb$ references a variable being modified by instruction $\aca$ \refeqn{defn-datadep}
(and similarly we write $\nddepab$ if there is no data dependence \refeqn{defn-ndatadep}).
For instance, $\datadep{x \asgn 1}{r \asgn x}$ but $\ndatadep{x \asgn 1}{r \asgn y}$.
The former can be expressed as $x \asgn 1$ ``carries a dependency into'' $r \asgn x$.
Two instructions are \emph{interference free} if there is no data dependence in either direction,
and they write to different variables \refeqn{defn-nointf}.%
\footnote{
This is a well-known property, the earliest example being Hoare's disjointness
\cite{HoareTowardsPP,HoareTowardsPP2002,50yHoare}, 
and is also called non-interference in separation logic
\cite{BrookesSepLogic07}.
This condition, formally discussed since the 1960's, is remarkably powerful for explaining the majority of observed behaviours of basic instructions types on modern hardware, although this
author did not find evidence of cross-over in older references (\cite{CDC6600,Tomasulo67}, etc).
}
Note that instructions that are interference-free may still load the same variables; we say 
they are \emph{load independent} if they access distinct \emph{shared} variables \refeqn{defn-loadindep}.
Finally, two instructions are ``order independent'' if the effect of executing them is independent of the execution order
\refeqn{defn-orderindep}.  The ``effect'' function $\effa$ denotes actions as a set of pairs of states in the usual imperative style
(defined later in \reffig{eff-semantics}), using
`$\comp$' for relational composition.
Note that order-independence is a weaker condition than non-interference,
for example, $x \asgn 1$ and $x \asgn 1$ are order-independent but not interference-free.
All of these definitions or concepts defined for instructions can be lifted straightforwardly to commands by induction on the syntax
(see \refappendix{syntax-lifting}).

The key aspect of interference-freedom is the following.
\begin{theorem}[Disjointness]
\begin{equation*}
	\nointf{\aca}{\acb}
	\imp
	\orderindepab
\end{equation*}
\end{theorem}
\begin{proof}
Straightforward: independent actions do not change each other's outcomes.
\end{proof}

We say $\acb$ may be \emph{reordered before} $\aca$ with respect to sequential semantics, written $\aca \roG \acb$, 
if they are interference-free and load-independent; the latter constraint maintains \emph{coherence} of loads.
\begin{definition}[\Gmm]
\labelmm{Gmm}
For instruction $\aca,\acb$,
\begin{equation}
	\aca \roG \acb \qquad \sdef  \qquad
		\nointfab \land
		\loadindepab 
\end{equation}
\end{definition}
A syntactic check may be used to validate $\aca \roG \acb$.  
Of course one could do a semantic check to establish the weaker property of order-independence,
however this is typically not done on a case-by-case basis but rather is used to justify
compiler transformations.
It is not feasible to check order independence directly for all cases.

\newcommand{\abut}{&\mbox{but}&}

We can derive the following, where we assume $x,y \in \SharedVars$, $r, r_1, r_2 \in \LocalVars$, 
and all are distinct, and $v,w \in \Val$.
We write $\aca \nroG \acb$ if $\neg (\aca \roG \acb)$.  
\begin{eqnarray}
	x \asgn v \roG y \asgn w
	\abut
	x \asgn v \nroG x \asgn w
	\labeleqn{c11ro-ww}
	\\
	x \asgn v \roG r \asgn y
	\abut
	x \asgn v \nroG r \asgn x
	\labeleqn{c11ro-wl}
	\\
	r_1 \asgn x \roG r_2 \asgn y
	\abut
	r_1 \asgn x \nroG r_2 \asgn x
	\labeleqn{c11ro-ll}
	\\
	x \asgn r \roG y \asgn r
	\labeleqn{c11ro-llr}
	\\
	\guard{r = v} \roG x \asgn w
	\labeleqn{c11ro-gw}
\end{eqnarray}
This shows that independent stores can be reordered, but not stores to the same variable
\refeqn{c11ro-ww}; 
similarly independent stores and loads can be reordered, as can loads, provided they are not referencing the same (shared) variable 
(\refeqns{c11ro-wl}{c11ro-ll}).
Accesses of the same local variable, however, can be reordered
\refeqn{c11ro-llr}, since no interference is possible.
Finally, stores can be reordered before (independent) guards
\refeqn{c11ro-gw}.  Since, as we show later, guards model branch points, this is a significant aspect of \Cmm.  For hardware weak memory models \refeqn{c11ro-gw} is not allowed,
as it implies a ``speculative'' write that may not be valid if $r = v$ is eventually found not to hold; however at compile-time
whether a branch condition will eventually hold may be able to be predetermined.

\subsubsection{Respecting fences}

\newcommand{\imageof}[2]{#1\limg #2 \rimg}
\newcommand{\compl}[1]{\overline{#1}}

\renewcommand{\isStore}[1]{#1 \in \AllStores}
\renewcommand{\isLoad}[1]{#1 \in \AllLoads}
\newcommand{\nisStore}[1]{#1 \notin \AllStores}
\newcommand{\nisLoad}[1]{#1 \notin \AllLoads}
\newcommand{\nisStorea}{\nisStore{\aca}}
\newcommand{\nisLoada}{\nisLoad{\aca}}

\newcommand{\getfences}[1]{|\!#1\!|}
\newcommand{\getfencesa}{\getfences{\aca}}
\newcommand{\getfencesb}{\getfences{\acb}}

Since the reorderings allowed by weak memory models may be problematic for establishing communication protocols between processes such
models typically have their own ``fence'' instruction types,
which
are \emph{artificial} constraints on reordering (as opposed to the ``natural'' data-dependence constraint).
\Clang has fences specifically to establish order between stores, between loads, or between either type.
We define $\roFNC$ below, relating fences and instructions.
\begin{eqnarray}
	\roFNC &\in& \FenceType \cross \Instr
	\qquad (\fencegen \roFNC \aca \sspace \sdef \sspace (\fencegen, \aca) \in \roFNC)
	\notag
	\\
	\storefencespec \nroFNC \aca
	\aiff
	\isStorea
	\labeleqn{defn-storebr}
	\\
	\loadfencespec \nroFNC \aca
	\aiff
	\isLoada
	\labeleqn{defn-loadbr}
	\\
	\fullfencespec \nroFNC \aca
	&&
	\mbox{for all $\aca$}
	\labeleqn{defn-scbr}
\end{eqnarray}
\refEqn{defn-storebr} states that a store fence blocks store instructions (recall \refeqn{defn-stores}),
while \refeqn{defn-loadbr} similarly states that load fences block loads.
\refEqn{defn-scbr} states that a ``full'' fence blocks all instruction types.

We use this base definition to define when two instructions can be reordered according 
to their respective fences (lifting the relation name $\roFNC$).
\begin{definition}[Fence reorderings]
\begin{eqnarray*}
	\aca 
	\roFNC 
	\acb
	\aiff
	(\forall \fencegen \in \getfencesa @ \fencegen \roFNC \acb)
	\sspace
	\land
	\sspace
	(\forall \fencegen \in \getfencesb @ \fencegen \roFNC \aca)
\end{eqnarray*}
\end{definition}
where $\getfences{\aca}$ extracts the fences present in $\aca$, \ie,
\begin{eqnarray}
	\getfences{x \asgn e} \sseq \ess
	\qquad
	\getfences{\guarde} \aeq \ess
	\qquad
	\getfences{\fencepf} \sseq \{\fencegen\}
\end{eqnarray}
Hence $\aca$ and $\acb$ can be reordered (considering fences only) 
if they each respect the others' fences.  Note that all \Clang fences are defined symmetrically so we simply check the pair
in each direction.%
\footnote{
For asymmetric fences, such as
Arm's \T{isb} instruction, both directions need to be specified \cite{pseqc-arXiv}.
}

Based on these definitions we can determine the following special cases,
where $x,y \in \SharedVars$ and $r,r_i \in \LocalVars$.
\begin{eqnarray}
	x \asgn 1
	\nroFNC
	&
	\storefence
	&
	\nroFNC
	y \asgn 1
	\\
	r_1 \asgn x
	\nroFNC 
	&
	\loadfence
	&
	\nroFNC 
	r_2 \asgn y
	\\
	x \asgn 1
	\nroFNC 
	&
	\fullfence
	&
	\nroFNC 
	r \asgn y
\end{eqnarray}
Statements of the form $\aca \nro \acb \nro \acc$ should be read as a shorthand for $\aca \nro \acb$ and $\acb \nro \acc$.
Inserting fences limits the number of possible traces of a program, an outcome which obviously may restore intended relationships between
variables.

\subsection{Memory ordering constraints}

\Clang introduces ``\emph{atomics}'' in the \T{stdatomic} library, which includes several types of ``memory ordering constraints'' (which we herewith call ordering constraints),  
which can tag loads and stores of variables declared to be
``atomic'' (\eg, type \T{atomic\_int} for a shared integer).  These ordering constraints control how atomic variables interact.
We start with an informal overview of how ordering constraints are intended to work, and then show how we incorporate them into the syntax of the language and
define a reordering relation over them.

For the discussions below we use a more compact notation for stores and loads,
as exemplified by the following \Clang equivalents.%
\footnote{
In \Cpp,
	\T{
		r = y.load(std::memory\_order\_relaxed)
	}
	and
	\T{
		y.store(3, std::memory\_order\_seq\_cst) 
	}
}
\begin{eqnarray}
	\T{
		r = atomic\_load\_explicit(\&y, memory\_order\_relaxed)
	} 
\asdef
	r \asgn \modX{y}
\\
	\T{
		atomic\_store\_explicit(\&y, 3, memory\_order\_seq\_cst)
	}
\asdef
	\modSC{y} \asgn 3
\end{eqnarray}

\subsubsection{Informal overview of ordering constraints}

We describe the six ordering constraints defined in type \T{memory_oder} below.
\begin{itemize}
\item
\emph{Relaxed} ($\relaxed$, \T{memory\_order\_relaxed}).
A \emph{relaxed} access is the simplest, not adding any extra constraints to the variable access, allowing both the hardware and compiler to potentially reorder
independent relaxed accesses.  For instance, consider the following program snippet (where `$;$' is \Clang's semicolon).
\begin{equation}
	\modxX \asgn 1 ; \modX{flag} \asgn \True
\end{equation}
If the programmer's intention is that the $flag$ variable is set to $\True$ after the data $x$ is set to 1, then they have erred: either the compiler
or the hardware may reorder the apparently independent stores to $x$ and $flag$.

\item
\emph{Release} ($\release$, \T{memory\_order\_release}).
The \emph{release} tag can be applied to stores, indicating the end of a set of shared accesses (and hence control can be ``released''
to the system).  Typically this tag is used on the setting of a flag variable, e.g., modifying the above case,
\begin{equation}
	\modxX \asgn 1 ; \modR{flag} \asgn \True
\end{equation}
now ensures that
any other
process that sees $flag = \True$ knows that the update ($\modxX \asgn 1$) preceding the change to $flag$ has taken effect.

\item
\emph{Acquire} ($\acquire$, \T{memory\_order\_acquire}).
The \emph{acquire} tag can be applied to loads, and is the reciprocal of a release: any subsequent loads will see everything the acquire
can see.  For example, continuing from above, a simple process that reads the $flag$ in parallel with the above process may be written as follows:
\begin{equation}
	f \asgn \modVar{flag}{\acquire} ; r \asgn \modxX
\end{equation}
In this case the loads are kept in order by the acquire constraint, and hence at the end of the program, in the absence of any other interference,
$f = 1 \imp r = 1$.

\item
\emph{Consume} ($\consume$, \T{memory\_order\_consume}).
The \emph{consume} tag is similar to, but weaker than, $\acquire$, in that it is intended to be partnered with a $\release$,
but
only subsequent loads that are data-dependent on the loaded value are guaranteed to see the change.  Hence,
\begin{equation}
	\modC{f \asgn flag} ; \modX{r \asgn x}
\end{equation}
does \emph{not} give $f = 1 \imp r = 1$.  However, after
\begin{equation}
	f \asgn \modC{flag} ; \modX{y \asgn f}
\end{equation}
then $f = 1 \imp y = 1$.   Data-dependence is maintained by the $\roG$ relation, and in that sense the $\consume$ constraint provides no extra reordering information
to that of $\relaxed$.  However
the intention is that a $\consume$ load indicates to the compiler that it must not lose, via optimisations, data dependencies to later instructions.
We return to $\consume$ in \refsect{consume}, in the context of expression optimisations,
but for concision in the rest of the paper we omit consideration of $\consume$,
which otherwise behaves as a $\relaxed$ constraint.

\item
\emph{Sequentially consistent} ($\seqcst$, \T{memory\_order\_seq\_cst}).
The \emph{sequentially consistent} constraint is the strongest, forcing order between $\seqcst$-tagged instructions and any other instructions.
For example, the snippet
\begin{equation}
	\modSC{x} \asgn 1; \modSC{flag} \asgn \True
\end{equation}
ensures $flag = 1 \imp x = 1$ (in fact only one instruction needs the constraint for this to hold).  This is considered a more ``heavyweight'' method for enforcing order than the acquire/release constraints.

\item
\emph{Acquire-release} ($\acqrel$, \T{memory\_order\_acq\_rel}).
This constraint is used on instructions that have both an acquire and release component, and as will be seen it is straightforward to combine them in our syntax.

\end{itemize}

\paragraph{Non-atomics}
Additionally shared data can be ``non-atomic'', \ie, variables that are not declared of type \T{atomic_*}, 
and as such cannot be directly associated with the ordering constraints above.
Programs which attempt to concurrently access shared data without any of the ordering mechanisms of \T{atomics.h}
essentially have no guarantees about the behaviour, and so we ignore such programs.
For programs that include shared non-atomic variables which are correctly synchronised, \ie, are associated with some
ordering mechanism (\eg, barriers or acquire/release flag variables) they can be treated as if relaxed, and are subject to 
potential optimisations as outlined in \refsect{expression-optimisations}.
Distinguishing between correctly synchronised shared non-atomic variables is a syntactic check that compilers carry out,
which could be carried over into our syntactic framework, but this is not directly relevant to the question of reorderings.

\subsubsection{Formalising memory order constraints}

The relationship between the ordering constraints can be thought of in terms of a reordering relationship, $\roOC$, as in the previous sections.
\OMIT{
There is an ordering on the strength of memory order induced by the constraints, namely
\begin{equation}
	\relaxed < \{\release, \acquire, \consume\} < \acqrel < \seqcst
\end{equation}
However, this is not all that helpful.
}
\begin{definition}[Memory ordering constraints]
\labeldefn{roOC}
\begin{eqnarray*}
	\roOC \asdef \{(\relaxed, \relaxed), (\relaxed, \acquire), (\release,\relaxed), (\release, \acquire)\}
\end{eqnarray*}
\end{definition}
We write $\oc_1 \roOC \oc_2$ if $(\oc_1, \oc_2) \in \roOC$, and as before $\oc_1 \nroOC \oc_2$ otherwise.
Expressing the relation as a negative is perhaps more intuitive, \ie, for all constraints $\oc$,
$\oc \nroOC \seqcst \nroOC \oc$, $\oc \nroOC \release$, and $\acquire \nroOC \oc$.
Additionally, the $\consume$ constraint is equal to a $\relaxed$ constraint for the purposes of ordering (but see \refsect{consume}), 
and $\acqrel$ is the combination of both $\acquire$ and $\release$.
An alternative presentation of the relationship of 
$\oc_1 \roOC \oc_2$
is as a grid, \ie,
\begin{equation}
\begin{array}{lcccc}
\downarrow\!\oc_1 \quad \oc_2\!\fun
		 & \relaxed & \release & \acquire & \seqcst  \\
\relaxed & \tick    & \ntick   & \tick    & \ntick  \\
\release & \tick   & \ntick   & \tick   & \ntick  \\
\acquire & \ntick   & \ntick   & \ntick   & \ntick  \\
\seqcst  & \ntick   & \ntick   & \ntick   & \ntick 
\end{array}
\end{equation}
Note that acquire loads may come before release stores, 
that is, \Clang follows the \RCmm model rather than the
the stronger \RCSCmm model in \cite{ReleaseConsistency90} (it is of course straightforward to accommodate either; see also \cite{pseqc-arXiv}).

We extend the syntax of instructions to incorporate ordering constraints, allowing
variables and fences
to be annotated with a set thereof,
as shown in \reffig{c11-ro}.
The reordering relation for instructions, considering only ordering constraints, can be defined as below.
\begin{definition}{Reordering instructions with respect to ordering constraints}
\labelmm{roOCs}
\begin{eqnarray}
	\aca \roOCs \acb 
	&\sdef&
	\getocs{\aca} \cross \getocs{\acb} \subseteq \roOC
\end{eqnarray}
\end{definition}
where the $\getocs{.}$ function extracts the ordering constraints from the expression and instructions syntax.
\begin{gather}
\notag
	\getocs{v} \sseq \ess 
	\qquad
	\getocs{\modxocs} \sseq \ocs 
	\\
	\labeleqndefn{getocs}
	\getocs{\binop e} \sseq \getocs{e} 
	\qquad
	\getocs{e_1 \binop e_2} \sseq \getocs{e_1}  \union \getocs{e_2} \\
\notag
	\getocs{\modxocs \asgn e} \sseq \ocs \union \getocs{e}
	\qquad
	\getocs{\guarde} \sseq \getocs{e}
	\qquad
	\getocs{\modAct{\fencepf}{\ocs}} \sseq \ocs
\end{gather}
For example, $\getocs{\modxR \asgn \modyA} = \{\release, \acquire\}$.
Thus $\aca \roOCs \acb$ is checked by comparing point-wise each pair of ordering constraints in $\aca$ and $\acb$.
If $\aca$ or $\acb$ has no ordering constraints then $\aca \roOCs \acb$ vacuously holds.

For example 
\begin{eqnarray*}
	\modxR \asgn \modyA \roOCs \modxX \asgn 1
	\aiff
	\getocs{\modxR \asgn \modyA} \cross \getocs{\modxX \asgn 1} \subseteq \roOC
	\\
	\aiff
	\{\release, \acquire \} \cross \{\relaxed\} \subseteq \roOC
	\\
	\aiff
	\release \roOC \relaxed \land \acquire \roOC \relaxed
	\\
	\aiff
	\False
	\quad
	\mbox{by \refdefn{roOC}}
\end{eqnarray*}

Note that a locals-only instructions, such as $r_1 \asgn r_2 * 2$, does not have any ordering constraints, and so is not affected by them,
that is $\getocs{r_1 \asgn r_2 * 2} = \ess$ and hence $r_1 \asgn r_2 * 2 \roOCs \acb$ for all $\acb$.

\begin{figure}
$x \in \AllVars 
\quad \ocs \in \power \OCtype$
\begin{eqnarray*}
	\orderingConstraint &\ttdef& 
		\relaxed \csep 
		\release \csep 
		\acquire \csep 
		\consume \csep 
		\seqcst
	\notag
	\\
	\acqrel \asdef \{\acquire, \release\}
	\\
	e \attdef v \csep \modxocs \csep \unop e \csep e_1 \binop e_2 
	\\
    \aca 
    \attdef
        \modxocs \asgn e
        \csep
        \guarde
        \csep
        \modAct{\fencepf}{\ocs}
\\
	\relfence
	\asdef
	\modR{\storefence }
	\\
	\acqfence
	\asdef
	\modA{\loadfence }
	\\
	\scfence
	\asdef
	\modSC{\fullfence }
	\\
	\Also
	\mbox{for $r \in \LocalVars$,}
	&&
	\mbox{$r$ abbreviates $\modVar{r}{\ess}$}
	\\
	\mbox{for $x \in \SharedVars$,}
	&&
	\mbox{$\modxoc$ abbreviates $\modx{\{\oc\}}$, and}
	\\
	&&
	\mbox{$x$ abbreviates $\modxX$}
	\\
	\mbox{for $\fencepf$ a fence,}
	&&
	\mbox{
	$ \modAct{\fencepf}{\oc} $
	abbreviates
	$ \modAct{\fencepf}{\{\oc\}} $
	}
\end{eqnarray*}

\Description{TODO}
\caption{
Syntax extensions for \Clang ordering constraints 
}
\labelfig{c11-ro}
\end{figure}

For convenience we define some abbreviations and conventions at the bottom of \reffig{c11-ro}:
we require every reference to a shared variable to have (at least one) ordering constraint, with the default being $\relaxed$;
when there is exactly one ordering constraint we omit the set comprehension brackets in the syntax, and typically, when the types are clear,
we abbreviate $\modxX$ to plain $x$ as $\relaxed$ is the default.
Local variables are by definition never declared ``atomic'' and hence their set of ordering constraints is always empty; hence, when
$r$ is a local variable, we abbreviate $\modVar{r}{\ess}$ to $r$.
Similarly we abbreviate ordering constraints on fences.
We can now define release, acquire and sequentially consistent fences as the combination of a fence and ordering constraint.
A ``release fence'' ($\relfence$) operates according to the $\release$ semantics above, and in addition blocks stores, hence is a combination of a $\storefence$ and $\release$ ordering,
and similarly an ``acquire fence'' ($\acqfence$) acts as a $\loadfence$ and $\acquire$ ordering.  
A ``sequentially consisitent'' fence ($\scfence$) is also defined straightforwardly; these fences map to \Clang's \T{atomic_thread_fence(...)} definition.


To complete the syntax extension
we update the syntax-based definitions for extracting variables from expressions and instructions by defining
$\wv{\modxocs \asgn e} = \{x\}$ and $\rv{\modxocs} = \{x\}$, that is, read/write variables do not include the ordering constraints
(and $\wv{\modocs{\fencepf}} = \rv{\modocs{\fencepf}} = \ess$).

\subsection{The complete reordering relation}

We can now define the \Clang memory model as the combination of the three aspects above.
\begin{definition}[Reordering of instructions in \Cmm]
\labelmm{Cmm}
\begin{eqnarray}
	\aca \ro \acb 
	&iff&
	\mbox{
		(i)~
		$\aca \roG \acb$,
		~~~~(ii)~
		$\aca \roFNC \acb$,
		~~~~and
		(iii)~
		$\aca \roOCs \acb$
	}
\end{eqnarray}
\end{definition}
Hence reordering of instructions within a \Clang program can occur provided the sequential semantics, fences, and ordering constraints are respected.
We show in later sections how this principle does not change for more complex language features, though, of course, 
the semantics and hence the analysis is correspondingly more complex.

As examples, for distinct $x,y \in \SharedVars$ and $r,r_i \in \LocalVars$,
\begin{eqnarray*}
	x \asgn 1 \ro y \asgn 1 
	&&
	r_1 \asgn x \ro r_2 \asgn y
	\labeleqn{rlx-wr-ro} 
	\\
	y \asgn 1 \nro \modxR \asgn 1 
	&but&
	\modxR \asgn 1 \ro r \asgn y
	\labeleqn{rel-w-ro} 
	\\
	r_1 \asgn \modxA \nro r_2 \asgn y 
	&but&
	r_1 \asgn y \ro r_2 \asgn \modxA 
	\labeleqn{acq-r-ro} 
	\\
	\aca \nro 
	\relfence
	& &
	\relfence
	\ro r \asgn x
	\\
	\acqfence
	\nro \aca
	&&
	x \asgn 1
	\ro
	\acqfence
\end{eqnarray*}

Hence, following the earlier definitions, we have various ways of enforcing program order
using the flag set/check (message passing) pattern from earlier.
We leave $\relaxed$ accesses of shared variables $x$ and $flag$ implicit.
\begin{equation} \begin{array}{ccc}
	\prefixC{x \asgn 1}{\prefixC{\relfence}{flag \asgn \True}}
	\arefeq
	\prefixsc{x \asgn 1}{\prefixsc{\relfence}{flag \asgn \True}}
	\\
	\prefixC{f \asgn flag}{\prefixC{\acqfence}{r \asgn x}}
	\arefeq
	\prefixsc{f \asgn flag}{\prefixsc{\acqfence}{r \asgn x}}
	\\
	\prefixC{x \asgn 1}{\modR{flag} \asgn \True}
	\arefeq
	\prefixsc{x \asgn 1}{\modR{flag} \asgn \True}
	\\
	\prefixC{f \asgn \modA{flag}}{r \asgn x}
    \arefeq
	\prefixsc{f \asgn \modA{flag}}{r \asgn x}
\end{array} \end{equation}

\section{An imperative language with reordering}
\labelsect{impro}

We now show how reordering according to the \Clang memory model
can be embedded into a more realistic imperative language that has conditionals and loops,
based on the previously described wide-spectrum language \impro \cite{SEFM21,FM18,pseqc-arXiv}.  
We give a small-step operational semantics and define trace equivalence for its notion of correctness.

\subsection{Syntax}
\newcommand{\CAS}[3]{\mathbf{cas}(#1,#2,#3)}
\newcommand{\CASx}[2]{\CAS{x}{#1}{#2}}
\newcommand{\CASxe}{\CASx{e}{e'}}
\newcommand{\CASxv}{\CASx{v}{v'}}
\newcommand{\FAA}[2]{\mathbf{faa}(#1,#2)}
\newcommand{\FAAx}[1]{\FAA{x}{#1}}

\newcommand{\IFM}[4]{\If^{#1} #2 \Then #3 \Else #4}
\newcommand{\IFm}[3]{\IFM{\mm}{#1}{#2}{#3}}
\newcommand{\IFC}[2]{\IFM{\Cmm}{#1}{#2}}
\newcommand{\IFbcm}{\IFm{b}{c_1}{c_2}}
\newcommand{\IFbcM}[1]{\IFM{#1}{b}{c_1}{c_2}}
\newcommand{\IFbcC}{\IFbcM{\Cmm}}
\newcommand{\IFbcSC}{\IFbcM{\SCmm}}

\newcommand{\IFsM}[3]{\If^{#1} #2 \Then #3}
\newcommand{\IFsm}[2]{\IFsM{\mm}{#1}{#2}}
\newcommand{\IFsC}[2]{\IFsM{\Cmm}{#1}{#2}}
\newcommand{\IFsSC}[2]{\IFsM{\SCmm}{#1}{#2}}
\newcommand{\IFsbcm}{\IFsm{b}{c}}
\newcommand{\IFsbcM}[1]{\IFsM{#1}{b}{c}}
\newcommand{\IFsbcC}{\IFsbcM{\Cmm}}
\newcommand{\IFsbcSC}{\IFsbcM{\SCmm}}

\newcommand{\WHM}[3]{\While^{#1} #2 \Do #3}
\newcommand{\WHm}[2]{\WHM{\mm}{#1}{#2}}
\newcommand{\WHbcm}{\WHm{b}{\cmdc}}
\newcommand{\WHbcM}[1]{\WHM{#1}{b}{\cmdc}}

\begin{figure}
\begin{eqnarray}
    \cmdc 
    \attdef
        \Skip 
        \cbar
        \vec{\aca}
        \cbar
        \cmdc_1 \ppseqm \cmdc_2
        \cbar
        \cmdc_1 \choice \cmdc_2
        \cbar
        \iteratecm
    \labeleqndefn{cmd}
\\
    \also
    \tau &\sdef& \guard{\True}
    \labeleqn{defn-tau}
    \\
    \cmdc_1 \ppseqc \cmdc_2
    \asdef
    \cmdc_1 \ppseq{\Cmm} \cmdc_2
    \labeleqn{defn-ppseqc}
    \\
    \cmdc_1 \bef \cmdc_2
    \asdef
    \cmdc_1 \ppseq{\SCmm} \cmdc_2
    \labeleqn{defn-ppseqs}
    \\
    \cmdc_1 \pl \cmdc_2
    \asdef
    \cmdc_1 \ppseq{\PARmm} \cmdc_2
    \labeleqn{defn-pl}
    \\
    \finiteratecm{0} \sdef \Nil
    \sspace&&\sspace
    \finiteratecm{n+1} \sdef c \ppseqm \finiteratecmn
    \labeleqndefn{finiter}
    \\
    \also
    \IFbcm
    \asdef
    \guard{b} \ppseqm \cmdc_1 \ \ \choice \ \  \guard{\neg b} \ppseqm \cmdc_2
    \labeleqndefn{if}
    \\
    \WHbcm
    \asdef
    \iterate{(\guardb \ppseqm \cmdc)}{\mm} \ \ppseqm \ \guard{\neg b}
    \labeleqn{defn-while}
	\\
	\CASxe
	\asdef
	\seqT{\guard{x = e} \asep x \asgn e'} \choice \guard{x \neq e}
    \labeleqndefn{cas}
	\\
	\modAR{\CASxe}
	\asdef
	\seqT{\guard{\modxA = e} \asep \modxR \asgn e'}  \choice \guard{\modxA \neq e}
    \labeleqn{defn-casAR}
	\\
	r \asgn \CASxe
	\asdef
	(\seqT{\guard{x = e} \asep x \asgn e'} \scomp r \asgn \True) \choice (\guard{x \neq e} \scomp r \asgn \False)
    \labeleqn{defn-r:=cas}
	\\
	r \asgn \FAAx{e}
	\asdef
	\seqT{r \asgn e \asep x \asgn x + e} 
    \labeleqn{defn-faa}
\end{eqnarray}
\caption{Syntax of \improc (building on \reffig{c11-ro})}
\labelfig{syntax}
\Description{TODO}
\end{figure}

The syntax of \improc is given in \reffig{syntax}, with
expressions and instructions remaining as shown in \reffig{c11-ro}.

\paragraph{Commands}
The command syntax \prefeqndefn{cmd} includes the terminated command, $\Nil$, 
a sequence of instructions, $\vecaca$ (allowing composite actions to be defined),
the parallelized sequential composition of two commands according to some memory model $\mm$,
$c_1 \ppseqm c_2$, 
a choice between two commands,
$c_1 \choice c_2$, 
or a \emph{parallelized iteration} of a command according to some memory model $\mm$,
$\iteratecm$.
From this language we can build an imperative language with conditionals and loops following algebraic patterns \cite{FischerLadner79,KozenKAT}.
We define the special action type $\tau$ as a $\True$ guard \refeqn{defn-tau}; this action has no observable effect and is not considered for the purposes of 
determining (trace) equivalence.
We allow the abbreviation $c_1 \ppseqc c_2$ for the case where the model parameter is \Cmm
\refeqn{defn-ppseqc}.
We also introduce abbreviations for strict (traditional) order \refeqn{defn-ppseqs}
and parallel composition
\refeqn{defn-pl},
based on the memory models $\SCmm$ and $\PARmm$ introduced earlier (\refeqndefns{mm-sc}{mm-par}).
The (parallelized) iteration of a command a finite number of times is defined inductively in \refeqndefn{finiter}.
Conditionals \refeqndefn{if} and while loops \refeqn{defn-while} can be constructed
in the usual way using guards and iteration.
We define an empty $\False$ branch conditional,
$ \IFsbcm$,
as 
$
	\IFM{\mm}{b}{\cmdc}{\Nil}
$.

\paragraph{Guards}
The use of a guard action type $\guarde$ allows established encodings of conditionals and loops as described above.  However
treating a guard as a separate action is useful in considering reorderings as well, and in particular in understanding
the interaction of conditionals with compiler optimisations: the fundamentals of reorderings involving guards are based on the principles
of preserving sequential semantics (on a single thread) as in \refsect{datadeps}, and these can be lifted to the conditional and loop command types
straightforwardly, without needing to deal with them monolithically.
Note that if a guard evaluates to $\False$ this represents a behaviour that cannot occur.

\paragraph{Composite actions}
Since we allow sequences of instructions $\vecaca$ to be the basic building block of the language,
with the intention that all instructions in the sequence are executed, in order, as a single indivisible step,
we can straightforwardly define complex instructions types such as ``compare-and-swap''.
For lists we write $\eseq$ for the empty list, $\cat$ for concatenation, and $\seqT{\ldots}$ as the list constructor.  
For notational ease we let a singleton sequence of actions $\seqT{\aca}$ just be written $\aca$ where the intended type is clear from the context.
For brevity we allow $\cat$ to accept single elements in place of singleton lists.
Note that
an instruction such as $x \asgn y$ happens as a single indivisible step, that is, the value for $y$ is fetched and written into $x$ in one step.
This is not realistic for \Clang, and as such we later (\refsect{incremental-evaluation}, see also \refappendix{indivis-vecaca}) 
show how instructions can be incrementally executed (in the above case, with the value for $x$ fetched, and only later
updating $y$ to that value).  For now we assume that anything evaluated incrementally is written out to make the granularity explicit, for example,
the above assignment becomes $tmp \asgn y ; x \asgn tmp$, for some fresh identifier $tmp$; see further discussion in \refsect{incremental-reasoning}.

The composite compare-and-swap command $\CASxe$ 
is defined as a choice between determining that $x = e$ and updating $x$ to $e'$
in a single indivisible step, or determining that $x \neq e$ \prefeqndefn{cas}.  This can be generalised to include ordering constraints, for instance,
in the $\acqrel$ case \refeqn{defn-casAR}.
Updates to a local variable can be included to record the result \refeqn{defn-r:=cas}.
It is of course straightforward to define other composite commands, such as fetch-and-add \refeqn{defn-faa}, which map to \Clang's inbuilt functions 
such as \atmcmpexchstrexpl{\ldots} and \atmfetchaddexpl{\ldots}.
We show these to emphasise that the definition of reordering does not have to change whenever a new instruction type is added; we can easily 
syntactically extract the required elements.
For instance, given the above definitions, we can determine the following.
\begin{equation}
	\CASxv \ro r \asgn y
	\quad \mbox{but} \quad
	\CASxv \nro r \asgn x
	\quad \mbox{and} \quad
	\modAR{\CASxv} \nro r \asgn y
\end{equation}
since $\wv{\CASxv} = \{x\}$, and $\acquire \in \getocs{\modAR{\CASxv}}$.
We lift the write/read variables of instructions to commands straightforwardly (see \refappendix{syntax-lifting}),
and as such
reordering on commands can be calculated, for example,
\begin{equation}
\labeleqn{c11ro-if-z}
	(\IFC{r > 0}{x \asgn 1}{y \asgn 1}) \ro z \asgn 1
\end{equation}
since the assignment to $z$ is not dependent on anything in the conditional statement.  

\subsection{Small-step operational semantics}
\labelsect{semantics}

\begin{figure}

\def\colwidthC{0.26\textwidth}
\def\colwidthA{0.35\textwidth}
\def\colwidthB{0.873\textwidth}
\ruledefNamed{\colwidthC}{Action}
{action-list}{
	\vec{\aca} \tra{\vec{\aca}} \Skip
}
\ruledefNamed{\colwidthA}{Choice}
{choice}{
	\cmdc \choice \cmdd \tra{\tau} \cmdc
	\qquad
	\cmdc \choice \cmdd \tra{\tau} \cmdd
}
\ruledefNamed{\colwidthC}{Iterate}
{iterate}{
	\iteratecm
	\tra{\tau} 
	\finiteratecmn
}

\ruledefNamed{0.91\textwidth}{\Pseqc}
{pseqcA}{
    \Rule{
        c_1 \tra{\aca} c_1'
    }{
        c_1 \ppseqm c_2 \tra{\aca} c_1' \ppseqm c_2
    }
\qquad
		\Nil \ppseqm c_2 \tra{\tau} c_2 
\qquad
%
%
	    \Rule{
        c_2 \tra{\acb} c_2'
        \qquad
        c_1 \rom \acb
    }{
        c_1 \ppseqm c_2 \tra{\acb} c_1 \ppseqm c_2'
    }
}

\caption{Operational semantics}
\labelfig{semantics}
\Description{TODO}
\end{figure}

The semantics of \impro is given in \reffig{semantics} (an adaptation of \cite{pseqc-arXiv}).
A step (action) is a sequence of instructions which are considered to be executed together, without interference;
in the majority of cases the sequences (actions) are singleton.
\refrule{action-list} places a list of instructions into the trace as a single list (so a trace is a list of lists).
\refrule{choice} states that a nondeterministic choice is resolved silently to either branch (in the rules we let $\tau$
abbreviate the action $\seqT{\tau}$; alternatively we could define $\tau$ as the empty sequence).
\refrule{iterate} nondeterministically picks a finite number ($n$) of times to iterate $\cmdc$,
where finite iteration is defined in \refeqndefn{finiter}.

\refrule{pseqcA} is the interesting rule, which generalises the earlier rule for prefixing:
a command $\cmdc_1 \ppseqm \cmdc_2$ can take a step of $c_1$, or begin execution of $c_2$ if $c_1$ is terminated,
or execute a step $\acb$ of $c_2$ if $\acb$ reorders with $c_1$.
Reordering of an action (list of instructions) with a command is lifted from reordering on instructions straightforwardly
(\refappendix{syntax-lifting}).

As an example, from \refeqn{c11ro-if-z} we can deduce the following.
\begin{equation*}
	(\IFC{r > 0}{x \asgn 1}{y \asgn 1}) \ppseqc z \asgn 1
	\ttra{z \asgn 1}
	(\IFC{r > 0}{x \asgn 1}{y \asgn 1}) \ppseqc \Nil
\end{equation*}
That is, the assignment to $z$ can occur before the conditional; this represents the compiler deciding to move
the store before the test since it will happen on either path.

\subsection{Trace semantics}
\labelsect{trace-semantics}

Given a program $\cmdc_0$ the operational semantics generates a \emph{trace}, that is, a finite sequence of steps
$\cmdc_0 \tra{\aca_1} \cmdc_1 \tra{\aca_2} \ldots$ where the labels in the trace are actions%
\footnote{
Since infinite traces do not add anything of
special interest to the discussion of weak memory models over and above finite traces,
we focus on finite traces only to avoid the usual extra complications that infinite traces introduce.
}.
We write $c \xtra{t} c'$ to say that $c$ executes the actions in trace $t$ and evolves to $c'$,
inductively constructed below.
The base case for the induction is given by $c \xtra{\eseq} c$.
\begin{eqnarray}
    c \tra{\aca} c' \land \Visa \land c' \xtra{t} c'' &\imp& c \xtra{\aca \cat t} c''
    \labeleqn{ptrace-visible}
    \\
    c \tra{\aca} c' \land \Silenta \land c' \xtra{t} c'' &\imp& c \xtra{t} c''
    \labeleqn{ptrace-silent}
    \\
    \also
    \Meaning{c} ~~\sdef~~  \{t | c \xtra{t} \Nil \}
    \labeleqn{def-meaning}
    \quad\qquad
    c \refsto d &\sdef & \Meaning{d} \subseteq \Meaning{c}
    \labeleqndefn{refsto}
    \quad\qquad
    c \refeq d ~~\sdef~~  c \refsto d \land d \refsto c
    \labeleqn{def-refeq}
\end{eqnarray}
Traces of \emph{visible} actions are accumulated into the trace (using `$\cat$' for list concatenation) \refeqn{ptrace-visible},
and \emph{silent} actions (such as $\tau$) are discarded \refeqn{ptrace-silent}, \ie, we have a ``weak'' notion of equivalence \cite{CCS}.
A visible action is any action with a visible effect, for instance, fences, assignments, and guards with free variables.
Silent actions include any guard which is $\True$ in any state \emph{and} contains no free variables; for instance, $\guard{0 = 0}$ is silent while
$\guard{x = x}$ is not.
A third category of actions, $\Infeasa$, includes exactly those guards $\guard{b}$ where
$b$ evaluates to $\False$ in every state.  This includes actions such as $\guard{x \neq x}$, with the simplest example being $\guard{\False}$, which we abbreviate
to $\Magic$ \cite{Morgan:94}.  Any behaviour of $c$ in which an infeasible action occurs does not result in a finite terminating trace, and hence is excluded from
consideration.  Such behaviours include those where a branch is taken that eventually evaluates to $\False$.

The meaning of a command $\cmdc$ is its set of all possible terminating behaviours $\Meaning{c}$, leading to the usual (reverse) subset inclusion notion of refinement, 
where $c \refsto d$ if every behaviour of $d$ is a behaviour of $c$; our notion of command \emph{equivalence} is refinement in both directions \refeqn{def-meaning}.

From the semantics we can derive the usual properties such as
$\Meaning{c \choice d} = \Meaning{c} \union \Meaning{d}$
and
$\Meaning{c \scomp d} = \Meaning{c} \cat \Meaning{d}$ (overloading `$\cat$' to mean pairwise concatenation of lists).
We can use trace equivalence to define a set of rules for manipulating a program under refinement or equivalence;
we elucidate a general set of these in the following section.

\section{Reduction rules}
\labelsect{reduction}

Using the notion of trace equivalence
the following properties can be derived for the language, and verified in Isabelle/HOL \cite{pseqc-arXiv}.
The usual properties of commutativity, associativity, \etc, for the standard operators of the language hold, and so we focus below on 
properties involving \pseqc.
\begin{eqnarray}
    \cmdc_1 \ppseqm \cmdc_2
    &\refsto&
    \cmdc_1 \bef \cmdc_2
    \labellaw{keep-order}
    \\
    c \choice d
    \ssrefsto
    c
	&&
    c \choice d
    \ssrefsto
    d
    \labellaw{chooseL}
\\
    (\aca \bef \cmdc) \pl \cmdd
    \sspace \refsto \sspace
    \aca \bef (\cmdc \pl \cmdd)
	&&
    \cmdc \pl (\acb \bef \cmdd)
    \sspace \refsto \sspace
    \acb \bef (\cmdc \pl \cmdd)
    \labellaw{fix-interleaving}
    \\
    (c_1 \choice c_2) \pl d 
    &\refeq&
    (c_1 \pl d) \choice (c_2 \pl d )
    \labellaw{dist-choice-pl}
	\\
    (c_1 \ppseqm c_2) \ppseqm c_3
    &\refeq&
    c_1 \ppseqm (c_2 \ppseqm c_3)
    \labellaw{pseqc-assoc}
    \\
    \aca \nro \acb
    \entails \sspace \sspace
    \aca \ppseqc \acb 
    &\refeq&
    \aca \bef \acb
    \labellaw{2actions-keep-order}
    \\
    \aca \ro \acb
    \entails \sspace \sspace
    \aca \ppseqc \acb 
    & \refeq &
    \aca \pl \acb 
\labellaw{2actions-reduce}
    \\
    \aca \ro \acb
    \entails \sspace \sspace
    \aca \ppseqc (\acb \bef \cmdc)
    & \refsto &
    \acb \bef (\aca \ppseqc \cmdc)
	\labellaw{reorder-action}
\end{eqnarray}

\reflaw{keep-order} states that a \pseqc can always be refined to a strict ordering.
\reflaw{chooseL} states that a choice can be refined to its left branch (a symmetric rule holds for the right branch).
\reflaw{fix-interleaving} says that the first instruction of either process in a parallel composition can be the first step of the composition as a whole.
Such refinement rules are useful for elucidating specific reordering and interleavings of parallel processes that lead to particular behaviours
(essentially reducing to a particular trace).
\reflaw{dist-choice-pl} is an equality, which states that if there is nondeterminism in a parallel process the effect can be understood by lifting the nondeterminism
to the top level; such a rule is useful for the application of, for instance, Owicki-Gries reasoning (\refsect{owicki-gries}).
\reflaw{pseqc-assoc} states that \pseqc is associative (provided the same model $\mm$ is used on both instances).
\reflaws{2actions-keep-order}{2actions-reduce} are special cases 
where, given two actions $\aca$ and $\acb$, if they cannot reorder then they are executed in order, 
and if they can it is as if they are executed in parallel.
\reflaw{reorder-action} straightforwardly promotes action $\acb$ of $\acb \bef \cmdc$ before $\aca$, and depending on the structure of $\cmdc$, further of its actions may be reordered before $\aca$.

\renewcommand{\svb}{\sv{b}}

We can extend these rules to more complex structures.
\begin{eqnarray}
	c_1 \ppseqc \scfence \ppseqc c_2
	\arefeq
	c_1 \bef \scfence \bef c_2
	\labellaw{reduce-scfence}
	\\
	\Also
	\IFsm{y \geq 0}{x \asgn y}
	\arefeq
	\IFsM{\SCmm}{y \geq 0}{x \asgn y}
	\labellaw{reduce-if-sc}
	\\
	b \notin \fve \union \fvf
	\impS
	\qquad
	\notag
	\\
	\IFM{\Cmm}{b}{x \asgn e}{y \asgn f}
	\arefeq
	(\guardb \pl x \asgn e) 
		\choice
	(\guardnb \pl y \asgn f) 
	\labellaw{reduce-if-C-pl}
\end{eqnarray}
\reflaw{reduce-scfence} shows that (full) fences enforce ordering.
\reflaw{reduce-if-sc} gives a special case of a conditional where the $\True$ branch depends on a shared variable in the condition,
in which case the command is executed in-order (assuming model $\mm$ respects data dependencies).
\reflaw{reduce-if-C-pl} elucidates the potentially complex case of reasoning about conditionals in which there are no dependencies: 
theoretically the compiler could allow inner instructions to appear to be executed before the evaluation of the condition.

Assuming $x \in \SharedVars$, $b \in \LocalVars$, and that $x$ is independent of commands $c_1$ and $c_2$ (\ie,
$
	x \notin \wvc \union \wvd
$), we can derive the following.
\begin{eqnarray}
	(\IFbcC) \ppseqc x \asgn v
	\arefeq
	(\IFbcC) \pl x \asgn v
	\\
	\IFC{b}{(c_1 \ppseqc x \asgn v)}{(c_2 \ppseqc x \asgn v)}
	\arefeq
	(\IFbcC) \pl x \asgn v
\end{eqnarray}
These sort of structural rules help elucidate consequences of the memory model for programmers at a level that is easily understood.

The next few laws allow reasoning about ``await''-style loops, as used in some lock implementations.
\begin{eqnarray}
	\iterate{\aca}{\mm}
	\arefeq
	\iterate{\aca}{\SCmm}
	\labellaw{iterate-one-action}
	\\
	\WHM{\Cmm}{b}{\Nil}
	\arefeq
	\iterate{\guardb}{\Cmm} \ppseqc \guard{\neg b}
	\labellaw{empty-while}
	\\
	\svb \neq \ess
	\impS
	\WHM{\Cmm}{b}{\Nil}
	\arefeq
	\WHM{\SCmm}{b}{\Nil}
	\labellaw{reduce-empty-while}
\end{eqnarray}
\reflaw{iterate-one-action} states that a sequence of repeated instructions in order, according to any model $\mm$,
can be treated as executing in order.
Using this property, and others like
$
	\Nil \ppseqm c 
	\refeq
	c
$
(we omit such trivial laws that do not involve reordering)
we can deduce \reflaw{reduce-empty-while}, that states that a spin-loop that polls a shared variable
can be treated as if executed in strict order .

\paragraph{Lifting to commands}

So far we have considered reduction laws that apply to relatively simple cases involving individual actions.
Lifting the concepts to commands is nontrivial in general, that is,
for $c_1 \ppseqc c_2$ there could be arbitrary dependencies between $c_1$ and $c_2$ which mean they partially overlap 
perhaps with pre- or post-sequences of non-overlapping instructions.
Here associativity (\reflaw{pseqc-assoc}) may help, allowing rearrangement of the text to split into sections, for instance,
\begin{equation}
	(c_1 \ppseqc \scfence \ppseqc c_2) \ppseqc (c_3 \ppseqc \scfence \ppseqc c_4)
	\ssrefeq
	c_1 \bef \scfence \bef (c_2 \ppseqc c_3) \bef \scfence \bef c_4
\end{equation}

\newcommand{\cmdmixsym}{\langle\!\pl}
\newcommand{\cmdmix}[3]{#1 \overset{#3}{\cmdmixsym} #2}
\newcommand{\cmdmixm}[2]{\cmdmix{#1}{#2}{\mm}}
\newcommand{\cmdmixc}[2]{\cmdmix{#1}{#2}{\Cmm}}
\newcommand{\cmdmixn}[2]{\cmdmix{#1}{#2}{}}

\newcommand{\allowall}[2]{\cmdmixn{#1}{#2}}
\newcommand{\allowallcd}{\allowall{c}{d}}

\newcommand{\nreducesToNil}[1]{#1 \not\refeq \Nil}
\newcommand{\nreducesToNilc}{\nreducesToNil{\cmdc}}

\newcommand{\blocknext}[2]{#1 \mathrel{\overlay{$\,\cross$}{$\leftarrow$}} #2}
\newcommand{\blocknextc}[1]{\blocknext{\cmdc}{#1}}
\newcommand{\blocknextcd}{\blocknextc{\cmdd}}

\newcommand{\blockall}[2]{#1 \red{\mathrel{\overlay{$\cross$}{\kern 0.06em $\pl$}}} #2}
\newcommand{\blockallc}[1]{\blockall{\cmdc}{#1}}
\newcommand{\blockallcd}{\blockallc{\cmdd}}

We now consider the cases where two commands are completely independent, and where one is always blocked by the other.
Independence can be established by straightforwardly lifting $\ro$ from instructions to commands, using 
lifting conventions in \refappendix{syntax-lifting}.
The key property is that 
\begin{equation}
	c \ro d \ssimp
	c \ppseqc d
	\refeq
	c \pl d
\end{equation}
This follows from partial execution of the semantics: at no point is there an instruction within $c$ 
that prevents an instruction in $d$ from executing
(a trivial case is where $c = \Nil$).

To define the converse case, where $c$ always prevents $d$ from executing, consider the following.
We write $\blockallcd$ to indicate that at no point can $d$ reorder during the execution of $c$
(a trivial case is where $c = \fullfence$).
If $\blockallcd$ then one can treat them as sequentially composed.
We define these concepts with respect to the operational semantics.
\begin{eqnarray}
	\blocknextcd
	\asdef
	\forall \aca, d' @ 
		(d \tra{\aca} d' \land \Visa) \imp
				c \nro \aca
	\labeleqn{defn-blocknext}
	\\
	\blockallcd
	\asdef
	\forall t, c' @
		(c \xtra{t} c' \land \nreducesToNil{c'}) \imp
			\blocknext{c'}{d}
	\labeleqn{defn-blockall}
\end{eqnarray}
We write 
$\blocknextcd$
when all immediate next possible steps of $d$ are blocked by $c$
\refeqn{defn-blocknext}.
Thus $\blockallcd$ holds when, after any unfinished partial execution of $c$ via some trace $t$ resulting in $c'$, 
$c'$
continues to block $d$
\refeqn{defn-blockall}.  
We exclude from the set of partial executions the cases where execution is effectively finished, \ie,
when $c'$ is $\Nil$ or equivalent
(otherwise $\blockallcd$ would never hold as the final $\Nil$ allows reordering).
From these we can derive:
\begin{eqnarray}
	\blockallcd
	\aimp
	c \ppseqc d
	\refeq
	c \bef d
	\labellaw{blockall-cd}
\end{eqnarray}
Such reduction may lift to more complex structures, for instance,
the following law is useful
for sequentialising loops when each iteration has no overlap with the preceding or succeeding ones.
\begin{eqnarray}
	\blockall{c}{c} \imp \iteratecm \refeq \iteratec{\SCmm}
	\labellaw{blockall-iterate}
\end{eqnarray}

We now show how the application of reduction laws to eliminate (elucidate) the allowed reorderings in terms of 
the familiar sequential and parallel composition enables the application of standard techniques for analysis.

\section{Applying concurrent reasoning techniques}
\labelsect{reasoning}

In this section we show how the reduction rules in the previous section
can be used as the precursor to the application of already established techniques for proving properties of concurrent programs.
For \emph{correct}, real algorithms influenced by weak memory models there will be no reordering except where it does not violate whatever the desired property is
(be that a postcondition, or a condition controlling allowed interference).
In such circumstances our framework supports making the allowed reorderings explicit in the structure of the program, with a corresponding influence on the proofs of correctness.
Where a violation of a desired property occurs due to reordering, the framework also supports the construction of the particular reordering that leads to the problematic
behaviour.  

The \impro language includes several aspects which do not exist in a standard imperative language, namely, fences and the ordering constraints that annotate variables
and fences.  We start with a simple syntactic notion of equivalence modulo these features, reducing a program in \impro into its underlying `plain' equivalent
(fences and ordering constraints have no effect on sequential semantics directly).
We then explain how techniques such as Owicki-Gries and rely/guarantee may be applied.

\newcommand{\nhtrip}[3]{\neg \htrip{#1}{#2}{\neg #3}}
\newcommand{\stateeq}{\overset{\sigma}{=}}
\newcommand{\eqmodoc}{\overset{\mathsmaller{plain}}{\approx}}

\newcommand{\stripoc}[1]{{#1}_{\sf std}}
\newcommand{\stripocc}{\stripoc{\cmdc}}

\newcommand{\xasgne}{\modxocs \asgn e}
\subsection{Predicate transformer semantics and weakest preconditions}
\labelsect{eff-semantics}

The action-trace semantics of \refsect{trace-semantics} can be converted into a typical pairs-of-states semantics straightforwardly, as shown in \reffig{eff-semantics}.
Let the type
$\State$ be the set of total mappings from variables to values, and let the effect function $\effword: \Instr \fun \power (\State \cross \State)$ return a relation on states
given an
instruction.
We let `$\idn$' be the identity relation on states, and
given a Boolean expression $e$ we write $\sigma \in e$ if $e$ is $\True$ in state $\sigma$,
and $\evalse$ for the evaluation of $e$ within state $\sigma$ (note that ordering constraints are ignored for the purposes of evaluation).
The effect of an assignment $\modxocs \asgn e$ is a straightforward update 
of $x$ to the evaluation of $e$ \prefeqndefn{effa}, where $\Update{\sigma}{x}{v}$ is $\sigma$ overwritten so that $x$ maps to $v$.
A guard $\guarde$ is interpreted as a set of pairs of identical states that satisfy $e$ \prefeqndefn{effg},
giving trivial cases
$\eff{\tau} = \idn$ and $\eff{\Magic} = \ess$.
A fence $\fencepf$ has no affect on the state \prefeqndefn{efff}.
Conceptually, mapping $\effword$ onto an action trace $t$ 
yields a sequence of relations corresponding to a set of sequences of pairs of
states in a standard Plotkin-style treatment \cite{Plotkin}.
We can lift $\effword$ to traces by composing such a sequence of relations, which is defined recursively in \refeqn{efft-def},
and the effect of a command is given by the union of the effect of its traces \refeqn{effc-def}.

\newcommand{\bigcomp}{{\mathlarger{\comp}}}

\begin{figure}

\begin{eqnarray}
    \eff{\modxocs \asgn e} \aeq \{(\sigma, \Update{\sigma}{x}{\evalse})\}
	\labeleqndefn{effa}
    \\
    \eff{\guarde}  \aeq \{(\sigma, \sigma) | \sigma \in e\}
	\labeleqndefn{effg}
    \\
    \eff{\fencepf} \aeq \idn
	\labeleqndefn{efff}
    \\
    \also
    \eff{\eseq} \aeq \idn
    \\
    \eff{a\#t} \aeq \effa \comp \eff{t}
    \labeleqn{efft-def}
    \\
    \effc  \aeq \bigcup \{\eff{t} | t \in \Meaning{c} \}
    \labeleqn{effc-def}
    \\
    \Also
    \wpcq \asdef \{\sigma | \forall \sigma' @ (\sigma, \sigma') \in \effc \imp \sigma' \in q \}
	\labeleqndefn{wpre}
\end{eqnarray}
\caption{Sequential semantics}
\Description{TODO}
\labelfig{eff-semantics}
\end{figure}
The predicate transformer for weakest precondition semantics is given in \refeqndefn{wpre}.
A predicate is a set of states, so that
given a command $\cmdc$ and predicate $q$, $\wpcq$ returns the set of (pre) states $\sigma$ where every post-state $\sigma'$ related to $\sigma$ by $\effc$ satisfies $q$
(following, \eg, \cite{DijkstraScholten90}).
We define Hoare logic judgements with respect to this definition
(note that we deal only with partial correctness as we consider only finite traces).
From these definitions we can derive the standard rules of weakest preconditions and Hoare logic for commands such as
nondeterministic choice and sequential composition,
but there are no general compositional rules for \pseqc.

Trace refinement is related to these notions as follows.
\begin{theorem}[Refinement preserves sequential semantics]
\labelth{refsto-eff}
Assuming $c \refsto c'$ then
\begin{equation*}
    \eff{c'} \subseteq \effc 
    \qquad and \qquad
    \wpcq \imp \wpre{c'}{q} 
\end{equation*}
\end{theorem}
\begin{proof}
Straightforward from definitions.
\end{proof}

The action-trace semantics can be converted into a typical pairs-of-states semantics straightforwardly, based on
the effect function. 
The relationship with standard Plotkin style operational semantics \cite{Plotkin} is straightforward,
\ie,
\begin{equation*}
\mbox{
if
$
		c \tra{\aca} c' 
		$
and
$
		(\sigma,\sigma') \in \effa
$
then
$
		\langle c, \sigma \rangle \tra{} \langle c',\sigma' \rangle
$
}
\end{equation*}
The advantage of our approach is that the syntax of the action $\aca$ can be used to reason about allowed reorderings using 
\refrule{pseqcA}, whereas in general one cannot reconstruct or deduce an action from a pair of states.

\subsection{Relating to standard imperative languages}

\newcommand{\impsynequiv}{imperative-syntax-equivalent\xspace}

The language \impro extends a typical imperative language with four aspects which allow it to consider weak memory models:
ordering constraints on variables; fence instructions; \pseqc; and \piter.
The first two have no direct effect on the values of variables (the sequential semantics essentially ignores them), while the second two affect
the allowed traces and hence indirectly affect the values of variables.  However, both \pseqc and \piter can be instantiated to correspond to usual notions of execution;
hence we consider a \emph{plain} subset of \impro which maps to a typical imperative program.

\newcommand{\cmdcpl}{\plain{\cmdc}}

\begin{definition}[Plain imperative programs]
\labeldefn{plain}
A command $\cmdc$ of \impro is \emph{plain} if:
\begin{enumerate}[i]
\item
all instances of \pseqc in $\cmdc$ are parameterised by either $\SCmm$ or $\PARmm$, \ie, correspond to standard sequential or parallel composition;
and
\item
all instances of parallelized iteration in $\cmdc$ are parameterised by $\SCmm$, \ie, loops are executed sequentially.
\end{enumerate}
\end{definition}

\begin{definition}[Imperative syntax equivalence]
\labeldefn{imp-syntax-equiv}
Given a plain command $\cmdc$ in \impro, a command $\cmdc'$ in a standard imperative language (which does not contain memory ordering constraints on variables or fences),
is \emph{imperative-syntax-equivalent} to $\cmdc$ if
$c$ and $c'$ are structurally and syntactically equivalent except that:
\begin{enumerate}[i]
\item
all variable references in $\cmdc$ (which are of the form $\modxocs$ for some set of ordering constraints $\ocs$) appear in $\cmdc'$ as simply $x$;
and
\item
no fence instructions are present in $\cmdc'$, that is, they can be thought of as no-ops and removed.
\end{enumerate}
\end{definition}
We write $c \eqmodoc d$ if $c$ is \impsynequiv to $d$, or if there exists a $c'$ where $c \refeq c'$ and $c'$ is \impsynequiv to $d$.

For example, the following two programs are \impsynequiv
(where `$\scomp$' in the program on the right should be interpreted as as standard sequential composition, \ie, `$\ppseqs$' or `$\bef$' in this paper).
\begin{equation}
	\modxR \asgn 1 \ppseqs \scfence \ppseqs \modyX \asgn 1 \ppseq{\PARmm} r \asgn \modzA
	\quad 
	\eqmodoc
	\quad
	x \asgn 1 \scomp y \asgn 1
	\pl
	r \asgn z
\end{equation}
The program on the left is plain because all instances of \pseqc correspond to sequential or parallel execution.

Note that there is no imperative-syntax-equivalent form unless all instances of \pseqc have been eliminated/reduced to sequential or parallel forms.
A structurally typical case is the following, where $\cmdc$ reduces to a straightforward imperative-syntax-equivalent program.
\begin{eqnarray}
\labeleqn{2proc-fence-example}
	\hskip -5mm
	(\aca_1 \ppseqc \scfence \ppseqc \aca_2 )
	\pl
	(\acb_1 \ppseqc \scfence \ppseqc \acb_2 )
	\aeq
	(\aca_1 \bef \scfence \bef \aca_2 )
	\pl
	(\acb_1 \bef \scfence \bef \acb_2 )
	\\
	&\eqmodoc&
	(\aca_1 \bef \aca_2 )
	\pl
	(\acb_1 \bef \acb_2 )
\end{eqnarray}
For reasoning purposes one can use the latter program rather than the former.%
\footnote{Strictly speaking we should not use $\aca_1$, \etc, on both sides of $\eqmodoc$ because the plain version of $\aca_1$ does not contain ordering constraints on variables;
for brevity we ignore the straightforward modifications required in this case.
}
This is because for a plain program in \impro the semantics correspond to the usual semantic interpretation of imperative programs, ignoring ordering constraints and treating fences as no-ops.
\begin{theorem}[Reduction to standard imperative constructs]
\labelth{imp-syntax-equiv}
If $\cmdc \eqmodoc \cmdc'$, then $\effc$ is exactly equal to the usual denotation of $\cmdc'$ as a relation (or sequence of relations) on states;
and hence any state-based property of $\cmdc$ based on $\effc$ can be established for $\cmdc'$ that uses a standard denotation for programs.
\end{theorem}
\begin{proof}
From \refdefn{imp-syntax-equiv} and \reffig{eff-semantics} we can see that 
i) ordering constraints have no effect on states;
ii) fence instructions are equivalent to a no-op (in terms of states) and hence can be ignored;
iii) \refrule{pseqcA} reduces to the usual rule for sequential composition when reordering is never allowed, 
and for parallel composition when reordering is always allowed;
and 
iv) \refrule{iterate} and \refeqndefn{finiter} reduce to the usual sequential execution of a loop when instantiated with $\SCmm$.
\end{proof}

We introduce some syntax to help with describing examples.
\begin{definition}[Plain interpretation]
\labeldefn{plain-version}
Given a command $\cmdc$ in \impro, which need not be plain,
we let $\plainc$ be the \emph{plain interpretation} of $\cmdc$, that is, where i) all instances of \pseqc and \piter in $\cmdc$ are instantiated by $\SCmm$, except for
instances of \pseqc instantiated by \PARmm (corresponding to parallel composition) which remain as-is; and ii) fences in $\cmdc$ do not appear in $\plainc$.
\end{definition}

Establishing properties about programs may proceed as outlined below.
Assume for some program $\cmdc$ in \impro that its plain interpretation $\plainc$ satisfies a property $P$ under the usual imperative program semantics according to some method $M$.
There are three approaches:
\begin{enumerate}
\item
\emph{Reduction to plain equivalence.}
Show that the ordering constraints (due to variables or fences) within $\cmdc$ are such that it reduces equivalently to $\plainc$ (ie, no reordering is possible within $\cmdc$).
Hence $\cmdc$ satisfies $P$ using $M$ in \Clang.
\item
\emph{Reduction to some plain form.}
Reduce $\cmdc$ equivalently to some plain command $\cmdc'$ and apply $M$ (or some other known method) to $\plaincp$ to show that $P$ holds.
\item
\emph{Reduce and deny.}
Refine $\cmdc$ to some plain command $\cmdc'$ and apply $M$ (or some other known method) to $\plaincp$ to show that $P$ does not hold.
This corresponds to finding some new behaviour (due to a new ordering of instructions) that breaks the original property $P$; in this paper $\cmdc'$ tends to be a particular path and we straightforwardly
apply Hoare logic to deny the original property $P$.
\end{enumerate}

We now explain how standard notions of correctness can be applied in our framework.

\subsection{Hoare logic}

We define a Hoare triple using weakest preconditions (because we only consider finite traces we do not deal with potential non-termination).
\begin{eqnarray*}
    \htpcq \asdef p \imp \wpcq
\end{eqnarray*}
As with \refth{refsto-eff} refinement preserves Hoare-triple inferences.
\begin{equation}
	c \refsto c' \imp \htpcq \imp \htrip{p}{c'}{q}
\end{equation}
Given some plain $c'$ that is equivalent to $c$ we can apply Hoare logic.
\begin{eqnarray*}
	c \eqmodoc c' 
	\aimp
	\htpcq
	\iff
	\htpq{\cmdc'}
\end{eqnarray*}

Traditional Hoare logic is used for checking every final state satisfies a postcondition, but it also useful to consider \emph{reachable} states,
which can be defined using the conjugate pattern 
(see, \eg, \cite{Hoare78PredTrans,OfWPandCSP,PredDFA}; it is related to, but different from, O'Hearn's concept of incorrectness logic
\cite{OHearnIncorrectnessLogic}
(which is stronger except in the special case when post-state $q$ is $\False$)).
\begin{eqnarray*}
	\reachtrip{p}{c}{q}
	\asdef
	\neg \htrip{p}{c}{\neg q}
\end{eqnarray*}
Given these definitions we naturally derive the following properties in terms of a relational interpretation of programs.
\begin{eqnarray}
\htpcq
\aiff
\forall \ssp \in \effc @ \sigma \in p \imp \sigma' \in q
\labeleqn{htrip-states}
\\
\reachpcq
\aiff
\exists \ssp \in \effc @ \sigma \in p \land \sigma' \in q
\labeleqn{reachtrip-states}
\end{eqnarray}

Due to its familiarity we use Hoare logic for top-level specifications, although other choices could be made.
However
Hoare logic is of course lacking in the presence of concurrency (due to parallelism or reordering), 
and hence we later consider concurrent verification techniques.
At the top level
we use the following theorems to formally express our desired (or undesired) behaviours.
\begin{theorem}[Verification]
\begin{eqnarray}
\labelth{reasoning-equiv}
    \cmdc \refeq \cmdc'
    \sspace \aimp \sspace
    \htpcq \iff \htrip{p}{\cmdc'}{q}
\\
\labelth{reasoning-deny}
    \cmdc \refsto \cmdc'
    \land
    \htrip{p}{\cmdc'}{\neg q}
    \sspace \aimp \sspace
    \neg \htrip{p}{\cmdc}{q}
\\
\labelth{reasoning-possible}
    \cmdc \refsto \cmdc'
    \land
    \htrip{p}{\cmdc'}{q}
    \sspace \aimp \sspace
    \reachtrip{p}{\cmdc}{q}
\end{eqnarray}
\end{theorem}
\begin{proof}
All are automatic from \refeqndefn{refsto} and \refeqns{htrip-states}{reachtrip-states}.
\end{proof}
\refth{reasoning-equiv} allows properties of a reduced program ($\cmdc'$) to carry over to the original program ($\cmdc$). 
Alternatively by \refth{reasoning-deny} if \emph{any} behaviour ($\cmdc'$) is found to violate a property then that property cannot hold for the original ($\cmdc$).
Finally by \refth{reasoning-possible} if \emph{any} behaviour ($\cmdc'$) is found to satisfy a property then it is a \emph{possible} behaviour of the original ($\cmdc$).

\subsection{Owicki-Gries}
\labelsect{owicki-gries}

The Owicki-Gries method \cite{OwickiGries76} can be used
to prove a top-level Hoare-triple specification of parallel program $\cmdc$.
If $\cmdc$ involves reordering, then reducing $\cmdc$ to some plain form allows the application of Owicki-Gries;
several examples of this approach appear in \cite{SEFM21}.

For instance, recalling \refeqn{2proc-fence-example}, given a program of the form
$
	(\aca_1 \ppseqc \scfence \ppseqc \aca_2 )
	\pl
	(\acb_1 \ppseqc \scfence \ppseqc \acb_2 )
$
for which one wants to establish some property specified as a Hoare triple,
one can use Owicki-Gries on the plain program
$
	(\aca_1 \bef \aca_2 )
	\pl
	(\acb_1 \bef \acb_2 )
$
by \refth{imp-syntax-equiv}:
the fences enforce program order, but can be ignored from the perspective of correctness.

Now consider the following slightly more complex case with nested parallelism, 
where one process uses a form of critical section.
Assume $\aca \ro \acb$.
\begin{eqnarray}
	c \asdef 
	(\acc_1 \ppseqc \scfence \ppseqc \aca \ppseqc \acb \ppseqc \scfence \ppseqc \acc_2)
	\pl
	(\acc_3 \ppseqc \scfence \ppseqc \acc_4 )
	\notag
	\\
	\arefeq
	(\acc_1 \bef \scfence \bef (\aca \pl \acb) \bef \fence \bef \acc_2)
	\pl
	(\acc_3 \bef \scfence \bef \acc_4 )
	\labeleqn{rg-example}
	\\
	& \eqmodoc &
	(\acc_1 \bef (\aca \pl \acb) \bef \acc_2)
	\pl
	(\acc_3 \bef \acc_4 )
	\notag
\end{eqnarray}
Unfortunately the Owicki-Gries method is not directly applicable due to the nested parallelism,
but
in this simple case we can use the fact that execution of two (atomic) actions in parallel means that
either order can be chosen.
\begin{equation}
	\aca \pl \acb \ssrefeq (\aca \bef \acb) \choice (\acb \bef \aca)
\end{equation}
Given an Owicki-Gries rule for nondeterministic choice this can be reasoned about directly, or alternatively
we can lift this choice to the top level.  Continuing from \refeqn{rg-example}:
\begin{eqnarray}
	\arefeq
	(\acc_1 \bef \aca \bef \acb \bef \acc_2)
	\pl
	(\acc_3 \bef \acc_4 )
	\choice
	(\acc_1 \bef \acb \bef \aca \bef \acc_2)
	\pl
	(\acc_3 \bef \acc_4 )
\end{eqnarray}
Now Owicki-Gries can be applied to \emph{both} top-level possibilities.
Clearly this is not desirable in general, however, and in the next section we show how the compositional rely/guarantee method
can be used to handle nested parallelism.

\subsection{Rely/guarantee}

\newcommand{\xbar}{\bar{x}}
\newcommand{\idxbar}{\id{\xbar}}
\newcommand{\idbar}[1]{\id{\overline{#1}}}
\newcommand{\qmid}{q_{mid}}
\newcommand{\Stable}[2]{stable(#1,#2)}
\newcommand{\Stabler}[1]{\Stable{#1}{r}}
\newcommand{\Stablepr}{\Stabler{p}}
\newcommand{\Stableqr}{\Stabler{q}}

\newcommand{\precomp}{\dres}

\newcommand{\postState}[1]{postState(#1)}
\newcommand{\postStater}{\postState{r}}
\newcommand{\preState}[1]{preState(#1)}
\newcommand{\preStater}{\preState{r}}
\newcommand{\Trans}[1]{#1^*}
\newcommand{\Transr}{\Trans{r}}

\newcommand{\relasgnxe}{\Meaning{x \asgn e}}

Rely/guarantee \cite{Jones-RG1,Jones-RG2} (see also \cite{SoundnessRG,rgsos12,rgsos14})
is a compositional proof method for parallel programs.
A rely/guarantee quintuple $\quintprgqc$ states that program $c$ establishes postcondition $q$ provided it is executed from a
state satisfying $p$ \emph{and} within an environment that satisfies (rely) relation $r$ on each step; in addition $c$ also satisfies (guarantee) relation $g$
on each of its own steps.
For instance, $\quint{x = 0}{x \leq x'}{c}{y' = y}{x \geq 10}$ states that $c$
establishes $x \geq 10$ when it finishes execution, and guarantees that it will not modify $y$,
provided initially $x = 0$ and the environment only ever increases $x$.  In the relations above we have used the common convention in relations that 
primed variables ($x'$) refer to their value in the post-state and unprimed variables ($x$) refer to their value in the pre-state.
A top-level Hoare-triple specification $\htpcq$ can be related to a rely/guarantee quintuple by noting
$
	\rgq{p}{\idn}{\True}{q}{c}
	\imp
	\htpcq
$,
that is, if $p$ holds initially in some top-level context that does not modify any variables, then
$\cmdc$ establishes $q$ (with the weakest possible guarantee).

\OMIT{
If we let $\id{X}$ be the identity relation on all the variables in set $X$, and let $\compl{X}$ be the complement of set $X$, then
rely/guarantee quintuples are related to Hoare triples by the following equation.
\begin{equation}
	\htpcq
	\iff
	\rgq{p}{\id{\fvc}}{\id{\compl{\wvc}}}{q}{c}
\end{equation}
That is, if $\htpcq$ then $c$ will establish $q$ from $p$ in any environment
where the free variables of $c$ are unmodified, and guarantees to not modify any variables other than the write variables of $c$.
}

Reasoning using the rely/guarantee method is compositional over parallel conjunction; for instance, without going into detail about rely/guarantee inference rules,
consider the plain program from \refeqn{rg-example}.
We can show a rely/guarantee quintuple holds provided we can find rely/guarantee relations that control the communication between the two subprocesses;
we abstract from these nested relations using different names.
\begin{eqnarray*}
\quintprgqc
\aiff
\quint{p}{r}{
	(\acc_1 \bef (\aca \pl \acb) \bef \acc_2)
	\pl
	(\acc_3 \bef \acc_4 )
}{g}{q}
\quad \mbox{(by \refeqn{rg-example} and \refth{imp-syntax-equiv})}
	\\
	\arevimp
\quint{p}{r_1}{
	(\acc_1 \bef (\aca \pl \acb) \bef \acc_2)
}{g_1}{q_1}
\land
\quint{p}{r_2}{
	(\acc_3 \bef \acc_4 )
}{g_2}{q_2}
	\\
	\arevimp
\quint{p}{r_1}{
	\acc_1 
}{g_1}{q_1'}
\land
\quint{q_1'}{r_1}{
\aca \pl \acb
}{g_1}{q_1''}
\land
\quint{q_1''}{r_1}{
\acc_2
}{g_1}{q_1}
\land
\\
&&
\quint{p}{r_2}{
	(\acc_3 \bef \acc_4 )
}{g_2}{q_2}
\end{eqnarray*}
This is straightforward application of standard rely/guarantee inference rules (where we leave the format of the predicates and relations unspecified).
Reasoning may proceed from this point, in particular,
analysing the quintuple containing the nested parallel composition $\aca \pl \acb$. 
The question becomes whether the guarantee and the intermediate relation $q_1''$ is satisfied regardless of the order that $\aca$ and
$\acb$ are executed, and if not, a fence or ordering constraints will need to be introduced to enforce order and eliminate the undesirable proof obligation.
While completing the proof may or may not be straightforward, the point is 
it will be matter of the specifics of the example in question, and not of a tailor-made inference system to manage a complex semantic representation.
The initial work (as shown in \refeqn{rg-example}) is to elucidate the allowed reorderings as either sequential or parallel composition.

\subsection{Linearisability}
\labelsect{linearisability}

\newcommand{\opname}{\T{op}}
Linearisability is a correctness condition for concurrent objects \cite{linearisability}, which is not an ``end-to-end'' property such as Hoare triples
and rely/guarantee quintuples, but rather requires that operations on some data structure appear to take place atomically, with weakened liveness guarantees.
The following abstract program $\opname$ follows a common pattern for operations on lock-free, linearisable data structures (in this case $x$),
where there may be other processes also executing $\opname$.
\begin{equation}
\begin{array}{l}
	\Repeat 
		\\ \t1
		r_1 \asgn x 
		\ppseqc \\ \t1
		r_2 \asgn f(r_1) 
		\ppseqc \\ \t1
		b \asgn \CAS{x}{r_1}{r_2}
	\ppseqc \\ 
	\Until b
\end{array}
\end{equation}
The algorithm repeatedly reads the value of $x$ (into local variable $r_1$), calculates a new value for
$x$ (applying some function $f$, and storing the result in local variable $r_2$), and attempts to atomically update $x$ to $r_2$; if interference is detected ($x$ has changed since being read into $r_1$)
then the loop is retried, otherwise the operation may complete.

The natural dependencies arising from the above structure naturally maintain order, that is,
$
\datadep{
		r_1 \asgn x 
}{\datadep{
		r_2 \asgn f(r_1) 
}{\datadep{
		b \asgn \CAS{x}{r_1}{r_2}
}{
	\guard{b}
}}}
$,
hence $\opname \eqmodoc \plain{\opname}$, and 
any reasoning about $\opname$ executing under sequential consistency -- in a vacuum -- applies to the \Clang version of $\opname$.
(Other algorithms may of course have fewer dependencies which will manifest as more complex parallel structures, which must be elucidated via reduction.)
Alternative approaches to reasoning about linearisability under weak memory models typically modify the original definition in some way to account for reorderings
\cite{ObservationalApproach,VerifyLinTSO}.

However when considering a calling context reordering must also be taken into account.
That is, assume $c$ is some program which contains a call to $\opname$, and that $c$ is just one of many parallel processes which may be calling $\opname$.
Multiple calls to $\opname$ (or other operations on $x$) within $c$ are unlikely to cause a problem because the natural dependencies on
$x$ will prevent problematic reordering; however one must consider the possibility of unrelated instructions in $c$ reordering with (internal) instructions of $\opname$.
For instance, consider a case where $x$ implements a queue and $\opname$ places a new value into the queue, and within $c$ a flag $f$ is set immediately after calling $\opname$
to indicate the enqueue has happened, \ie, 
\begin{equation*}
	c \sdef \opname() \ppseqc flag \asgn \True
\end{equation*}
The assignment $flag \asgn \True$ can be reordered with instructions within $\opname$ and destroy any reasoning that uses $flag$ to determine when $\opname$ has been called.
Placing a release constraint on the flag update resolves the issue (which can be shown straightforwardly in our framework since
$
	\opname() \ppseqc \modR{flag} \asgn \True \refeq \opname() \bef \modR{flag} \asgn \True
$, or in other words, $c \eqmodoc \plainc$), but more generally this leaves a
question about how to capture dependencies within the \emph{specification} of $\opname$.
Hence, 
while one can argue that $\opname$ is ``linearisable'' in the sense that if it is executed in parallel with other
linearisable operations on $x$ it will operate correctly, whether or not a calling context $c$ works correctly
will still depend on reordering
\cite{LibraryAbstractionsForC,SmithGrovesRefine20}; this can also be addressed in our framework, using reduction and applying an established technique (\eg, \cite{SoundCompleteLin,ModelCheckLin,Filipovic10})

\section{Examples}
\labelsect{examples}

\newcommand{\exfont}[1]{\mathsf{#1}}
\newcommand{\oota}{\exfont{oota}}
\newcommand{\ootaplain}{\plain{\oota}}
\newcommand{\ootaC}{\oota}
\newcommand{\ootaA}{\exfont{oota_A}}
\newcommand{\ootaD}{\exfont{oota_D}}

\newcommand{\vinit}[1]{0_{#1}}
\newcommand{\vinitxy}{\vinit{x,y}}
\newcommand{\vinitxyr}{\vinit{x,y,r}}
\newcommand{\vinitxyrr}{\vinit{x,y,r_1,r_2}}
\newcommand{\vinitxyrf}{\vinit{x,y,r,f}}

In this section we explore the behaviour of small illustrative examples from the literature, specifically classic 
``litmus test'' patterns that readily show the weak behaviours of memory models, and also examples taken from the \Clang standard.

To save space we define the following initialisation predicate for $x,y,\ldots$ a list of variables.
\begin{equation}
	\vinit{x,y,\ldots} \sdef x = 0 \land y = 0 \land \ldots
\end{equation}


\subsection{Message passing pattern}


\newcommand{\MPplain}{\plain{\MP}}

Consider the message passing (``MP'') communication pattern.
\begin{equation}
\labeleqndefn{mp}
\MP \sdef
	(x \asgn 1 \ppseqc flag \asgn 1)
	\pl
	(f \asgn flag \ppseqc r \asgn x)
\end{equation}
The question is whether in the final state that $f = 1 \imp r = 1$, \ie, if the ``flag'' is observed to have been set, can one assume that the ``data'' ($x$) has been transferred?
This is of course expected under a plain interpretation (\refdefn{plain-version}).
\begin{theorem}[Message passing under sequential consistency]
	\labelth{rgq-mpsc}
\begin{equation}
	\htrip{\vinitxyrf}{\MPplain}{f = 1 \imp r = 1}
\end{equation}
\end{theorem}
\begin{proof}
The proof is checked in Isabelle/HOL using Nieto's encoding of rely/guarantee inference \cite{RGinIsabelle};
the key part of the proof is the left process guarantees $x = 0 \imp flag = 0$ and $flag = 1 \imp x = 1$ (translated into relational form).
\end{proof}

If we naively code this pattern in \Clang this property no longer holds.
\begin{theorem}[Naive message passing fails under $\Cmm$ ]
	\labelth{rgq-mpc}
\begin{equation*}
	\neg \htrip{\vinitxyrf}{\MPC}{f = 1 \imp r = 1}
\end{equation*}
\end{theorem}
\begin{proof}
By definition of $\Cmm$ (\refmm{Cmm}) we have
both
$    (x \asgn 1) \ro (flag \asgn 1)
$
and
$
    (f \asgn flag) \ro (r \asgn x)
$.
\begin{align}
	\MPC & \sdef
		(x \asgn 1 \ppseqc flag \asgn 1)
		\pl
		(f \asgn flag \ppseqc r \asgn x)
& \mbox{\refeqndefn{mp}}
		\\
		& \refeq
		x \asgn 1 \pl flag \asgn 1
		\pl
		f \asgn flag \pl r \asgn x
& \mbox{by \reflaw{2actions-reduce}}
		\\
		& \refsto 
		flag \asgn 1 \bef f \asgn flag \bef r \asgn x \bef x \asgn 1
& \mbox{by \reflaw{fix-interleaving}}
\end{align}
All instructions effectively execute in parallel; in the final step we have picked
one particular interleaving that breaks the expected postcondition,
that is,
\begin{equation}
	\htrip{\vinitxyrf}{
		flag \asgn 1 \bef f \asgn flag \bef r \asgn x \bef x \asgn 1
	}{
		f = 1 \land r = 0
	}
\end{equation}
Since $f = 1 \land r = 0 \imp \neg (f = 1 \imp r = 1)$ we complete the proof by \refth{reasoning-deny}.
\end{proof}

\Clang's release/acquire atomics are the recommended way to instrument message passing; we make this explicit below
using release-acquire constraints on $flag$ (leaving $x$ relaxed).
\begin{eqnarray}
	\MPra &\sdef&
	(x \asgn 1 \ppseqc \modR{flag} \asgn 1)
	\pl
	(f \asgn \modA{flag} \ppseqc r \asgn x)
\end{eqnarray}
The new constraints prevent reordering in each branch and therefore the expected outcome is reached.
\begin{theorem}[Release/acquire message passing maintains sequential consistency]
	\labelth{rgq-mp}
\begin{equation}
	\htrip{\vinitxyrr}{\MPra}{f = 1 \imp r = 1}
\end{equation}
\end{theorem}
\begin{proof}
By definition of $\Cmm$ we have
both
$    (x \asgn 1) \nro (\modR{y} \asgn 1)
$
and
$
    (r_1 \asgn \modA{y}) \nro (r \asgn x)
$.
Hence
\begin{equation}
	\MPra 
	\sspace = \sspace
		(x \asgn 1 \bef \modR{y} \asgn 1)
		\pl
		(r_1 \asgn \modA{y} \bef r \asgn x)
	\sspace \eqmodoc \sspace
	\MPplain
\end{equation}
The proof follows immediately from \refths{rgq-mpsc}{imp-syntax-equiv}.
\end{proof}

An alternative approach to restoring order is to insert fences, \eg,
\begin{gather}
	(x \asgn 1 \ppseqc \scfence \ppseqc flag \asgn 1)
	\pl
	(f \asgn flag \ppseqc \scfence \ppseqc r \asgn x)
	\\
	(x \asgn 1 \ppseqc \relfence \ppseqc flag \asgn 1)
	\pl
	(f \asgn flag \ppseqc \acqfence \ppseqc r \asgn x)
\end{gather}
No reordering is possible within each thread and hence these cases also reduce to $\MPplain$.
We emphasise that these proofs are as straightforward as one would expect: the programmer has enforced order and so properties of the
sequential version of the program carry over.  There is no need to appeal to complex global data types that maintain information about fences, orderings, \etc.

\subsection{Test-and-set lock}

\newcommand{\thelock}{\lbl}
\newcommand{\taken}{taken}
\newcommand{\lockproc}{\T{lock()}}
\newcommand{\plainlockproc}{\plain{\lockproc}}
\newcommand{\unlockproc}{\T{unlock()}}
\newcommand{\RPT}[3]{{\Repeat}^{#1}~ #2 \Until #3}
\newcommand{\RPTm}[2]{\RPT{\mm}{#1}{#2}}
\newcommand{\RPTmcb}{\RPT{\mm}{\cmdc}{b}}

We now consider the application of the framework to more realistic code, in this case a lock implementation 
taken from \cite{HerlihyShavit2011}(Sect. 7.3).
Conceptually the shared lock $\thelock$ is represented as a boolean which is $\True$ when some process holds the lock, and $\False$ otherwise.
\begin{eqnarray}
	\RPT{\mm}{c}{b}
	\asdef
	c \ppseqm \iterate{(\guardnb \ppseqm c)}{\mm} \ppseqm \guardb
	\labeleqndefn{repeat}
	\\
	r \asgn x.getAndSet(v) 
	\asdef
	\seqT{r \asgn \modA{x} \asep \modR{x} \asgn v}
	\labeleqndefn{getAndSet}
	\\
	\lockproc \asdef
	\RPT{\Cmm}{
		\taken \asgn \thelock.getAndSet(\True)
	}{
		\neg \taken
	}
	\labeleqndefn{lockproc}
	\\
	\unlockproc \asdef
		\modR{\thelock} \asgn \False
	\labeleqndefn{unlockproc}
\end{eqnarray}
A ``repeat-until'' command $\RPTmcb$ repeatedly executes $\cmdc$ (at least once) until $b$ is true, under memory model $\mm$, which we encode using
parallelized iteration \prefeqndefn{repeat}.
A ``get-and-set'' command $r \asgn x.getAndSet(e)$ returns the initial value of $x$ into $r$ and updates $x$ to the value of $e$, as a single step \prefeqndefn{getAndSet}.
The load of $x$ is defined as an $\acquire$ access and the update is a $\release$ write.
Using these the concurrent $\lockproc$ procedure is defined to repeatedly set $\thelock$ to $\True$, and finish when a $\False$ value for $\thelock$ is read.
If the lock is already held by another process then the get-and-set has no effect and the loop continues, until the case where $\thelock$ is $\False$ (as recorded in the local variable $\taken$),
when $\thelock$ is updated to $\True$ to indicate the current process now has the lock.
Unlocking is implemented by simply setting $\thelock$ to $\False$.

We have the following general rule for spin loops of the form in \lockproc.
\begin{equation}
	\datadep{\aca}{\guardb} 
	\quad \imp \quad
	\RPT{\Cmm}{\aca}{b} 
	\refeq
	\RPT{\SCmm}{\aca}{b} 
	\labellaw{repeat-sc}
\end{equation}
Intuitively, since each iteration of the loop updates a variable that is then checked in the guard, no reordering is possible, and it is as if the loop is executed in sequential order.
This holds by the following reasoning.
\begin{derivation}
	\step{
    	\RPT{\Cmm}{\aca}{b} 
	}

	\trans{\sdef}{\refeqndefn{repeat}}
	\step{
		\aca \ppseqc \iterate{(\guardnb \ppseqc \aca)}{\Cmm} \ppseqc \guardb
	}

	\trans{\refeq}{Using \reflaw{2actions-keep-order} by assumption $\datadep{\aca}{\guardb}$ (and hence also $\datadep{\aca}{\guardnb}$)}
	\step{
		\aca \ppseqc \iterate{(\guardnb \bef \aca)}{\Cmm} \ppseqc \guardb
	}

	\trans{\refeq}{Using \reflaw{blockall-iterate}, since $\datadep{\aca}{\guardb}$ implies $\blockall{(\guardnb \bef \aca)}{(\guardnb \bef \aca)}$.}
	\step{
		\aca \ppseqc \iterate{(\guardnb \bef \aca)}{\SCmm} \ppseqc \guardb
	}

	\trans{\refeq}{Similarly using \reflaw{blockall-cd} and $\blockall{\aca}{\blockall{(\guardnb \bef \aca)}{\guardb}}$.}
	\step{
		\aca \bef \iterate{(\guardnb \bef \aca)}{\SCmm} \bef \guardb
	}

	\trans{\sdef}{\refeqndefn{repeat}}
	\step{
    	\RPT{\SCmm}{\aca}{b} 
	}

\end{derivation}

Recalling that $\plainlockproc$ is the plain version of $\lockproc$, \ie, ignoring the ordering constraints and using sequential composition only,
we have the following.
\begin{theorem}
\labelth{lockproc-plain}
$\lockproc \eqmodoc \plainlockproc$
\end{theorem}
\begin{proof}
$\lockproc$ reduces to sequential form by $\datadep{taken \asgn l.getAndSet(\True)}{\guard{\taken}}$ and \reflaw{repeat-sc},
which is \impsynequiv to $\plainlockproc$ by \refdefn{imp-syntax-equiv}.
\end{proof}

Therefore this $\Cmm$ implementation is `correct' if the original is.  
However, as discussed in \refsect{linearisability}, one must consider the calling context.
\refth{lockproc-plain} does not depend on the ordering constraints in the definition of \lockproc, 
as its natural data dependencies maintain order.
However this does not also guarantee that for instance,
$
	\lockproc \ppseqc c \ppseqc \unlockproc
	\refeq
	\lockproc \bef c \bef \unlockproc
$,
which may be a desirable property if the behaviour of $c$ assumes the lock is held.
By \refeqndefn{lockproc} we have
$
	\getocs{\lockproc} = \{\release, \acquire\}
$
and
$
	\getocs{\unlockproc} = \{\release\}
$
which in many cases will be enough to allow sequential reasoning in the calling context;
alternatively fences could be inserted around $c$.

\subsection{Out-of-thin-air behaviours}
\labelsect{oota}

We now turn our attention to the ``out of thin air'' problem, where some memory model specifications allow
values to be assigned where those values do not appear in the program.
Firstly consider the following program, which appears in the $\Clang$ standard.
\begin{eqnarray}
	\oota
	\asdef
	r_1 \asgn x \cbef (\IFsC{r_1 = 42}{y \asgn 42})
	~\pl~
	r_2 \asgn y \cbef (\IFsC{r_2 = 42}{x \asgn 42})
\end{eqnarray}
Under a plain interpretation 
neither store ever happens: one of the loads must occur, and the subsequent test fail, first, preventing the condition in the other process from succeeding.
\begin{theorem}
\labelth{oota-plain}
$\htrip{\vinitxyrr}{\ootaplain}{x = 0 \land y = 0}$.
\end{theorem}
\begin{proof}
Trivial using Owicki-Gries reasoning, checked in Isabelle/HOL \cite{OGinIsabelle}.
\end{proof}

However, this behaviour is allowed under the \Clang memory model according to the specification
(although compiler writers are discouraged from implementing it!).
The behaviour is allowed in our framework, that is,
it is possible for both $r_1$ and $r_2$ to read 42.
This is because (unlike hardware memory models \cite{pseqc-arXiv}) we allow stores to come before guards (via the $\roG$ relation, \eg, \refeqn{c11ro-gw}).
\begin{theorem}
$\reachtrip{\vinitxyrr}{\ootaC}{x = 42 \land y = 42}$.
\end{theorem}
\begin{proof}
We have $\guard{r_1 = 42} \ro y \asgn 42$, and similarly for $x$, hence
\begin{derivation}
	\step{
		r_1 \asgn x \ppseqc (\IFsC{r_1 = 42}{y \asgn 42}) 
	}

	\trans{\refeq}{\Refeqndefn{if}}
	\step{
		r_1 \asgn x \ppseqc 
			(\guard{r_1 = 42} \ppseqc y \asgn 42)
			\choice
			\guard{r_1 \neq 42}
	}

	\trans{\refsto}{\reflaw{chooseL}}
	\step{
		r_1 \asgn x \ppseqc \guard{r_1 = 42} \ppseqc y \asgn 42
	}

	\trans{\refeq}{\reflaw{2actions-reduce}}
	\step{
		r_1 \asgn x \ppseqc (\guard{r_1 = 42} \pl y \asgn 42)
	}

	\trans{\refsto}{\reflaw{fix-interleaving} (taking the right-hand action); \reflaw{reorder-action} from $r_1 \asgn x \ro y \asgn 42$}
	\step{
		y \asgn 42 \bef r_1 \asgn x \bef \guard{r_1 = 42} 
	}

	\end{derivation}
The second process reduces similarly to 
$
		x \asgn 42 \bef r_2 \asgn y \bef \guard{r_2 = 42}$.
Interleaving the two processes (\reflaw{fix-interleaving}) gives the following reduction
to a particular execution.
\begin{eqnarray*}
\ootaC \arefsto 
		(y \asgn 42 \bef r_1 \asgn x \bef \guard{r_1 = 42} )
		\pl
		(x \asgn 42 \bef r_2 \asgn y \bef \guard{r_2 = 42})
\\
\arefsto
		y \asgn 42 \bef 
		x \asgn 42 \bef 
		r_1 \asgn x \bef 
		r_2 \asgn y \bef 
		\guard{r_1 = 42} 
		\bef
		\guard{r_2 = 42}
\end{eqnarray*}
Straightforward sequential reasoning gives the following.
\begin{equation}
	\htrip{\vinitxyrr}{
		y \asgn 42 \bef 
		x \asgn 42 \bef 
		r_1 \asgn x \bef 
		r_2 \asgn y \bef 
		\guard{r_1 = 42} 
		\bef
		\guard{r_2 = 42}
	}
	{x = 42 \land y = 42}.
\end{equation}
The final state is therefore possible
by \refth{reasoning-possible}.
\end{proof}

Under \emph{hardware} weak memory models (the observable effects of) writes can not happen before branch points, and so out-of-thin-air behaviours 
are not possible.

Consider 
the following variant of $\oota$.
\begin{equation}
\labeleqndefn{ootaD}
	\ootaD
	\sssdef
	r_1 \asgn x \ppseqc (\IFsC{r_1 = 42}{y \asgn r_1}) 
	~~\pl~~
	r_2 \asgn y \ppseqc (\IFsC{r_2 = 42}{x \asgn r_2}) 
\end{equation}
The inner assignments have changed from $y \asgn 42$ (resp. $x \asgn 42$) to $y \asgn r_1$ (resp. $x \asgn r_2$).
Arguably the compiler knows that within the true branch of the conditional it must be the case that 
$r_1 = 42$, and thus the assignment $y \asgn r_1$ can be treated as $y \asgn 42$, reducing to the original $\ootaA$.
But this outcome is expressly disallowed, by both the standard and naturally in our framework.
That is, we can show that every possible final state satisfies $x = y = 0$.
\begin{theorem}
\labelth{ootaD-c11}
$\htrip{\vinitxyrr}{\ootaD}{x = 0 \land y = 0}$.
\end{theorem}
\begin{proof}
The initial load into $r_1$ creates a dependency with the rest of the code, \ie,
$
	 \blockall{r_1 \asgn x}{(\IFC{r_1 = 42}{y \asgn r_1})}
$.
Hence we can sequence the initial load of $x$ (into $r_1$) with the remaining code.
\begin{eqnarray}
	r_1 \asgn x \ppseqc (\IFsC{r_1 = 42}{y \asgn r_1}) 
	\ssrefeq
	r_1 \asgn x \bef (\IFsC{r_1 = 42}{y \asgn r_1}) 
\end{eqnarray}
Although the guard and the assignment in the conditional may be reordered with each other (\ie, $\guard{r_1 = 42} \ro y \asgn r_1$), the fact that the initial load must happen first means that,
similarly to \refth{oota-plain},
there is no execution of $\ootaD$ in which a non-zero value is written to any variable.
\end{proof}

Note that
the \Clang standard allows the behaviour of $\oota$ and forbids the behaviour of $\ootaD$, and both results arise
naturally in our semantics framework without the introduction of any ad-hoc mechanisms.  We return to the question of whether guards should be allowed to simplify instructions in 
\refsect{sfp}.

\section{Incremental evaluation of code}
\labelsect{incremental-evaluation}

In this and subsequent sections we consider more complex aspects of the \Clang language and execution where they relate to concurrent behaviour,
namely, in this section, non-atomic evaluation of instructions (until now we have assumed all instructions appearing in the text of the program
are executed in a single indivisible step), in \refsect{expression-optimisations} optimisations of expressions (reducing expressions to improve efficiency), 
and in \refsect{forwarding} forwarding (using earlier program text to simplify later instructions).
We show how each can be incorporated into the framework straightforwardly 
without any need for change to the underlying definition of the $\Cmm$ memory model,
although the consequences for reasoning about particular programs may not be straightforward.

\subsection{Incremental evaluation of expressions}

In \Clang one cannot assume expressions are evaluated in a single state.
That is, programmers are allowed to write ``complex'' assignments and conditions but these may be compiled into multiple (indivisible)
assembler instructions.  For instance, the assignment $z \asgn x + y$ may compile into (at least) three separation instructions, one to load the value of $x$, one to load the value of $y$,
and one to finally store the result to $z$.

\renewcommand{\load}[2]{\guard{#1 = #2}}
\renewcommand{\loadx}[1]{\load{x}{#1}}
\renewcommand{\loady}[1]{\load{y}{#1}}
\newcommand{\loadxA}[1]{\load{\modxA}{#1}}
\newcommand{\loadyX}[1]{\load{\modyX}{#1}}
\newcommand{\loadxv}{\loadx{v}}
\newcommand\ry[1]{r_{y#1}}

\newcommand{\lex}[1]{\abs{#1}}
\newcommand{\apply}[1]{\mathsf{apply}(#1)}
\newcommand{\applyu}[1]{\apply{\unop, #1}}
\newcommand{\applyuv}{\applyu{v}}
\newcommand{\applyb}[1]{\apply{\binop, #1}}
\newcommand{\applybvv}{\applyb{v_1, v_2}}

We give an operational semantics for incremental expression evaluation in \reffig{semantics-expr}.  
(We use the term ``incremental'' rather than the more usual ``non-atomic'' to avoid terminology clash with \Clang's \T{atomics}.)
Recall the syntax of an expression $e$ from \reffig{c11-ro}.
Each evaluation step of $e$ either loads a value of a free variable or reduces the expression in some way, and 
evaluation stops when a single value remains.
\refrule{eval-var} states that a variable access $\modxocs$ is evaluated to a value $v$ via the guard $\guard{\modxocs = v}$.  
Any choice of $v$ that is not the value of $x$ will result in a false guard,
\ie, leads to an infeasible behaviour, which can be ignored.  Only the correct value for $v$ leads to a feasible behaviour.
Note that the set of ordering constraints on $x$ also appears in the label.
\refrule{eval-unop} evaluates a unary expression $\unop e$ by simply inductively evaluating $e$, while 
\refrule{eval-binop} similarly evaluates each operand of a binary expression, in either order (it is of course straightforward to instead insist on left-to-right evaluation).
\refrule{eval-apply} uses a meta-level functor application method $\apply{.}$ to calculate 
the final value of an expression once all variables have been evaluated.

\newcommand{\shortrule}[3]{
	\shorteqn{#1}{\labelrule{#2} #3 }
}

\begin{figure}
	\shortrule{0.27\textwidth}{eval-var}{
		\modxocs \ttra{\guard{\modxocs = v}} v
	}
	\shortrule{0.35\textwidth}{eval-unop}{
		\Rule{
			e \tra{\aca} e'
		}{
			\unop e
			\tra{\aca}
			\unop e'
		}
	}

	\shortrule{0.66\textwidth}{eval-binop}{
		\Rule{
			e_1 \tra{\aca} e_1'
		}{
			e_1 \binop e_2
			\tra{\aca}
			e_1' \binop e_2
		}
		\qquad
		\Rule{
			e_2 \tra{\aca} e_2'
		}{
			e_1 \binop e_2
			\tra{\aca}
			e_1 \binop e_2'
		}
	}

	\shortrule{0.66\textwidth}{eval-apply}{
		\Rule{
			\applyuv = v'
		}{
			\unop v \tra{\tau} v'
		}
		\qquad
		\Rule{
			\applybvv = v'
		}{
			v_1 \binop v_2 \tra{\tau} v'
		}
	}

\caption{Incremental expression evaluation semantics}
\labelfig{semantics-expr}
\Description{TODO}
\end{figure}

As an example, consider the following possible evaluation of the expression $\modxA + \modyX$.
\begin{equation}
	\modxA + \modyX \ttra{\loadxA3} 3 + \modyX \ttra{\loadyX2} 3 + 2 \tra{\tau} 5
	\labeleqn{expr-eval-example}
\end{equation}
The first step applies \refrule{eval-binop} and (inductively) \refrule{eval-var}, and the second step is similar.
The choice of values 3 and 2 are arbitrary, and any values could be chosen, reflecting the behaviour of the code in
any state.
The third and final step follows from (the assumed interpretation) $\apply{+, 3, 2} = 5$.
Note that the (relaxed) load of $y$ is not restricted by the acquire of $x$ appearing within the same expression, since we have given a nondeterministic expression evaluation order.
Of course, one can change the rule for evaluating binary expressions to evaluate left-to-right and consider constraints appearing ``earlier'' in the expression.

In practice \Clang is not restricted to ``laboriously'' evaluating each instruction step by step; in some cases evaluation can be wrapped 
into a single optimisation.  We give such a rule in \refsect{expression-optimisations}, which subsumes \refrule{eval-apply}.

\subsection{Incremental execution of instructions}

Now consider the incremental execution of a single instruction
(for brevity, in this section we assume a single instruction $\aca$ is the base action of the language, rather than 
a list $\vecaca$ as in previous sections \prefeqndefn{cmd}; we give a full semantics for incremental execution of lists of instructions in \refappendix{indivis-vecaca}).
We have the concept of \emph{indivisible} ($\indivisword$) actions, which are the only instructions that may be executed directly in the
operational semantics. 
We define an indivisible instruction as one where there are no shared variables to be read, \ie,
\begin{eqnarray*}
\labeleqn{defn-indivis}
	\indivisa \asdef \rsva = \ess
\end{eqnarray*}
From this we can derive
\begin{equation}
	\indivisact{x \asgn e} 
	\iff
	\sve = \ess \qquad
	\indivisact{\guarde} 
	\iff
	\sve = \ess \qquad
	\indivisact{\fencepf} 
	\iff
	\True   
	\labeleqndefn{indivisact}
\end{equation}
For instance, $x \asgn 1$ is indivisible, while $r \asgn y$ is not -- $y$ must be separately loaded and the result later stored into $r$.

\begin{figure}

\shortrule{0.3\textwidth}{eval-asgn}{
	\Rule{
		e \tra{\aca} e'
	}{
		x \asgn e \tra{\aca} x \asgn e'
	}
}
\shortrule{0.3\textwidth}{eval-guard}{
	\Rule{
		e \tra{\aca} e'
	}{
		\guarde \tra{\aca} \guard{e'}
	}
}
\shortrule{0.3\textwidth}{eval-indivis}{
	\Rule{
		\indivisa
	}{
		\aca \tra{\aca} \Nil
	}
}

\caption{Incremental instruction execution semantics}
\labelfig{semantics-non-atomic-instructions}
\Description{TODO}
\end{figure}

Incremental execution rules for instructions are given in \reffig{semantics-non-atomic-instructions}.
\refrule{eval-asgn} states that an assignment instruction is evaluated by first evaluating the assigned expression,
and similarly \refrule{eval-guard} states that a guard is executed by first incrementally evaluating the expression.
\refrule{eval-indivis} states that
directly executable instructions can be executed as a single action. 
This rule applies for fences, and when evaluation of assignment or guard instructions has reduced them to an indivisible form.
Note that we allow the (final) evaluation steps to include an arbitrary number of local variables.
We insist only on shared variables being evaluated in separate steps, as these involve interactions with the memory system.  

\newcommand{\tsep}{~~~}

\newcommand{\pChoice}[2]{\left(\underset{#1}{\Bigg\sqcap} ~ #2 \right)}
\newcommand{\pChoicev}[1]{\pChoice{v}{#1}}

As an example of instruction evaluation, recalling \refeqn{expr-eval-example}, we place this expression evaluation in a release write.
\begin{equation*}
	\modzR \asgn \modxA + \modyX \xtra{\loadxA3 \tsep \loadyX2} \modzR \asgn 5 \ttra{\modzR \asgn 5} \Nil
	\labeleqn{instr-eval-example}
\end{equation*}
The first two steps are inherited from the expression evaluation semantics (via \refrule{eval-asgn}).  The final step is via \refrule{eval-indivis}, noting
$\indivisact{\modzR \asgn 5}$.
The individual evaluation of local variables is not affected by the shared memory system
and hence the following incremental execution is also allowed (where we leave implicit \relaxed accesses), 
\begin{equation*}
	z \asgn x + r \ttra{\loadx3} z \asgn 3 + r \ttra{z \asgn 3 + r} \Nil
	\labeleqn{instr-eval-example2}
\end{equation*}

Incremental execution 
makes explicit what a compiler might do with C code involving shared variables.  Note that, importantly, the reordering relation did not have to change or be reconsidered at all,
even though the scope of \Clang considered was increased (and made more complex to reason about, as is always the case when considering incremental evaluation).  

\subsection{Reasoning about incremental evaluation}
\labelsect{incremental-reasoning}

Syntax-directed, general inference rules (\eg, in an Owicki-Gries or rely/guarantee framework) 
are rare for concurrent languages with incremental evaluation of expressions and instructions,
irrespective of weak memory models.
For instance, a simple program like $x \asgn x + 1 \pl x \asgn x + 1$ can result in final states where $x \in \{1,2\}$ when
executed incrementally, but this is not immediately derivable from typical compositional proof rules
\cite{HayesNondetExprEval}.
The problem is that, in general, syntax-directed proof rules do not directly apply as there are many places where interference may occur, which do not neatly align with syntactic terms.
The situation is more complicated again with non-deterministic evaluation order, as specified by the \Clang standard.

From a practical perspective the issue is typically resolved by
making the atomicity explicit (possibly requiring the introduction of new temporary variables), \ie, we may rewrite the above program as follows.
\begin{eqnarray}
	(r_1 \asgn x ; x \asgn r_1 + 1) 
	\pl
	(r_2 \asgn x ; x \asgn r_2 + 1) 
\end{eqnarray}
In this format one can apply standard (non-incrementally evaluated) syntax-directed proof rules, such as Owicki-Gries or rely/guarantee,
and show that $x \in \{1,2\}$ in the final state (possibly requiring the further introduction of auxiliary variables or other techniques). 
Of course specialised rules to handle particular forms of incrementally-evaluated instructions can be derived, and these may be applied in some cases, but in general the intent 
is to precisely deal with communication/behaviour as written in the text of program.

\newcounter{remarkcounter}[section]
\counterwithin{remarkcounter}{section}
\newtheorem{remark}[remarkcounter]{Remark}

\newcommand{\atomicsyntaxstructured}{atomic-syntax-structured\xspace}
\newcommand{\atmsynstruct}{atomic-syntax-structured\xspace}
\newcommand{\Atmsynstruct}{Atomic-syntax-structured\xspace}

As the difficulty of reasoning about possibly dynamically changing code structure 
impacts on reasoning about \Clang programs, especially with reordering and incremental evaluation, we
make this clearer by expressing it in terms of a definition and some remarks.
\begin{definition}[\Atmsynstruct code]
\labeldefn{atmsynstruct}
A command in standard imperative programming syntax, where
all basic building blocks (conditions, assignments) of command $c$ are evaluated/executed atomically, and execution proceeds in program-order,
is \emph{\atmsynstruct} code.
\end{definition}

Note that a subset of \impro can be \atmsynstruct, especially if taking the \emph{plain} subset (\refdefn{plain}) and using the semantics of \refsect{semantics}.

\begin{remark}
Most inference systems for concurrent code work on the basis the code is \atmsynstruct; it is non trivial to apply syntax-based approaches if the syntax does not directly map to execution.
Often the atomicity is made explicit by introducing temporary variables, or a non-syntax based approach is used for verification, \eg, translating into automaton systems where, again, the
atomicity is explicit.
\end{remark}

\begin{remark}
Many other approaches are still applicable to non-\atmsynstruct code, for instance, model checking.
\end{remark}
We emphasise that \Clang programs are in general not \atmsynstruct, and thus complicates analysis by some techniques, regardless of whether or not the \Cmm memory model is taken into account.

It is beyond the scope of this paper to develop rules that handle non-\atmsynstruct code
but, as before, one may apply such rules after reduction.
We argue that the level of granularity should be made explicit, or in other words, programmers (who wish to do analysis) should restrict
themselves to instructions that are directly executable.
For instance, normal assignments that reference as most one shared variable (see \eg, \cite{HayesRGLaws-arXiv} for rules coping with such situations).

If the developer insists on reasoning about code that is not \atmsynstruct then 
some of the reduction rules need provisos under incremental execution.
For instance, \reflaw{2actions-keep-order}
holds only when both instructions are indivisible, \ie, it must be updated to ensure the relevant instructions are $\indivisword$.
\begin{equation*}
	\aca \nro \acb \land \indivisact{\aca, \acb} \ssimp \aca \ppseqc \acb \refeq \aca \bef \acb
\end{equation*}
If either is not indivisible then
there may be parts of $\acb$ that can be incrementally evaluated before $\aca$, for instance consider:
\begin{equation*}
	x \asgn 1 
	\ppseqc
	z \asgn x + y
\end{equation*}
Although
$
	x \asgn 1 
	\nro
	z \asgn x + y
$
due to $x$, the load of $y$ can come before $x \asgn 1$, \ie,
it is not the case that
$
	\blockall{
	x \asgn 1 
	}{
	z \asgn x + y
	}
$
under incremental execution.
The reference to $y$ can be incrementally evaluated and reordered before the store to $x$.
\begin{equation}
	x \asgn 1 
	\ppseqc
	z \asgn x + \underline{y}
	\ttra{\loady3}
	x \asgn 1 
	\ppseqc
	z \asgn x + \underline{3}
	\xtra{x \asgn 1 \tsep \loadx1 \tsep z \asgn 4}
	\Nil
\end{equation}
As with proof methods, specific reduction rules to handle incremental evaluation can be derived (possibly using a program-level encoding of evaluation as given in \cite{SemanticsSRA,FACS19}).

\section{Expression optimisations}
\labelsect{expression-optimisations}

\newcommand{\optop}{\overset{\mathsmaller{\sf{opt}}}{\xsucc}}
\newcommand{\expropt}[2]{#1 \optop #2}
\newcommand{\nexpropt}[2]{#1 \not\optop #2}
\newcommand{\cmdopt}[2]{\expropt{#1}{#2}}
\newcommand{\optwrt}[3]{#3 \imp \expropt{#1}{#2}}
\newcommand{\opteep}{\expropt{e}{e'}}

\newcommand{\lexmore}{\overset{\mathsmaller{\sf{lex}}}{\xsucc}}
\newcommand{\ocmore}{\overset{\mathsmaller{\sf{oc}}}{\succeq}}

\newcommand{\eqnstackrcll}[1]{
	\left\{\begin{array}{rcll}
		#1
	\end{array}\right.
}

\newcommand{\eqnstack}[1]{
	\left\{\begin{array}{ll}
		#1
	\end{array}\right.
}

We now consider a further important factor influencing execution of programs under the \Clang memory model:
expression optimisations
(we consider structural optimisations
for instance, changing loop structures, in \refsect{ctrans}).
There are three principles when considering ``optimising'' expression $e$ to $e'$:

\begin{itemize}
\item \emph{Value equality.}
Expression $e'$ must be
equal to $e$ in all states.
However, as we see below, extra contextual information can be used.
\item
\emph{Lexicographic simplification.}
We say $e \lexmore e'$ if $e'$ is a ``more optimised'' expression than $e$, in the sense that it is less computationally intensive to evaluate.
A precise definition of $\lexmore$ for \Clang is beyond the scope of this work, but we assume
it is irreflexive and transitive, 
and that intuitive properties such as
$3 + 2 \lexmore 5$, and $0 * r \lexmore 0$ hold.
An important aspect is that $\lexmore$ may allow the removal of variables (as in $0 * r \lexmore 0$), and this could have an effect on allowed reorderings according to \refmm{Gmm}, \Gmm
(\ie, one cannot rely on ``false dependencies'' \cite{HerdingCats}).
\item 
\emph{Memory ordering constraint simplification.}
We say $e \ocmore e'$ if $e'$ does not lose any \emph{significant} memory ordering constraints.
For instance, it may be the case that compilers should not ``optimise away'' an explicit $\seqcst$ constraint, even if valid according to the other optimisations.
Again, the precise definition of this is a matter for the \Clangctee, but we explore some of the options and their consequences below in \refsect{ocmore}.
\end{itemize}

These three constraints must be satisfied before an expression $e$ is `optimised' to expression $e'$, written $\opteep$.
\begin{eqnarray}
\labeleqndefn{opteep}
\opteep
\asdef
		e = e'
		\land
		e \lexmore e'
		\land
		e \ocmore e'
\end{eqnarray}

The following operational rule 
allows optimisations as an expression evaluation step,
superseding \refrule{eval-apply} and in some cases removes the need for \refrule{eval-binop}, etc.

\ruledefNamed{0.35\textwidth}{Optimise expression}
	{optimise-expr}{
	\Rule{
		b \entails \opteep
	}{
		e \ttra{\guard{b}} e'
	}
}

\refrule{optimise-expr} 
states that an expression $e$ can be optimised to some expression $e'$, in the process emitting a guard $\guardb$,
where $b$ provides the context which makes the optimisation valid (the expression $b$ is used only to show $e = e'$ in \refeqndefn{opteep}).
The guard acts as a check that the optimisation is valid in the current state; for many optimisations $b$ will simply be $\True$.

As an example, assuming a definition of $\ocmore$ where $e \ocmore e'$ holds provided $e$ contains only relaxed or no ordering constraints,
and $e'$ has a subset of those, then the following steps are allowed by \refrule{optimise-expr}.
\begin{eqnarray*}
	&
	x \asgn r - r \tra{\tau} x \asgn 0
	&
	\mbox{Since $\True \entails r - r = 0 \land r - r \lexmore 0 \land r - r \ocmore 0$}
	\\
	&
	x \asgn r_1 - r_2 \ttra{\guard{r_1 = r_2}} x \asgn 0
	&
	\mbox{Since $r_1 = r_2 \entails r_1 - r_2 = 0 \land r_1 - r_2 \lexmore 0 \land r_1 - r_2 \ocmore 0$}
	\\
	&
	r \asgn \modX{x} * 0 \tra{\tau} r \asgn 0
	&
	\mbox{Since $\True \entails \modxX * 0 = 0 \land \modxX * 0 \lexmore 0 \land \modxX * 0 \ocmore 0$}
\end{eqnarray*}
On the other hand, given the above assumption about $\ocmore$,
	$r \asgn \modxSC * 0$ cannot be optimised to $r \asgn 0$ as this would lose a significant reordering constraint ($\modxSC * 0 \not\ocmore 0$ and hence $\nexpropt{\modxSC * 0}{0}$).
We consider other examples in the subsequent section.

\subsection{Defining allowed changes to ordering constraints}
\labelsect{ocmore}

\newcommand{\asubseteq}{&\subseteq&}

\newcommand{\acs}{\textsc{acs}}

One of the tensions in the development of the \Clang memory model is how accepted compiler optimisations interact with memory ordering constraints.
Rather than take a particular position, or try to be exhaustive, we show how different options can be expressed and their consequences 
enumerated formally and (relatively) straightforwardly.

Consider the following five possible definitions for $\ocmore$.
Recall that $\getocse$ extracts the memory ordering constraints from expression $e$ (\refeqndefn{getocs}).
We abbreviate
$\acs  \sdef \{\acquire, \consume, \seqcst\}$, as these are the significant ordering constraints that may appear in expressions
($\release$ constraints appear only on the left-hand side of assignments and thus are not subject to expression optimisation).
\begin{eqnarray}
	e \ocmore e' 
	\asdef
	\eqnstackrcll{
		\getocs{e'} \aeq \getocse
		&
		\mbox{(a) Do not modify constraints}
		\\
		\getocs{e'} \asubseteq \getocse \sssubseteq \{\relaxed\}
		&
		\mbox{(b) Simplify/eliminate relaxed}
		\\
		\getocse \asubseteq \{\relaxed\}
		&
		\mbox{(c) Strengthening allowed}
		\\
		\getocse \int \acs \aeq \getocsep \int \acs
		& 
		\mbox{(d) Never optimise $\acs$ away} 
		\\
		&
		\True
		& 
		&
		\mbox{(e) Do not constrain the compiler}
	}
	\labeleqn{defn-ocmore}
\end{eqnarray}

\newcommand{\aocmore}{&\ocmore&}
\newcommand{\dunno}{?}

\begin{enumerate}[{Option} (a)]
\item
is the most conservative option and simply says the compiler must not change the constraints in $e$ at all. This would be simple to implement and reason about,
but possibly prevents some sensible/expected optimisations.
\item
says that only relaxed or non-atomic expressions $e$ may be removed, and the optimised expression $e'$ can either remain
relaxed or become non-atomic itself.  Stronger constraints ($\acquire$, $\consume$, and $\seqcst$) 
will not be ``optimised away''.
\item
says that only relaxed or non-atomic expressions can be optimised, however, any such constraints can be strengthened (\eg, a $\relaxed$ access can be strengthen to $\seqcst$).
While more subtle than the other options, it imposes fewer constraints on the compiler writer, and would only have the effect of reducing the
number of behaviours of code.
\item
requires $\acquire$ and $\seqcst$ constraints to be maintained, if they occur, but other parts of a complex expression can be optimised.
\item
allows full freedom to the compiler, leaving the programmer unable to rely entirely on memory ordering constraints to enforce order.
\end{enumerate}

\newcommand{\LEXCOL}[1]{}

We give some examples in tabular form to understand the consequences of these choices.
\begin{equation*}
\begin{array}{r@{~}c@{~}lcccccc}
&& & (a) & (b) & (c) & (d) & (e) & \LEXCOL\lexmore \\
5*4 \aocmore 20 &
	\tick & \tick & \tick & \tick  & \tick & \LEXCOL\tick
\\
\modxX * 0 \aocmore 0 &
	\tick & \cross & \tick & \tick  & \tick & \LEXCOL\tick
\\
\modxSC * 0 \aocmore 0 &
	\cross & \cross & \cross & \cross  & \tick & \LEXCOL\tick
\\
\modxSC \aocmore \modxX &
	\cross & \cross & \cross & \cross  & \tick & \LEXCOL\dunno
\\
\modxX \aocmore \modxSC &
	\cross & \cross & \tick & \cross  & \tick & \LEXCOL\dunno
\end{array}
\end{equation*}
As with incremental evaluation of expressions, while it is straightforward to incorporate optimisations into the semantics, 
programs may not be \atmsynstruct (\refdefn{atmsynstruct}),
and hence
a programmer interested in serious analysis 
is well advised to avoid potential simplification of expressions that a compiler may or may not choose to make.

\subsection{The consume ordering constraint}
\labelsect{consume}
\newcommand{\condep}{\mokeyword{condep}}

Because \Clang's consume ($\consume$) constraint has no reordering constraint beyond that of data dependencies we have for brevity not included it in earlier sections.
The intent of a $\consume$ load is to indicate to the compiler not to lose data dependencies during optimisations.
Options (b)-(e) for $\ocmore$ allow $\relaxed$ variable accesses to be removed, which may also remove a data dependency to some earlier load.
This is the situation that a $\consume$ load is intended to avoid.
For instance, the following optimisation should \emph{not} be allowed.  
\begin{equation*}
	r \asgn \modxC \ppseqc y \asgn r * 0
	\tra{\tau}
	r \asgn \modxC \ppseqc y \asgn 0
\end{equation*}
The flow-on effect is that now $y \asgn 0$ is independent of $r \asgn \modxC$ and may be reordered before it
(as $\consume$ is equivalent to $\relaxed$ in calculating $\roOC$).  
To faithfully implement this a compiler must track data dependencies, and apparently this has never been implemented; as such all known compilers
translate $\consume$ constraints directly to the stronger $\acquire$ constraint (which on many architectures results in a more computationally expensive mechanism than necessary).
Such syntactic tracking is straightforward, if tedious, to implement in a formal setting. 
For instance, by tracing data-dependencies (via write and read variables) from all $\consume$ loads,
marking them with some special ordering constraint `$\condep$' (consume-dependent), and requiring the definition of $\ocmore$ to never allow $\condep$ constraints to be removed.
As always, a programmer is well-advised to minimise the use of $\consume$ constraints with later code that may allow expression optimisations to break data dependencies.

\subsection{Examples}
We show how the compiler may ``optimise-away'' an ordering constraint and hence open up more behaviours than the programmer may expect.
In the following assume Option (b) above for the definition of $\ocmore$.  
A programmer may choose to enforce order in the $\MP$ program (\refeqndefn{mp}) using a data dependency.
\begin{equation*}
	x \asgn 1 \scomp flag \asgn 1 + (x * 0)
\end{equation*}
Although $\datadep{x \asgn 1}{flag \asgn 1 + (x * 0)}$ and thus it seems the assignment to $y$ must occur after the assignment to $x$,
after optimising the second instruction the updates can occur in the reverse order.  That is, by \refrule{optimise-expr} (within \refrule{eval-asgn}),
$flag \asgn 1 + (x * 0) \tra{\tau} flag \asgn 1$, and thus we have the following behaviour.
\begin{equation}
	x \asgn 1 \scomp flag \asgn 1 + (x * 0)
	\tra{\tau}
	x \asgn 1 \scomp flag \asgn 1
	\xtra{flag \asgn 1}
	x \asgn 1 
	\ttra{x \asgn 1}
	\Nil
\end{equation}
As a consequence, for reasoning about the original code, one must accept the following refinement.
\begin{equation*}
	x \asgn 1 \scomp flag \asgn 1 + (x * 0)
	\ssrefsto
	flag \asgn 1 \bef x \asgn 1 
\end{equation*}
Alternatively
a programmer may choose to avoid the dependence on data and instead enforce order using an \seqcst constraint in the flag expression on an unrelated variable.
Under option (e) this would be erroneous, since $\expropt{1 + (\modySC * 0)}{1}$, and thus as above,
\begin{equation*}
	x \asgn 1 \scomp flag \asgn 1 + (\modySC * 0)
	\ssrefsto
	flag \asgn 1 \bef x \asgn 1 
\end{equation*}

\section{Forwarding}
\labelsect{forwarding}

A key aspect of hardware pipelines is that instructions later in the pipeline can read values from earlier instructions (under certain circumstances).  
At the microarchitectural level,
rather than waiting for an earlier instruction to commit and write a value to a register before reading that value from the register,
it may be quicker to 
read an ``in-flight'' value directly from an earlier instruction before it commits.
This behaviour may be implemented by so-called ``reservation stations'' or related mechanisms \cite{Tomasulo67}.
Fortunately there is a
straightforward way to capture this in a structured program, by allowing later instructions that are reordered to pick up and use values in the text of earlier instructions.
For instance, the rules of the earlier section forbid $r \asgn x$ reordering before $x \asgn 1$, \ie, $x \asgn 1 \nro r \asgn x$, which is conceptually prevented because
reading an earlier value of $x$ (than 1) is prohibited by
coherence-per-location (formally, $x \asgn 1 \nro r \asgn x$ because $\wv{x \asgn 1} = \{x\} \subseteq \fv{r \asgn x}$).  
However processors will commonly just use the value 1 and assign it to $r$ as well, before any other process has seen the write to $x$.  

In hardware the mechanism of using earlier values in later instructions is called \emph{forwarding}.
Notationally we write
$\fwdab$ to mean the effect of forwarding (assignment) action $\aca$ to $\acb$,
which is just simple substitution (ignoring memory ordering constraints).
For instance, in the above situation the value 1 assigned to $x$ can be ``forwarded'' to $r$, written
$(\fwd{x \asgn 1}{r \asgn x}) = r \asgn 1$, avoiding the need to access main memory.
\begin{definition}[Forwarding]
\labeldefn{forwarding}
Given an assignment instruction $\aca$ of the form $\modxocs \asgn e$, 
forwarding of $\aca$ to an expression $f$ (written $\fwd{\aca}{f}$) is standard replacement of
instances of $x$ by $e$ within $f$, ignoring ordering constraints.
This is lifted to forwarding to instructions as below.
\begin{eqnarray*}
    \fwda{(\modocs{y} \asgn f)} \sseq \modocs{y} \asgn (\fwda{f})
	\qquad
    \fwda{\guardb} \aeq \guard{\fwda{b}}
	\qquad
    \fwda{\fencepf} \sseq \fencepf
\end{eqnarray*}
\end{definition}
Forwarding to/from commands and traces is similarly straightforward; see \refappendix{lifting-forwarding}.
Note that $\nddepab \imp \fwdab = \acb$, \ie, forwarding is relevant only if there is a data dependence.
The following examples show the relatively straightforward application of forwarding.
\begin{gather}
	\fwd{r_1 \asgn 1}{\guard{r_1 = 1}} = \guard{1 = 1} = \guard{\True} = \tau
	\notag
	\\
	\fwd{r_1 \asgn r_2}{\guard{r_1 = r_2}} = \guard{r_2 = r_2} = \guard{\True} = \tau
	\labeleqn{fwd-r1r2-g}
\end{gather}

To capture the potential for an earlier instruction to affect a later one
we generalise reordering to a triple,
$\robpab$.
Whereas earlier we considered whether $\aca \ro \acb$ directly, we now allow $\aca$ to affect $\acb$ via forwarding, 
and for the resulting instruction ($\acb'$) to be considered for the purposes of calculating reordering.
More precisely, for instructions $\aca$ and $\acb$, 
\begin{equation}
	\robpab
	\sspace
	\sdef
	\sspace
	\aca \ro \acb' 
	\land
	\acb' = \fwdab
	\land 
	\acb \ocmore \acb'
	\labeleqndefn{roa}
\end{equation}
The reordering triple notationally gives the idea of executing $\acb$ earlier with respect to $\aca$, with possible modifications due to forwarding.
The new instruction $\acb'$ is the result of forwarding $\aca$ to $\acb$,
and it must be reorderable with $\aca$ (note, therefore, that reordering is calculated \emph{after} applying forwarding). 
For instance, 
$\rocmdg{r \asgn 1}{x \asgn 1}{r \asgn x}$ expresses that $r \asgn x$ can reorder with $x \asgn 1$, after forwarding is taken into account.
Additionally the effect of forwarding must not have significantly altered the 
ordering constraints, \ie, $\acb \ocmore \fwdab$ (recall the options in \refsect{ocmore}).  
For instance, depending on the definition of $\ocmore$, it may or may not be the case that
$
	\rocmdg{r \asgn 1}{x \asgn 1}{r \asgn \modxSC}
$, 
as this depends on whether 
$
	\modxSC \ocmore 1
$.
The reordering triple lifts to traces and commands straightforwardly; see \refappendix{lifting-triples}.

We use this more general triple in place of binary reordering 
in \refrule{pseqcA}.
\\
\ruledefNamed{0.55\textwidth}{Reorder with forwarding}{ro-fwd}{
	\Rule{
        c_2 \tra{\acb} c_2'
        \qquad
        \rocmdm{\acb'}{c_1}{\acb}
    }{
        c_1 \ppseqm c_2 \tra{\acb'} c_1 \ppseqm c_2'
    }
}
\\
Now, given command $c_1 \ppseqm c_2$, and that $c_2$ can execute step $\acb$, then the composition can execute the \emph{modified} $\acb'$, 
where $\acb'$ takes in to account any forwarding that may occur from instructions in $c_1$ to $\acb$.

\labelsect{fwd-example}

Consider the simple statement $x \asgn r_1 - r_2$, which can be optimised to $x \asgn 0$ provided $r_1 = r_2$.  That is, by \refrule{optimise-expr},
$
	x \asgn r_1 - r_2
	\ttra{\guard{r_1 = r_2}}
	x \asgn 0
$.
Let us consider this statement immediately following the assignment $r_1 \asgn r_2$.
By \refdefn{forwarding}
and by \refeqn{fwd-r1r2-g}
we have
$
	\roc{\tau}{r_1 \asgn r_2}{\guard{r_1 = r_2}}
$.
Hence by applying \refrule{ro-fwd}
this program can take an initial silent step, representing an optimisation of the compiler, to simplify the assignment to $x$.
\begin{equation}
	r_1 \asgn r_2 \ppseqc x \asgn r_1 - r_2
	\tra{\tau}
	r_1 \asgn r_2 \ppseqc x \asgn 0
\end{equation}
From here the assignment to $x$ can proceed first, despite the data dependence $\datadep{r_1 \asgn r_2}{x \asgn r_1 - r_2}$,
that is,
\begin{equation*}
	r_1 \asgn r_2 \ppseqc x \asgn r_1 - r_2
	\refsto
	x \asgn 0 \bef r_1 \asgn r_2 
\end{equation*}

\OMIT{
We show some typical examples below, covering the pattern of a store to $x$ followed by a load of $x$: $\modx{\oc_1} \asgn v \ppseqc r \asgn \modx{\oc_2}$,
where we consider different combinations of $\oc_1$ and $\oc_2$.
Note that there are several cases we do not need to consider:
by well-formedness $\acquire \neq \oc_1$ and $\release \neq \oc_2$.
\begin{minipage}{0.31\textwidth}
\begin{gather*}
	\rocmdg{r \asgn 1}{\modxX \asgn 1}{r \asgn \modxX}
	\\
	\rocmdg{r \asgn 1}{\modxX \asgn 1}{r \asgn \modxA}
	\\
	\nrocmdg{r \asgn 1}{\modxX \asgn 1}{r \asgn \modxSC}
\end{gather*}
\end{minipage}
\begin{minipage}{0.31\textwidth}
\begin{gather*}
	\rocmdg{r \asgn 1}{\modxR \asgn 1}{r \asgn \modxX}
	\\
	\rocmdg{r \asgn 1}{\modxR \asgn 1}{r \asgn \modxA}
	\\
	\nrocmdg{r \asgn 1}{\modxR \asgn 1}{r \asgn \modxSC}
\end{gather*}
\end{minipage}
\begin{minipage}{0.31\textwidth}
\begin{gather*}
	\rocmdg{r \asgn 1}{\modxSC \asgn 1}{r \asgn \modxX}
	\\
	\rocmdg{r \asgn 1}{\modxSC \asgn 1}{r \asgn \modxA}
	\\
	\rocmdg{r \asgn 1}{\modxSC \asgn 1}{r \asgn \modxSC}
\end{gather*}
\end{minipage}
In the final case we have allowed an sc write to be forwarded to an sc load.  This controversial but justified, depending on whether you think
sc should block register accesses...
}

\subsection{Reduction with forwarding}
\labelsect{reduction-fwd}

We update some of the reduction rules from \refsect{reduction} to include forwarding via \refrule{ro-fwd},
essentially replacing $\aca \ro \acb$ with the more general $\roacbab$.
We assume $\indivisa$ and $\indivisb$.
Note that if $\nddep{\aca}{\acb}$ then $\roacbab \iff (\aca \ro \acb \land \acb' = \acb)$,
and hence in many cases the original rules apply.
\begin{eqnarray}
    \roacbab 
    \sspace \entails \sspace
    \aca \ppseqc \acb 
    &\refsto&
    \acb' \bef \aca
\labellaw{2actions-swap-order-fwd}
    \\
    \roacbab 
    \sspace \entails \sspace
    \aca \ppseqc \acb 
    & \refeq &
    (\aca \scomp \acb) \choice (\acb' \scomp \aca)
\labellaw{2actions-reduce-fwd}
\\
    \roaCbcb 
    \sspace \entails \sspace
    \cmdc \ppseqc \acb 
    &\refsto&
    \acb' \bef \cmdc
\labellaw{bcb}
    \\
	\aca \nro \fwdab
    \sspace \entails \sspace
    \aca \ppseqc \acb 
	\arefeq
	\aca \bef \acb
\labellaw{2actions-reduce-nofwd}
\end{eqnarray}
\reflaw{2actions-swap-order-fwd} 
expresses the reordering of $\acb$ earlier than $\aca$ but with any forwarding
from $\aca$ to $\acb$ taken into account in the promoted instruction,
while
\reflaw{2actions-reduce-fwd} is the corollary of \reflaw{2actions-reduce}.
\reflaw{bcb} is the generalisation of \reflaw{2actions-swap-order-fwd} to a command on the left.
\reflaw{2actions-reduce-nofwd} applies when reordering is not possible even taking into account the effects of forwarding,
replacing \reflaw{2actions-keep-order}.
The presence of forwarding therefore complicates reduction, however the derived properties for reasoning (\refsect{reasoning}) still apply to
any reduced program.

A subtlety of a chain of dependent instructions with forwarding is that associativity can be lost.
For instance, consider the program $x \asgn 1 \ppseqc r_1 \asgn x \ppseqc r_2 \asgn x$.
The behaviours are different depending on how this is bracketed:
$
	p_1 \sdef	(x \asgn 1 \ppseqc r_1 \asgn x) \ppseqc r_2 \asgn x
$
or
$
	p_2 \sdef	x \asgn 1 \ppseqc (r_1 \asgn x \ppseqc r_2 \asgn x)
$.
We have 
$p_1 \refsto r_2 \asgn 1 \bef r_1 \asgn 1 \bef x \asgn 1$ 
by \reflaw{bcb}, however 
$p_2 \nrefsto r_2 \asgn 1 \bef r_1 \asgn 1 \bef x \asgn 1$ 
because $r_1 \asgn x \nro r_2 \asgn x$.
Hence $p_1 \nrefeq p_2$, and so associativity has been broken.
(A more complete discussion of associativity and monotonicity is given in \cite{pseqc-arXiv}.)

The addition of forwarding also explains why a compiler transformation such as ``sequentialisation'' \cite{PromisingSemantics} is not valid.
While it is straightforward that $c \pl d \refsto c \bef d$, \ie, enforcing a strict order between two parallel processes, 
it is not the case in general that $c \pl d \refsto c \ppseqc d$, due to forwarding.
The simplest case is $x \asgn 1 \pl r \asgn x$, which has exactly two possible traces (interleavings), however 
$x \asgn 1 \ppseqc r \asgn x \refsto r \asgn 1 \bef x \asgn 1$
by \reflaw{2actions-swap-order-fwd}, where $r$ receives the value 1 before it is assigned to $x$, which is not possible when executing in parallel.

\subsection{Example}


Forwarding admits some perhaps unexpected behaviours, for instance, in the following code, although $r \asgn x$ and $x \asgn 1$ are
strictly ordered locally, it is possible for $r \asgn x$ to receive that \emph{later} value, ostensibly breaking local coherence.
\begin{theorem}
$
	\reachtrip{\vinit{x,y,r}}{
	(r \asgn x \ppseqc x \asgn 1 \ppseqc y \asgn x) 
	\pl 
	x \asgn y 
	}{
		r = 1
	}
$
\end{theorem}
\begin{proof}
Call the left and right programs $p_1$ and $p_2$, respectively.
Note that in $p_1$ although $r \asgn x$ precedes the assignment $x \asgn 1$ (the only occurrence of the value 1 in the program),
and $x \asgn 1$ cannot be reordered before $r \asgn x$, this theorem states that $r \asgn x$ can read that value.
Firstly note that  $\rocmdg{y \asgn 1}{r \asgn x \ppseqc x \asgn 1}{y \asgn x}$,
that is, $y \asgn x$ can read $x \asgn 1$ and be reordered before the preceding instructions,
and hence by \reflaw{bcb} we have
\begin{equation*}
	p_1 
		\refsto
	y \asgn 1 \bef r \asgn x \bef x \asgn 1 
\end{equation*}
Hence interleaving $p_1$ and $p_2$ so that $x \asgn y$ in $p_2$ becomes the second instruction executed gives the following.
\begin{equation*}
	p_1 \pl p_2 \refsto
	y \asgn 1 \bef x \asgn y \bef r \asgn x \bef x \asgn 1 
\end{equation*}
This reordering and interleaving satisfies 
$
	\htrip{\vinitxyr}{
		y \asgn 1 \bef x \asgn y \bef r \asgn x \bef x \asgn 1 
	}{
		r = 1
	}
	$,
and thus the outcome $r = 1$ is a possible final state of $p_1 \pl p_2$ by \refth{reasoning-possible}.
\end{proof}
It appears that coherence has been broken (indirectly through a concurrent process).
However, the intuition is that the compiler (and/or the processor) decides that the programmer intended $y \asgn 1$ 
when they wrote $y \asgn x$, and it is \emph{this} value that is read by $r \asgn x$ via the second process.  To make this clearer, consider changing 
the trailing assignment $y \asgn x$ in the first process to $y \asgn x + 1$.  In this case a possible final state is $r = 2$ (but not $r = 1$), meaning that
the initial load of $x$ has read the (arguably independent) write to $y$, not the write to $x$.

Forwarding of assignments is considered standard, and behaviours such as that shown by this example are accepted (in real code,
typically a programmer will try to avoid loading a variable that has already been locally calculated, so this pattern, while certainly valid and
has its uses, is not necessarily widespread).  However, the situation is more complicated if one wishes to forward \emph{guards} as well as assignments to later instructions,
which we explore in \refsect{sfp}.

\subsection{Combining forwarding with optimisation/simplification}

To keep the discussion of forwarding relatively simple we separated it from optimisation, but we can generalise
\refeqndefn{roa} to 
incorporate optimisations as well.
\begin{eqnarray}
\labeleqndefn{roa-opt}
	\roabab
	\asdef
	\aca \ro \acb' \land \cmdopt{\acb'}{\fwdab} 
\end{eqnarray}
This definition allows any ``optimised'' version of $\acb$ to be reordered before $\aca$, provided forwarding is taken into account.
It can become the basis for
the application of \refrule{ro-fwd}, and thus the derived reduction rules in \refsect{reduction-fwd}.
For instance, in comparison with the example in \refsect{fwd-example}, we can reorder, forward and optimise in a single step since,
by \refeqndefn{roa-opt},
$
	\rocmdg{x \asgn 0}{r_1 \asgn r_2}{x \asgn r_1 - r_2}
$.
The corresponding deduction rules allow us to show, for instance, the following allowed reordering and simplification of code by the compiler.
\begin{equation}
	x \asgn y \scomp z \asgn x - y 
	\refsto
	z \asgn 0 \bef x \asgn y 
\end{equation}
because
$
	\rocmdg{\tau}{x \asgn y}{\guard{x = y}}
$.
This answers the question ``can $z \asgn x - y$ be reordered before $x \asgn y$?'': provided all references are relaxed 
then $z \asgn 0$ can be executed earlier than the update to $x$.
Of course, structured reasoning about code which is susceptible to such compiler transformations may be nontrivial as it is not \atmsynstruct.

\section{Read-from-untaken-branch (RFUB)/Self-fulfilling prophecies (SFPs)}
\labelsect{sfp}

In this section we separately consider a problematic situation which is debated by the \Clang committee,
where what is considered  reasonable compiler optimisation leads to complex behaviour which is difficult to reason about.
We show how the problematic behaviour is disallowed by the semantics we have given so far, and a small modification 
-- which we call allowing \emph{self-fulfilling prophecies} (SFP) -- can be made
which then allows the problematic behaviour; and we also show that allowing SFPs contradicts other, simpler cases which 
are expressly forbidden.  We believe this provides a firm, accessible basis on which to assess which compiler optimisations should be allowed
in the presence of shared-variable concurrency (\ie, \Clang's \T{atomics}).
Importantly the framework gives a step-based, relatively intuitive, explanation for the different possible behaviours,
and flexibility to accommodate different decisions.

\newcommand{\RFUB}{{\sf{rfub}}\xspace}
\newcommand{\ctrans}{\mathrel{\overlay{$\fun$}{$\hookrightarrow$}}} 
\newcommand{\actrans}{&\ctrans&} 

\newcommand{\sfprefsto}{\overset{\mathsmaller{sfp}}{\refsto}} 
\newcommand{\asfprefsto}{&\sfprefsto&} 

\subsection{Read-from-untaken-branch behaviour}
\labelsect{rfub}

Consider the following program, which, for particular compiler optimisations, exposes a ``read from untaken branch'' (rfub) behaviour.
\begin{eqnarray}
\labeleqndefn{rfub}
	\RFUB
	\asdef
	(r \asgn y ;
	(\IFsC{r \neq 42}{b \asgn \True ; r \asgn 42}) ;
	x \asgn r)
	\sspace
	\pl
	\sspace
	y \asgn x
\end{eqnarray}

Taking a plain, sequential execution of $\RFUB$, starting with $\neg b$ and all other variables 0,
then both $r$ and $x$ are 42 in the final state, and $y$ 
is either 0 or 42 depending on when the right-hand process interleaves.  The $\True$ branch of the
conditional is always executed.
\begin{theorem}
\labelth{rfub-plain}
\begin{align*}
	\htrip{\vinit{x,y,r} \land \neg b}
	{
	&
		\plain{\RFUB}
	}
	{x = 42 \land r = 42 \land b \land (y = 0 \lor y = 42)}
\\
\mbox{
and hence
}
\quad
	\neg \reachtrip{\vinit{x,y,r} \land \neg b}{&\plain{\RFUB}}{x = y = r = 42 \land \neg b}
\end{align*}
\end{theorem}
\begin{proof}
Straightforward: the only assignment of 42 is within the conditional, \emph{after} the read of $y$, and hence at that point $y$ cannot be 42.
The condition $r \neq 42$ always holds and the $\True$ path is taken, setting $r$ to 42 and $b$ to $\True$, and finally setting $x$ to 42.
The final value of $y$ depends on whether the assignment $y \asgn x$ occurs before or after the assignment to $x$.
By corollary it is therefore not possible to reach a final state where all variables are 42 and $\neg b$.
\end{proof}

\newcommand{\initr}{r_{init}}

Given this, in a sequential setting
a compiler may choose to optimise the conditional as follows, letting `$\ctrans$' stand for a compiler transformation.
\begin{equation}
\labeleqn{rfub-ctrans}
	(\IFsC{r \neq 42}{b \asgn \True ; r \asgn 42})
	\ctrans
	b \asgn (r \neq 42) ; r \asgn 42
\end{equation}
The conditional is eliminated, and now $b$ will be $\True$ if $r \neq 42$ when the branch would have been entered, and $r = 42$ in the final
state.
More precisely, if we let $cond$ be the original conditional and $cond'$ be the optimised version,
they both satisfy the following assuming $\neg b$ initially, and $\initr$ is the initival value of $r$.
\begin{equation*}
	\htrip{\neg b}{cond}{r = 42 \land (b \iff \initr \neq 42))}
	\quad and \quad
	\htrip{\neg b}{cond'}{r = 42 \land (b \iff \initr \neq 42))}
\end{equation*}

The transformed code preserves the expected behaviour of $\RFUB$.
\begin{theorem}
\labelth{rfub'-plain}
Let $\RFUB'$ be $\RFUB$ using the transformed conditional from \refeqn{rfub-ctrans}.  Then
\begin{equation*}
    \neg \reachtrip{\vinit{x,y,r} \land \neg b}{\plain{\RFUB'}}{x = y = r = 42 \land \neg b}
\end{equation*}
\end{theorem}
\begin{proof}
Straightforward reasoning as in \refth{rfub-plain}.
\end{proof}

If we consider reordering according to the \Cmm model and semantics we have given so far, the above state is still not reachable.
\begin{theorem}
\labelth{rfub-c11}
\begin{equation*}
    \neg \reachtrip{\vinit{x,y,r} \land \neg b}{\RFUB}{x = y = r = 42 \land \neg b}
\end{equation*}
\end{theorem}
\begin{proof}
More behaviours are possible, but it is still
straightforward in our step-based semantics: if the $\True$ branch is never taken there is no way for the value 42 to be forwarded to
$x$ and so $b = \True$.
Making the paths explicit (\refeqndefn{if}) the left process of \RFUB is equal to:
\begin{equation*}
\labeleqn{rfub-left-expanded}
	(r \asgn y \ppseqc
	\guard{r \neq 42} \ppseqc b \asgn \True \ppseqc r \asgn 42 \ppseqc
	x \asgn r)
	\sspace
	\choice
	\sspace
	(r \asgn y \ppseqc
	\guard{r = 42} \ppseqc
	x \asgn r)
\end{equation*}
The first case allows $x = y = r = 42$, following reasoning from \refth{rfub-plain}, but clearly $b$ holds in any such final state.
In the second case $x \asgn r$ is blocked by $r \asgn y$ (but not by $\guard{r = 42}$), and since there are no (out-of-thin-air) assignments
of 42 to $y$ in the concurrent process the guard will never be $\True$, and so \emph{this branch can never be taken}.
Although reordering $x \asgn 42$ before the guard is possible in the first branch (due to forwarding/optimisations),
the guard $\guard{r \neq 42}$ stays in the program as an ``oracle'', preventing inconsistent behaviour.
The compiler/hardware can allow the store to go ahead,
but it \emph{must} be checked for validity later.
The assignment to $x$ depends on $r$ which, in the $\False$ branch, depends on $y$, and the only way that $x$ can receive 42 in that case is if $y$ receives 42, but this is not
possible as the only instance of 42 is in the $\True$ branch, which is already ruled out.
\end{proof}

However, the compiler transformation makes the state reachable.
\begin{theorem}
\labelth{rfub'-c11}
Recalling $\RFUB'$ is $\RFUB$ using the transformed conditional in \refeqn{rfub-ctrans},
\begin{equation*}
    \reachtrip{\vinit{x,y,r} \land \neg b}{\RFUB'}{x = y = r = 42 \land \neg b}
\end{equation*}
\end{theorem}
\begin{proof}
In the left-hand process of $\RFUB'$ the assignment $x \asgn r$ can be reordered to be the first instruction executed, which by forwarding becomes $x \asgn 42$.
More formally, by \reflaw{bcb} we get the following.
\begin{eqnarray}
	r \asgn y \ppseqc
	b \asgn (r \neq 42) \ppseqc
	r \asgn 42 \ppseqc
	x \asgn r
	\ssrefsto
	x \asgn 42 \bef
	(r \asgn y \ppseqc
	b \asgn (r \neq 42) \ppseqc
	r \asgn 42)
\end{eqnarray}
Now consider the effect of interleaving $y \asgn x$ from the second process immediately after this assignment.
\begin{equation*}
	\htrip{\vinitxyr \land \neg b}{
		x \asgn 42 \bef
		y \asgn x \bef
		r \asgn y \bef
		b \asgn (r \neq 42) \bef
		r \asgn 42 
	}{
		x = y = r = 42 \land \neg b
	}
\end{equation*}
Because $r \asgn y$ reads 42 the value assigned to $b$ is $\False$, and thus $\neg b$ holds in the final state.
The proof is completed by \refth{reasoning-possible}.
\end{proof}
The difficulty with this result is that, since $\neg b$ in the final state, the $\True$ branch of the original conditional from $\RFUB$
could not have been taken (since it contains $b \asgn \True$), 
therefore the $\False$ branch was taken, \ie, the condition $r \neq 42$ fails, and hence $r = 42$.  
But the only way for $r = 42$ to hold is if the $\True$ branch is/was executed, containing $r \asgn 42$,
but this has already been ruled out.

Since $\RFUB'$ has a behaviour which $\RFUB$ does not, it cannot be the case that $\RFUB \refsto \RFUB'$, which suggests the compiler transformation
is not valid.
Essentially,
the transformation eliminates a guard and hence breaks a dependency-cycle check.

The question remains, however, whether the compiler transformation \refeqn{rfub-ctrans} is reasonable in the presence of the $\Cmm$ memory model,
and if so what the implications for the memory model are.
We show how we can tweak the framework (specifically, increasing the circumstances under which the concept of forwarding can apply) to allow 
the possible state using the original version of $\RFUB$ (thus justifying the compiler transformation).  However, there are other consequences.
We can straightforwardly accommodate either outcome in this framework, and can do so with a clear choice that has enumerable consequences to other code: can guards be used to simplify expressions?

\subsection{Self-fulfilling prophecies (forwarding guards)}
\newcommand{\gfwd}[2]{{}_{#1\mbox{\raisebox{-0.5pt}{\guillemetright\!\guillemetright}}}#2}

\newcommand{\gfwdab}{\gfwd{\aca}{\acb}}

In \refdefn{forwarding} we defined $\fwdab$, for $\aca$ an assignment, to update expressions in $\acb$ according to the assignment.
If $\aca$ is a guard (or fence) then $\fwdab = \acb$.
However, it is reasonable, in the sequential world at least, to allow guards to assist in transformations.

We introduce an extended version of forwarding, where $\gfwdab$ returns a \emph{set} of possible outcomes. 
\begin{eqnarray}
	\gfwd{x \asgn e}{\aca} \aeq \{\fwd{x \asgn e}{\aca}\}
	\labeleqn{gfwd-au}
	\\
	\gfwd{\guardb}{\guard{e}} \aeq \{\guard{e'} | b \imp (e' \iff e)\} 
	\labeleqn{gfwd-gg}
	\\
	\gfwd{\guardb}{x \asgn e} \aeq \{x \asgn e' | b \imp (e' = e) \} 
	\labeleqn{gfwd-ga}
\end{eqnarray}
If $\aca$ is an assignment it returns a singleton set containing just
$\fwdab$ \refeqn{gfwd-au}, \eg, $\gfwd{r \asgn 42}{x \asgn r} = \{x \asgn 42\}$. 
However, if $\aca$ is a guard $\guardb$ then $b$ can be used as context to modify (optimise) $\acb$,
\eg, 
$\gfwd{\guard{r = 42}}{x \asgn r} = \{x \asgn 42, \ldots\}$ by \refeqns{gfwd-gg}{gfwd-ga}.

We modify the reordering triple \refeqndefn{roa} to accommodate $\gfwdab$ in place of $\fwdab$.
\begin{eqnarray}
	\roabab 
	\asdef
	\aca \ro \acb' \land \acb' \in \gfwd{\aca}{\acb} \land 
	\acb \ocmore \acb'
	\labeleqndefn{gfwd-roa}
\end{eqnarray}
Thus we incorporate using the guards in conditionals to justify reordering via \refrule{ro-fwd}.
For instance, following from above,
$
	\rocmdg{x \asgn 42}{\guard{r = 42}}{x \asgn r}
$.

The use of \refeqndefn{gfwd-roa} for the reordering triple used in \refrule{ro-fwd}, and the derived laws such as \reflaw{bcb}, allows behaviours that weren't possible before,
for instance,
\begin{eqnarray*}
	\guard{r = 42} \ppseqc x \asgn r
	\arefsto
	x \asgn 42 \bef \guard{r = 42} 
\end{eqnarray*}
Since we also have
\begin{eqnarray*}
		(\guard{r \neq 42} \ppseqc b \asgn \True \ppseqc r \asgn 42 ) \ppseqc x \asgn r
		\arefsto
		x \asgn 42 \bef (\guard{r \neq 42} \ppseqc b \asgn \True \ppseqc r \asgn 42 )
\end{eqnarray*}
(which holds under \refeqndefn{roa} as well as \refeqndefn{gfwd-roa})
we can show that a trailing assignment can use information from a preceding conditional for simplification.
\begin{equation}
\labeleqn{rfub-if-reorder}
	(\IFsC{r \neq 42}{b \asgn \True \ppseqc r \asgn 42}) \ppseqc x \asgn r
	\ssrefsto
	x \asgn 42 \bef (\IFsC{r \neq 42}{b \asgn \True ; r \asgn 42}) 
\end{equation}
The trailing assignment $x \asgn r$ can be executed first, using the value 42 for $r$, since both branches imply
that will always be the value assigned to $x$.
We call this use of a guard to simplify a later assignment as a self-fulfilling prophecy, for reasons which will become clear below.

\subsection{Behaviour of read-from-untaken-branch with self-fulfilling prophecies}

We now return to
the behaviour of $\RFUB$, this time allowing self-fulfilling prophecies via forwarding (\refeqndefn{gfwd-roa}).
To highlight the impact this has we note refinement steps that are allowed under SFPs (using \refeqndefn{gfwd-roa} for the antecedent of \refrule{ro-fwd})
using $\sfprefsto$, and steps that are allowed normally (using \refeqndefn{roa} for the antecedent of \refrule{ro-fwd}) as $\refsto$.
The following theorem stands in contradiction to \refth{rfub-c11}.
\begin{theorem}
Allowing self-fulfilling prophecies,
\begin{equation}
	\reachtrip{\vinit{x,y,r} \land \neg b}{
	\RFUB
	}{
		x = y = r = 42 \land \neg b
	}
\end{equation}
\end{theorem}
\begin{proof}
Let $\RFUB_1$ (resp. $\RFUB_2$) refer to the first (resp. second) process of $\RFUB$ (\refeqndefn{rfub}).
\begin{eqnarray}
	\RFUB_1
	\asdef
	r \asgn y \ppseqc
	(\IFsC{r \neq 42}{b \asgn \True \ppseqc r \asgn 42})  \ppseqc
	x \asgn 42 
	\\
	\asfprefsto
	x \asgn 42 \bef
	r \asgn y \ppseqc
	(\IFsC{r \neq 42}{b \asgn \True \ppseqc r \asgn 42}) 
	\quad \mbox{by \refeqn{rfub-if-reorder}}
	\\
	&\refsto&
	x \asgn 42 \bef
	r \asgn y \bef
	\guard{r = 42} 
	\quad \mbox{by \refeqndefn{if}, \reflaw{chooseL}}
\end{eqnarray}
Now we fix the interleaving.
\begin{eqnarray}
	\RFUB
	\sspace
	\sfprefsto
	\sspace
	(x \asgn 42 \bef
	r \asgn y \bef
	\guard{r = 42} )
	\pl
	y \asgn x 
	\sspace
	\refsto
	\sspace
	x \asgn 42 \bef
	y \asgn x \bef
	r \asgn y \bef
	\guard{r = 42} 
\end{eqnarray}
This interleaving reaches the specified state.
\begin{equation}
	\htrip{\vinit{x,y,r} \land \neg b}{
	x \asgn 42 \bef
	y \asgn x \bef
	r \asgn y \bef
	\guard{r = 42} 
	}{
		x = y = r = 42 \land \neg b
	}
\end{equation}
The proof is completed by \refth{reasoning-possible}.
\end{proof}

This theorem shows that the debated outcome of \RFUB, which is possible under a seemingly reasonable compiler transformation,
arises naturally if conditionals are allowed to modify future instructions.  Indeed, this is essentially what is used to justify the
transformation itself: if the $\False$ branch is taken then already $r = 42$, so any later use of $r$ can be assumed to use the value 42.
However, in a concurrent, reordering setting, there may be unexpected consequences.%
\footnote{
Forbidding branches from simplifying later calculations 
does not prevent all reasonable compiler optimisations, for instance,
\begin{eqnarray*}
	(\IFC{b}{\ldots \ppseqc r \asgn 42}{\ldots \ppseqc r \asgn 42}) \ppseqc
	x \asgn r
	\arefsto
	x \asgn 42 \bef
	(\IFC{r \neq 42}{\ldots}{\ldots}) \ppseqc r \asgn 42
\end{eqnarray*}
This is because there is a definite assignment to $r$ in both branches.
}

We do not intend to take a stand about what is ultimately the best choice for \Clang; rather we show that whichever 
choice is taken can be linked to clear principles about whether or not guards should be allowed to simplify assignments in a concurrent
setting.
Note that we can explain it \emph{in a step-based semantics}.

\subsection{Behaviour of out-of-thin-air with self-fulfilling prophecies}

Recall the out-of-thin-air example $\ootaD$ (\refeqndefn{ootaD}).  
\begin{equation}
	\ootaD
	\sssdef
	r_1 \asgn x \ppseqc (\IFsC{ r_1 = 42 }{ y \asgn r_1}) 
	~~\pl~~
	r_2 \asgn y \ppseqc (\IFsC{ r_2 = 42 }{ x \asgn r_2}) 
\end{equation}
Its behaviour, which is intended to be forbidden, is allowed under SFPs, in contrast to \refth{ootaD-c11}.
\begin{theorem}
Under SFPs,
$\reachtrip{\vinitxyrr}{\ootaD}{x = 42 \land y = 42}$.
\end{theorem}
\begin{proof}
Focus on the left-hand process, and in particular the behaviour when the $\True$ branch is taken.
The key step depends on $\rocmdg{y \asgn 42}{\guard{r_1 = 42}}{y \asgn r_1}$, that is, the guard can be used to simplify the inner assignment.
\begin{eqnarray}
	r_1 \asgn x \ppseqc (\IFC{ r_1 = 42 }{ y \asgn r_1}) 
	\arefsto
	r_1 \asgn x \ppseqc \guard{r_1 = 42} \ppseqc y \asgn r_1
	\\
	\asfprefsto
	r_1 \asgn x \ppseqc y \asgn 42 \bef \guard{r_1 = 42} 
	\\
	\arefsto
	y \asgn 42 \bef r_1 \asgn x \bef \guard{r_1 = 42} 
\end{eqnarray}
Using symmetric reasoning in the second process and interleaving we have the following behaviour.
\begin{equation}
	\ootaD
	\sfprefsto
	x \asgn 42 
	\bef
	y \asgn 42 
	\bef 
	r_1 \asgn x 
	\bef 
	r_2 \asgn y 
	\bef 
	\guard{r_1 = 42} 
	\bef 
	\guard{r_2 = 42} 
\end{equation}
Here both stores have been promoted to the first instruction executed, which are then read into local registers and subsequently satisfy the
guards (a ``satisfaction cycle'' \cite{LibraryAbstractionsForC}).  The postcondition is straightforwardly reached.
\begin{equation}
	\htrip{\vinitxyrr}{
	x \asgn 42 
	\bef
	y \asgn 42 
	\bef 
	r_1 \asgn x 
	\bef 
	r_2 \asgn y 
	\bef 
	\guard{r_1 = 42} 
	\bef 
	\guard{r_2 = 42} 
	}{x = 42 \land y = 42}
\end{equation}
The proof is completed by \refth{reasoning-possible}.
\end{proof}

This demonstrates the problematic nature of allowing SFPs, and unifies the known underlying problem with these two related patterns ($\oota$ and \RFUB).

\section{Further extensions}

In this section we discuss some other extensions to the model that may be of interest in some domains; however we emphasise that the definition of the memory model
does not need to change, we simply make the language and its execution model richer.

\subsection{Power and non-multicopy atomicity}
\labelsect{power}
IBM's multicore Power architecture 
\cite{TutorialARMandPOWER,UnderstandingPOWER,AxiomaticPower}
(one of the first commercially available) 
has processor pipeline reorderings similar to Arm \cite{ModellingARMv8},
but in addition
has a cache coherence system that provides weaker guarantees than that of Arm (and x86): 
(cached) memory operations can appear in different orders to different processes.  
For instance, although on ARM the instructions $x \asgn 1$ and $y \asgn 1$ may be reordered in the pipeline of a particular
processor, whichever order they occur in will be the same for every other process in the system.  
On Power, however, one process may see them in one order and another may see the modification in the reverse order (assuming no fences are used).
This extra layer of complexity is introduced through cache interactions which allow one processor that shares a cache with another to read writes early, before the memory operation is fully committed to
central storage \cite{UnderstandingPOWER}.  

Programmers targetting implementations on Power architectures must take these extra behaviours into account (\Clang accommodates the lack of multicopy atomicity to be compatible with this
hardware).
A formalisation of the Power cache system, separate from the processor-level reordering, and compatible with the framework in this paper,
is given in \cite{FM18}, which is based on an earlier 
formalisation of the Power memory subsystem given in \cite{UnderstandingPOWER}.
The cache subsystem, and its formalisation in \cite{FM18} sits conceptually above the processes within the system, 
and stores and loads interact with it rather than with the traditional variable-to-value mapping.
As such, local reorderings, and the influence of ordering constraints and fences are still relevant.
However under the influence of such a system, no code can be assumed to be \atmsynstruct (\refdefn{atmsynstruct}), and thus traditional syntax-based reasoning techniques will not directly apply;
however by instituting the system of \cite{FM18} the consequence on behaviours can be determined.
We do not devote more time to this feature of the \Clang model in this paper for several reasons.
\begin{itemize}
\item
Power is still subject to processor reorderings which are captured by the concepts herein; the memory subsystem is governed by separate mechanisms;
\item
Power is the only commercially available hardware that exhibits such behaviours.
Arm once allowed them, but no manufacturer ever implemented them on a real chip, and Arm has since removed these behaviours from its model \cite{ARMv8.4}.
A large reason for omitting the possibility of such weak behaviours 
was due to the complexity of the induced axiomatic models \cite{HerdingCats}, which were difficult to specify and reason about.
Arm is generally considered one of the faster architectures, 
and does not appear to suffer a significant penalty due to not having a similarly weak cache coherence system.  
(Intel's Itanium architecture, and mooted memory model (\eg \cite{Itanium2003}), was also intended to allow these weak behaviours; however it ceased production in 2020.)
\item
In practice, Power's weak behaviours are so difficult to reason about that developers make heavy use of
fences to eliminate the effects; and as a result whatever performance gains
there may have been are eliminated.
\item
The behaviours that arise are entirely to do with that micro-architectural feature and not due to compiler transformations or processor reorderings.
Reasoning about \Clang memory model features that affect local code (for, say, x86) architectures carries over identically to reasoning about local code in Power.
\item
Although, theoretically, under the \Clang memory model a compiler could instrument code transformations that mimic the behaviour of Power's cache system on architectures that do not exhibit 
them natively, this seems highly unlikely;
therefore, only developers targetting Power will ever experience the resulting complexity
(which arguably can and must be removed by the insertion of fences).
\end{itemize}

Outside of such considerations, the \Clang model as we have presented it 
is relatively straightforward; the behaviour introduced by Power's cache system is completely separate to pipeline reordering and compiler transformations.  It seems
unfortunate to complicate the entirety of the \Clang memory model to accommodate one architecture, 
when its effects can be captured in a separate memory subsystem specification (which can be ignored by developers 
targetting different architectures).  Hence
we recommend keeping the extra behaviours induced by this separate mechanism, peculiar to a single currently-in-production architecture, separately specified within the model;
we point to \cite{UnderstandingPOWER,SewellRelaxedAtomics,FM18} as examples of how this can be done, which fulfil similar roles to time-stamped event
structures in \cite{PromisingSemantics} and the
shared action graph of \cite{OGforC11}.

\subsection{Incorporating compiler transformations}
\labelsect{ctrans}

Any possible compiler transformation can be straightforwardly incorporated into our framework as an operational rule.
For instance, if $pattern$ is some generic code pattern that can be transformed into $pattern'$ then this can become a rule
$pattern \tra{\tau} pattern'$, \ie, a silent transformation that can occur at any time (an example of such a transformation might be to refactor a loop).
The downside of including such transformations within the semantics is, of course, depending on the structural similarity of $pattern$ and $pattern'$,
one cannot expect to reason about $pattern$ using syntax-directed proof rules (one has broken the principle of \atmsynstruct code, \refdefn{atmsynstruct}).
Of course, trace-based approached can still be applied.


As an example, consider a transformation discussed in \refsect{rfub},
which simplifies and removes a conditional.  This can be expressed as an operational rule that applies in some circumstances.

\ruledefNamed{0.8\textwidth}{Eliminate conditional}{elim-cond}{
	(\IFsC{r \neq n}{b \asgn \True \ppseqc r \asgn n})
	\tra{\tau}
	(b \asgn (r = n) \ppseqc r \asgn n )
}

As discussed in \refsect{rfub}, allowing this transformation has significant effects on behaviours in a wider context.
Consider also transformations to eliminate redundant loads or stores,
the simplest instrumentation of which are given by the following rules.

\ruledefNamed{0.5\textwidth}{Load coalescing}{load-coalesce}{
	r_1 \asgn x \ppseqc r_2 \asgn x
	\tra{\tau}
	r_1 \asgn x \ppseqc r_2 \asgn r_1
}
\ruledefNamed{0.4\textwidth}{Write coalescing}{write-coalesce}{
	x \asgn e_1 \ppseqc x \asgn e_2
	\tra{\tau}
	x \asgn e_2
}

These transformations eliminate an interaction with main memory.  \refrule{load-coalesce}
reduces the number of behaviours of the program since $r_2$ will always take the same final value as $r_1$, whereas
in the original code it is possible for them to receive different final values.
The transformation encoded in \refrule{write-coalesce} reduces the overall behaviours of the system as a parallel process can never receive the value of $e_1$
(of course, these transformations could be made dependent on memory ordering constraints using, for instance, the $\ocmore$ relation).
We believe it should be outside of the scope of the definition of \Clang memory model, and certainly outside the scope of this work,
to consider every possible transformation of every possible compiler;
however we have provided a
framework in which the consequences of a particular transformation can be relatively straightforwardly assessed,
separately to the specification and principles of the memory model itself.
In particular this may feed in to the development of verified compilers, and the development of a bespoke set of transformations that are valid within concurrent systems.

\section{Related work}

The development of the \Clang (and \Cpp) memory model is ongoing and governed by an ISO committee (but heavily influenced by compiler writers \cite{DepthsOfC}), 
covering the vast range of features of the language itself and including 
a semi-formal description of the memory model.  
Boehm and Adve \cite{BoehmAdveC++Concurrency} were highly influential initially, building on earlier work on memory model specifications (\eg, 
\cite{LamportSC,AccessBufferingInMultiprocs86,ShashaSnir88,ReleaseConsistency90,Dill98,MM=RO+At,FoxHarman2000}).
Since then many formal approaches have been taken to further understand the implications of the model and tighten its specification, especially the works of 
Batty et. al 
\cite{MathematizingC++, LibraryAbstractionsForC,LeakySemicolon}
and Vafeiadis et. al
\cite{VafeiadisC11, PromisingSemanticsKang,TamingRC,RepairingSCinC11}.
The model abstracts away from the various hardware memory models that \Clang is designed to run on (\eg, Arm, x86, Power, RISC-V, SPARC),
leaving it to the compiler writer to translate the higher-level program code into assembler-level primitives that will provide the behaviour specified by the \Clang model.
This behaviour is described with respect to a cross-thread ``happens-before'' order, which is influenced by so called ``release sequences'' of events within the system as a whole.
As a formalism this is difficult to reason about, and in particular, it removes the ability to think thread-locally -- whether or not enough fences or constraints have been
placed in a thread requires an analysis of all other events in the system, even though, with the exception of the Power architecture (see \refsect{power}), any weak behaviours
are due to local compiler transformations or instruction-level parallelisation at the assembler level.
Our observations of a myriad of discussions on programming discussion boards online is that programmers tend to think in terms of local reorderings, and appealing to
cross-thread concepts is not a convenient abstraction.
The framework we present here is based on a close abstraction of instruction-level parallelisation that occurs in real pipelined processors or instruction shuffling that a compiler
may perform.  For simple programs, where enough synchronisations have been inserted to maintain order on the key instructions, 
reasoning can be shown to reduce to a well-known sequential form; or specific reorderings that are problematic can be elucidated.
The underlying semantics model is the standard interleaving of relations on states (mappings from variables to values).  We argue this is simpler and intuitive -- as far as 
instruction reorderings can be considered so -- than the current description in the \Clang standard.
Part of the complication of the standard arises from handling the complexity of the Power cache coherence system, which cannot be encoded in a thread-local manner;
however those complications can be treated separately \cite{UnderstandingPOWER,SewellRelaxedAtomics,FM18}, and in any case, the Power model also involves instruction-level parallelism governed by the same principles
as the other major architectures (Arm, x86, RISC-V, SPARC).

In comparison with other formalisms in the literature 
\cite{VafGroundingThinAir,DemskyOOTA,SewellOperationalSemanticsC11}, 
many use an axiomatic approach (as exemplified by \cite{HerdingCats}), which are necessarily cross-thread specifications
based on the happens-before relationship, many use a complex operational semantics, and many combine the two.  
The Promising semantics \cite{PromisingSemantics,Promising2} is operational in flavour but has a complex semantics
\cite{VafGroundingThinAir}
involving time-stamped operations and several abstract global data structures for managing events.
In these frameworks reasoning about even simple cases is problematic, whereas in our framework the effects of reordering can be analysed by reduction rules at the program-level (\refsect{reduction}),
and our underlying model is that of the typical relations on states and so admits standard techniques (\refsect{reasoning}).
A recent advance in reasoning for the \Clang memory model is that of Wright et. al \cite{OGforC11}, but that framework appeals to a shared event graph structure,
and does not consider memory fences or \seqcst constraints.
Very few of the formalisms surveyed have a simple way of address consume (\consume) constraints (\refsect{consume});
we also argue that our formal approach to understanding pathological behaviours such as out-of-thin-air (\refsect{oota}) and read-from-untaken-branch (\refsect{rfub})
provides a clearer framework for understanding and addressing their consequences.

An overarching philosophy for our approach has been that of a separation-of-concerns, meaning that a programmer can consider which mechanisms for enforcing order 
should suffice for their concurrent code; of course, if their ordering mechanisms are embedded in complex expressions that may be optimised (\refsect{expression-optimisations})
or incrementally evaluated (\refsect{incremental-evaluation}) then the picture becomes more complex, but this can be analysed separately, and without regard to how such 
features interact with cross-thread, complex, abstract data structures controlling events.
The reordering framework we present here is based on earlier work \cite{SEFM21,pseqc-arXiv,FM18}, which provides a model checker (written in Maude \cite{Maude}) 
and machine-checked refinement rules (in Isabelle/HOL \cite{Paulson:94,IsabelleHOL}) for the language in \refsect{semantics} (\ie, ignoring the possibility of incremental evaluation and optimisations)
with forwarding (\refsect{forwarding}).
We straightforwardly encoded the definition of the \Clang memory model (\refmm{Cmm}) in the model checker and used it on the examples in this paper as well as those 
provided by the Cerberus project \cite{CerberusSewell}.
The framework has provided the basis for other analyses involving security \cite{PseqcSpectre-short,CoughlinInfoFlow,WinterBackwardsInfoFlow} and 
automated techniques \cite{CoughlinRGWMMs,SharmaDynamicC11}

\OMIT{
Consider the following parallel program with two processes, where we leave the second ($c_2$) unspecified.
\begin{equation*}
	(x \asgn 1 \ppseqc \relfence \ppseqc r_1 \asgn \modyA \ppseqc z \asgn 1 \ppseqc r_2 \asgn w)
	\pl 
	c_2
\end{equation*}
This short program contains fences, acquires and release operations, stores and loads.
Our rules allow us to understand the actual behaviour quite straightforwardly: the first process reduces to 
\begin{equation*}
	(x \asgn 1 \ppseqs z \asgn 1) \pl (r_1 \asgn y \ppseqs r_2 \asgn w)
	\pl 
	c_2
\end{equation*}
The release fence orders the stores, and the acquire tag orders the loads, and there is no other (variable-based) dependencies between instructions.
The release fence has no effect on the state and is redundant since it has no further reordering effect (all instances of $\ppseqm$ are sequentially consistent)
and so it can be ignored for the purpose of analysis.
Showing that this program satisfies some property may or may not be trivial, but that will be because of a combination of the property, and not because the semantics itself is complex
or uses complex data structures.

To understand the behaviour of this using an axiomatic approach requires looking at all possible traces, not just of $c_1$ but also requires knowing $c_2$.
This confounds local reasoning and prohibits the use of existing techniques for sequentially consistent (and atomically-structured) programs.
To understand the behaviour of this process using the Promising semantics is also quite complex, requiring the application of many tailored semantic rules and 
several complex data structures to elucidate.
The semantics we give allows reasoning about reorderings to occur at the syntactic level, which, as argued earlier, is anecdotally how programmers think.
}

\section{Conclusions}

We have given a definition of the \Clang memory model which keeps the fundamental concepts involved (fences, ordering constraints, and data dependencies) separated from other aspects
of the language
such as expression evaluation and optimisations,
which are present regardless of whether or not the memory model is considered. 
Provided programmers keep to a reasonable discipline of programming that aids analysis of concurrent programs, \ie, program statements are generally indivisible 
(at most one shared variable per-expression), 
our framework facilitates structured reasoning about the code. 
This involves elucidating the effect of orderings on the code as-written, and then applying existing techniques for establishing the desired properties. 
We argue that our framework more closely expresses how programmers of concurrent code think about memory model effects than the abstraction given in the current standard 
(cross-thread happens-before relationships and release sequences).
This is largely because the effects that a \Clang programmer needs to consider are compiler reorderings of instructions for efficiency reasons, or 
architecture-level reorderings in instruction pipelines.
The \Clang language is rich in features and any claim to a full semantics of arbitrary \Clang code with concurrency requires a full semantics of \Clang in general, 
and as far as we are aware this has not been fully realised as yet; but we intend that our approach can be relatively straightforwardly incorporated into such by virtue of its 
\emph{separation of concerns} - the fundamental properties of the memory model are universal and consistent even in the presence of complicating factors.

We note that 
the difficulties that arise in the attempt to formalise the \Clang memory model stem from the tension between well-established compiler transformations as well as the need to 
support a multitude of 
hardware-level memory models seamlessly versus the well-known intricacies of programming correct shared-variable algorithms
\cite{Cerberus-BMC}.
This will be an ongoing balancing act that involves many competing factors, especially and including efficiency, and, increasingly, safety and security \cite{SafeByDefaultConcurrency};
if we were to take a position, it would be that sections of code -- hopefully, relatively small and localised  -- can be protected from arbitrary transformations from compiler theory and practice.
For instance, a \Clang scoping construct \T{concurrent \{} $\ldots$ \T{\}} which is compiled with minimal optimisations or reorderings, regardless of command-line flags such as \T{-O}.
The \Clang standard may then be able to provide guarantees about executions for such delimited code, while also allowing the programmer flexibility to make use of optimisations where they
have determined they are applicable.

\bibliography{biblio,colvinpubs,colvinTR,references}

\appendix

\section{Lifting from actions to traces and commands}
\labelappendix{syntax-lifting}

Throughout the paper we have used several functions and definitions based on actions, which for the most part are straightforwardly lifted to commands and traces.
For completeness we give these below.

\subsection{Extracting variables}

\newcommand{\genfnword}{{\sf fn}}
\newcommand{\genfn}[1]{\genfnword(#1)}

For any function $\genfn{.}$ that simply collects syntax elements into a set (\ie, $\fv{.}, \wv{.}, \rv{.}, \getocs{.}, \getfences{.}$, \etc)
and is defined over instructions can be straightforwardly lifted to commands, actions and traces in the following generic pattern.
\begin{eqnarray}
	\genfn{\Nil} \aeq \ess
	\\
	\genfn{\vecaca} \aeq \bigcup_{\aca \in \vecaca} \genfn{\aca}
	\\
	\genfn{\cmdc_1 \choice \cmdc_2} \aeq \genfn{\cmdc_1} \union \genfn{\cmdc_2}
	\\
	\genfn{\cmdc_1 \ppseqm \cmdc_2} \aeq \genfn{\cmdc_1} \union \genfn{\cmdc_2}
	\\
	\genfn{\iteratecm} \aeq \genfn{\cmdc}
	\\
	\Also
	\genfn{\eseq} \aeq \ess
	\\
	\genfn{\aca \cat t} \aeq \genfn{\aca} \union \genfn{t}
\end{eqnarray}

\OMIT{
from which we can straightforwardly derive:
\begin{eqnarray}
	\eseq \ro \acb \aiff \True
	\\
	\aca \cat t \ro \acb \aiff \aca \ro \acb \land t \ro \acb
	\\
	\cmdc \ro \eseq \aiff \True
	\\
	\cmdc \ro \aca \cat t \aiff \cmdc \ro \aca \land \cmdc \ro t
	\\
	\Nil \ro \acb \aiff \True
	\\
	c_1 \choice c_2 \ro \acb \aiff 
		c_1 \ro \acb \land c_2 \ro \acb 
	\\
	c_1 \ppseqc c_2 \ro \acb \aiff 
		c_1 \ro \acb \land c_2 \ro \acb 
	\\
	\iteratecm \ro \acb \aiff 
		c \rom \acb 
\end{eqnarray}
}

\subsection{Lifting forwarding/reordering triples}
\labelappendix{lifting-forwarding}

Forwarding (see \refsect{forwarding}) an action $\aca$ to action $\acb$ is defined below.
Assume $x \neq y$.
\begin{eqnarray*}
    \fwda{(\modocs{y} \asgn f)} \sseq \modocs{y} \asgn (\fwda{f})
	\qquad
    \fwda{\guardb} \aeq \guard{\fwda{b}}
	\qquad
    \fwda{\fencepf} \sseq \fencepf
	\\
	\Also
    \fwda{v} \sseq v
	\qquad
    \fwda{(\unop f)} \aeq \unop (\fwda{f})
	\qquad
    \fwda{(e_1 \binop e_2)} \sseq (\fwda{e_1}) \binop (\fwda{e_2})
	\\
	\Also
    \fwd{\modx{\ocs_1} \asgnlbl e}{\mody{\ocs_2}} \sseq \mody{\ocs_2} 
	\qquad && \qquad
    \fwd{\modx{\ocs_1} \asgnlbl e}{\modx{\ocs_2}} \sseq 
		e 
\end{eqnarray*}
Given $\aca$ is an assignment $\xasgne$ then $\fwdab$ essentially replaces references to $x$ (with any constraints) by $e$.

\labelappendix{lifting-triples}

The reordering relation can be lifted from actions to commands straightforwardly as below.
\begin{eqnarray}
	\robpb{\Nil}
	\asdef
	\acb' = \acb
	\\
	\robpab
	\asdef
	\aca \ro \acb' \land \acb' = \fwdab
	\\
	\rocmdm{\acb''}{c_1 \ppseqc c_2}{\acb}
	\asdef
	\exists \acb' @ 
		\rocmdm{\acb''}{c_1}{\acb'} 
		\land
		\robpb{c_2}
	\\
	\robpb{c_1 \choice c_2}
	\asdef
		\robpb{c_1}
		\land
		\robpb{c_2}
	\\
	\robpb{\iteratecm}
	\asdef
		\forall i \in \nat @ \robpb{\finiteratecm{i}}
\end{eqnarray}

\OMIT{
\subsection{Removing memory model-specific notation}

Given a program $\cmdc$ in \impro we can eliminate the parts of it relating specifically to memory models to give its
fundamental form, that is, removing ordering constraints from variables and turning fences into no-ops,
under the assumption that all uses of \pseqc are either sequential or parallel composition.  If this latter constraint doesn't hold
then the effects of stripping memory model artifacts is irrelevant.
We straightforwardly define $\stripoc{(.)}$ for expressions, actions and commands below.

\begin{eqnarray}
\stripoc{\modxocs} \sseq x
&&
\stripoc{v} \sseq v
\\
\stripoc{(\unop e)} \sseq \unop (\stripoc{e})
&&
\stripoc{(e_1 \binop e_2)} \sseq (\stripoc{e_1}) \binop (\stripoc{e_2})
\\
\also
\stripoc{(\modxocs \asgn e)}
\aeq
x \asgn (\stripoc{e})
\\
\stripoc{\guarde}
\aeq
\guard{\stripoc{e}}
\\
\stripoc{\fencepf}
\aeq
\tau
\\
\Also
\stripoc{\Nil} \aeq \Nil
\\
\stripoc{\vecaca} \aeq map(\stripoc{(.)}, \vecaca)
\\
\stripoc{(c_1 \choice c_2)} \aeq \stripoc{c_1} \choice \stripoc{c_2}
\\
\stripoc{(c_1 \bef c_2)} \aeq \stripoc{c_1} \bef \stripoc{c_2}
\\
\stripoc{(c_1 \pl c_2)} \aeq \stripoc{c_1} \pl \stripoc{c_2}
\\
\stripoc{(\iteratec{\SCmm})} \aeq \iterate{(\stripoc{{c}})}{\SCmm}
\end{eqnarray}
We leave the effect of stripping memory model notation undefined for programs not in sequential form,
\ie, where all instances of \pseqc are parameterised by either \SCmm (`$\bef$') or \PARmm (`$\pl$'),
and iteration is strictly sequential.

These give, for instance,
\begin{eqnarray*}
	\stripoc{((\modxR \asgn \modyA) \bef \ffence \bef r \asgn \modyX ~\pl~ \guard{\modyA = 1})}
	\aeq
	(x \asgn y) \bef \tau \bef r \asgn y
	~\pl~
	\guard{y = 1}
	\\
	\stripoc{\ffence}
	\aeq
	\tau
	\\
	\stripoc{(\modxR \asgn \modyA)}
	\aeq
	x \asgn y
\end{eqnarray*}

\robnote{Do I really need this?  Rather than strip better to have a notion of equivalence that is state-based (like $\eqmodoc$), 
and show that fences can be ignored.  The strip might be useful for defining $\effa$, however}

The command $\stripocc$ is a plain version of $\cmdc$, that is, traces of the memory model have been removed (fences and ordering constraints), leaving just a plain expressions/assignments language 
(defined inductively over the command syntax).
For example:
\begin{eqnarray}
	c
	\eqmodoc
	d
	\asdef
	\stripoc{\cmdc}
	= 
	d
	\qquad
	or
	\quad
	(\exists d' @ \stripoc{\cmdc} \mbox{~'is'~} d' \land d' \refeq d)
\end{eqnarray}
}

\OMIT{

\section{Standard rely/guarantee inference rules}

See figure; based on Presna-Nieto's encoding in Isabelle/HOL.

\begin{figure}

\begin{equation}
	\Rule{
		\Stablepr
	}{
		\rgqpr{\idn}{p}{\Skip}
	}
\end{equation}

\begin{equation}
\labellaw{rgq-conditional}
	\Rule{
		\rgq{p \land b}{r}{g}{q}{\cmdc_1}
		\qquad
		\rgq{p \land \neg b}{r}{g}{q}{\cmdc_2}
	}{
		\rgqprgq{\IFbc}
	}
\end{equation}

\begin{equation}
\labellaw{rgq-2act-bef}
	\Rule{
		\rgq{p}{r}{g}{\qmid}{\cmdc}
		\qquad
		\rgq{\qmid}{r}{g}{q}{\cmdd}
	}{
		\rgq{p}{r}{g}{q}{\cmdc \bef \cmdd}
	}
\end{equation}

\begin{equation}
	\Rule{
		\rgqp{r \lor g_2}{g_1}{q_1}{c_1}
		\\
		\rgqp{r \lor g_1}{g_2}{q_2}{c_2}
	}{
		\rgqp{r}{g_1 \lor g_2}{q_1 \land q_2}{c_1 \pl c_2}
	}
\label{rule:rgq-pl}
\end{equation}

\begin{equation}
	\Rule{
		\Stablepr 
		\quad
		\Stableqr 
		\\
		\postState{p \precomp \relasgnxe} \imp q
		\qquad
		(p \precomp \relasgnxe) \imp g
	}{
		\rgqprgq{x \asgn e}
	}
\end{equation}

\begin{equation}
\label{rule:rqg-refeq}
	\Rule{
		c \refeq d
		\qquad
		\rgqprgq{d}
	}{
		\rgqprgqc
	}
\end{equation}

Definitions:
\begin{eqnarray}
	\relasgnxe
	\asdef
	x' = e \land \idxbar
	\\
	\postStater
	\asdef
	\{\sigma' | (\exists \sigma @ \ssp \in r) \}
	\\
	p \precomp r
	\asdef
	\{\ssp \in r | \sigma \in p\}
	\quad
	\mbox{for $p$ a predicate and $r$ a relation}
	\\
	\Stable{p}{r}
	\asdef
	\postState{p \precomp \Transr} \imp p
	\\
	\idn
	\asdef
	\{\ss | \sigma \in \Sigma \}
	\\
	\id{X}
	\asdef
	\{\ssp | (\forall x \in X @ \sigma(x) = \sigma'(x))\}
	\\
	\xbar
	\asdef
	\{y \in \Var | y \neq x\}
\end{eqnarray}

\caption{Standard rules for rely/guarantee reasoning}
\labelfig{rg-std}
\Description{TODO}
\end{figure}

\OMIT{

Just a few extra rules are required for the special action types, and these are trivial.  We also repeat from Presna-Nieto, Coleman et al in the appendix.
A proof of the inference rules is out of scope of this work, indeed the point is that we can reuse.  See \cite{SoundnessRG,rgsos12,rgsos14}.
For elucidation we give a couple of specialisations for the case of exactly two instructions in 
\reffig{rg-wmms}.

\begin{figure}

\begin{equation}
\labellaw{rgq-2act-cbef-nro}
	\Rule{
		\aca \nro \acb
		\qquad
		\rgq{p}{r}{g}{q}{\aca \bef \acb}
	}{
		\rgq{p}{r}{g}{q}{\aca \ppseqc \acb}
	}
\end{equation}

\begin{equation}
\labellaw{rgq-2act-cbef-ro}
	\Rule{
		\aca \ro \acb
		\qquad
		\rgq{p}{r}{g}{q}{\aca \bef \acb}
		\qquad
		\rgq{p}{r}{g}{q}{\acb \bef \aca}
	}{
		\rgq{p}{r}{g}{q}{\aca \ppseqc \acb}
	}
\end{equation}

\begin{equation}
	\labellaw{rgq-fence}
	\rgqpr{\idn}{p \land r^*}{\scfence}
	\qquad
	\qquad
	\rgqpr{\idn}{p \land r^*}{\tau}
\end{equation}

\begin{equation}
\labellaw{rgq-fence-divides}
	\Rule{
		\rgqprgq{\cmdc \ppseqs \cmdd}
	}{
		\rgq{p}{r}{g}{q}{\cmdc \ppseqc \scfence \ppseqc \cmdd}
	}
\end{equation}

\robnote{Generalise to commands now}

\caption{Weak-memory-model-specific rules for rely/guarantee reasoning}
\labelfig{rg-wmms}
\Description{TODO}
\end{figure}
}
}

\OMIT{
\section{Alternative treatments of atomicity}

The separation of expression evaluation is modelling a compiler that reads a \Cxi statement $z \asgn x + y$ and compiles it to
\begin{equation}
	tmp_1 \asgn x \scomp tmp_2 \asgn y \scomp z \asgn tmp1 + tmp2
\end{equation}
for fresh temporary variables.  In our framework we avoid that extra complication but retain the interactions with the shared state.
(\cf, Doherty et al's semantics \cite{VerifyingC11Doherty}.)

Algebraically we have the following, where $v, v_i$ range over type $\Val$:
\begin{equation}
	r \asgn x \refeq 
		\pChoicev{\guard{x = v} \ppseqc r \asgn v}
	\\
	z \asgn x + y \refeq 
	\pChoice{v_1,v_2}{\loadx{v_1} \ppseqc \loady{v_2} \ppseqc z \asgn v_1 + v_2}
\end{equation}
This can be generalised for arbitrary expressions; see \cite{SRA16}.

\robnote{Should mention that we aren't modelling the "steps" of the compiler, just the steps of what it can eventually produce}
}

\OMIT{

\section{Splitting out optimisation steps}

\robnote{Reference from main text}

In \refrule{optimise-expr} we did not constrain the form of $b$, allowing in particular, $b$ to contain references to shared variables.
This may be a bit unrealistic, as it allows behaviours one would not expect to be indivisible.

Consider:
\[
	x \asgn y ; z \asgn x - y
\]
With optimisation let $\rocmdg{z \asgn y - y}{x \asgn y}{z \asgn x - y}$.
Obviously not allowed but then we let this become $z \asgn 0$ since 
\[
	\getocs{0} \subseteq \getocs{y - y} \subseteq \{\relaxed\}
	\\ and \\
	\getocs{y - y} \subseteq \getocs{x - y} \subseteq \{\relaxed\}
\]
\[
	\rocmdg{\guard{b'}}{x \asgn e}{\guardb}
	\sdef
	b' \equiv \fwd{x \asgn e}{b} \land b \lexmore b' \land b \ocmore b'
	\land
	x \asgn e \ro \guardbp 
\]

To both allow for local optimisations using shared variables but prevent them from propagating one can institute a new label type ``$opt(b)$''
which is similar to a guard but has a constraint that at the top level we only consider $b$'s such that $\svb = \ess$, however,
\refrule{optimise-expr} can still emit more complex $b$s; the point is for forwarding to eliminate them before they hit the storage system, as in the above example.
For simplicity we left this consideration out of the main text.  We note that reasoning becomes a bit harder in this case...
}

\OMIT{

\section{C equivalents}

For reference, here are the following translations of library functions in 
\url{
https://en.cppreference.com/w/c/atomic
}

{
\footnotesize
\begin{eqnarray*}
	r \asgn \atmexchexpl{x, e, \oc}
	\asdef
	\seqT{r \asgn \modxoc \asep \modxoc \asgn e} 
\\
	\atmcmpexchstrexpl{x, e, e', \oc_{succ}, \oc_{fail}}
	\asdef
    \seqT{\guard{\modx{\ocfail} = e} \asep \modx{\ocsucc} \asgn e'} \choice \guard{\modx{\ocfail} \neq e}
\\
	r \asgn \atmfetchaddexpl{x, e, \oc}
	\asdef
	\seqT{r \asgn \modxoc \asep x \asgn \modxoc + e} 
\end{eqnarray*}
}

Atomic \T{\_sub}, \T{\_or}, \T{\_xor}, and \T{\_and} versions of $\atmfetchaddexpl{x, v, \oc}$ are the obvious analogies.

\robnote{A "fence" comes from \T{atomic\_thread\_fence}, so maybe don't use "fence"

\robnote{Could consider a ``return'' statement; point to OO paper for how to do it}

\T{atomic\_signal\_fence}
(\url{
https://en.cppreference.com/w/c/atomic/atomic_signal_fence
})
description suggests compiler reordering can be turned off.
}

{
\footnotesize
\begin{eqnarray*}
	r \asgn \atmfetchadd{x, v}
	\asdef
	r \asgn \atmfetchaddexpl{x, v, \seqcst}
\\
	\atmcmpexchstr{x, e, e'}
	\asdef
	\atmcmpexchstrexpl{x, e, e', \seqcst, \seqcst}
\\
	r \asgn \atmexch{x, e}
	\asdef
	r \asgn \atmexchexpl{x, e, \seqcst}
\end{eqnarray*}
}
}

\OMIT{
\section{Swap}

\robnote{Not actually needed as this is assembler not C.  Could be an impl. of \T{atomic\_compare\_exchange}, sect 7.17.7.4.  Maybe turn this into a section that addresses other stuff in 7.17..}

\robnote{Merge to previous appendix}

One of the main ways of enforcing atomicity on x86-TSO is the \T{xchg} operation, which swaps (exchanges) the values of a register and address in memory.
As a demonstration of the flexibility of our approach we show how to extend the language to handle this common instruction type.
Note that handling it directly in \impro is non-trivial as value swapping requires a temporary value.  However we extend the language straightforwardly to allow multiple assignment.
\begin{eqnarray}
	\aca \attdef \ldots \csep \xvec \asgn \evec
	\\
	\wv{\xvec \asgn \evec} \aeq \xvec \\
	\rv{\xvec \asgn \evec} \aeq \fv{\evec} 
	\\
	\mathbf{xchg}(x,r) \asdef \seqT{x,r} \asgn \seqT{r,x}
\end{eqnarray}
We assume $\xvec$ is a list of distinct variables and $\evec$ is a list of expressions of the same length as $\xvec$.
}

\OMIT{
\section{Non-multi-copy atomicity}
\labelsect{lmca}

Here we justify assuming multicopy atomicity:

\begin{itemize}
\protect \item
TSO and ARMv8 require multicopy atomicity
\item
Older versions of ARM do not require mca, but no real hardware ever made use of it, ie, it was never observed
\item
ARM engineers decided it was too complex to allow reasoning, and confusion resulted in too many fences in final code to avoid the subtleties
\item
This leaves only the POWER architecture which lacks mca, and this is not as widely used as the others
\item
Although the \Cxi standard (apparently?) does not require mca, we assume that behaviours that arise when lacking \emph{do not} appear when run on architectures which do not exhibit it, that is, compiler writers do not explicitly code in
lack of mca behaviour into their translations, and that no programmers depend on such behaviours being possible (since the behaviours that arise due to mca could also occur even if mca is not assumed)
\item
In summary: in general the lack of mca was considered so complex to reason about that it was not worth allowing from ARM's perspective, one of the biggest vendors of hardware; it does not occur in another widely used wmm, TSO; and probably is not visible for
algorithms with 2 or fewer processes.  Therefore we use the simpler RG rules.  Of course, as demonstrated in \cite{FM18}, we can express the lack of mca in our framework, and therefore it is possible to develop more complex rules to handle it.

\end{itemize}

Non-mca architectures are rare: Power is the only known, and this due to its cache system.  Clearly this cannot be explained by local
reorderings.  It is not derived from pipeline reorderings on processors.  There is no guarantee some other company won't come up with its own,
so dealing with this specifically is not that useful in general.  ARM got rid of it, after no one used, the spec was complicated, and ARM was
fast enough anyway.  Pretty hard to build in non-mca into mca architecture; we assume this can't happen, and if we are wrong, that it won't
happen anyway.

Verification gets harder.
Options
\begin{enumerate}

\item
Treat the write queue $W$ in FM18 as a global variable, and treat updates and fences as a modification as outlined in FM18.  Reads become \emph{per-thread} queries which
return a set of possible values.  Predicates $P,R$ over states can stay similar, but have to be interpreted per-thread somehow at the top level; they have to hold for all
possible values that their thread can see

\item
Current preferred: distinguish between local and global knowledge.  Setting $x \asgn 1$ lets the local process assume $local(x) = 1$ but cannot infer anything about the
other processes' knowledge.  A fence converts local to global.  Such a set of rules will not be complete, but should be sufficient for any reasonable example.  For instance, I am
thinking that such rules could not be used to prove a specific possible outcome for the IRIW litmus test, as the non-multicopy atomic behaivour is not easy to reason about
conceptually even, but for standard proofs, in particular rely/guarantee-based proofs, a simple local to global rule is enough.
Proving soundness of such a rule will be complex

Note that local knowledge can be obtained by loads, and a (full) fence is required to convert an observation like $y = 1$ into a global observation.  This relates to
Graeme's example of how the lack of multicopy atomicity creates security issues.

\end{enumerate}

\sbitem{
A lack of multicopy atomicity.  The Power memory model lacks multicopy atomicity, that is, in a system of 3 or more processes, two processes may have a different view of the order or modifications to global variables
(this is not to be confused with reordering of instructions on a single thread -- while a process's view of the order of modifications may disagree with the order in the \emph{progam text}, we are referring here to two
processes disagreeing with each other on the actual memory order).  Although the semantic framwork developed in \cite{FM18} can express the lack of multicopy atomicity, developing rely-guarantee rules to cope for it is
nontrivial (however, the Promising Semantics does \cite{PromisingSemantics}).  Recently the ARM memory model was strenthened to require multicopy atomicity \cite{PultePOPL}, and TSO has always had it.  As these are
significantly more widely used than POWER, we do not consider it a limitation to focus our framework on these cases, even though the \Cxi memory model does not assume multicopy atomicity.  Also see \refsect{lmca} for
further discussion on this limitation
}

	\sbitem{
This is in large part so that C code can be compiled to assembler on processors that are not ``multicopy atomic''; however the only mainstream desktop processor that is not multicopy atomic is 
POWER, but the far more widely used x86 and Arm processors are multicopy atomic.  In Arm's case (used in over a billion devices) they recently removed the possibility of non-multicopy atomic behaviours
because of their complexity.
}

}

\OMIT{
\section{Non-atomics}
\labelsect{non-atomics}

\robnote{Is this even a thing?  What does the spec say about not using them?
I guess the point is that they are shared vars, but not declared atomic.
Are they therefore still affected by all the release sequence stuff?  I think drop}

The precise semantics of these is unclear.  I think the idea is that if a location is not atomic (non-atomic) then it can be treated as a local variable by the compiler.
However this does not hold if the program is data-race free, ie., if non-atomic accesses are guarded by locks, release/acqire etc.  Therefore a different semantics is
needed in each case, but this is context-dependent.  (By which I mean, if we just treat non-atomics as locals, then we will not get the correct behaviour for 
(shared) non-atomics when they are protected by locks.)   We may have to forgo a satisfactory account.

According to Wickerson et al \cite{AutomaticallyComparingMCMs}, Example 1 in S1.2, $\na$ actions cannot be reordered with Rel/Acq, that is, 
\begin{equation}
	(r_1 \asgn a \ppseqc \modR{x \asgn 1}) \pl (\modA{r_2 \asgn x} \ppseqc a \asgn 1)
\end{equation} 
with a final state where $r_1 = 1 \land r_2 = 1$ is forbidden by \Cxi.
However they then go on to say that the semantics of the above program is undefined because it has a data race, so maybe this is not a good lead...

\robnote{Other text..:}
The memory model provides for untagged accesses to atomic (shared) variables.  We have assumed unmarked accesses are relaxed ($\relaxed$), but depending on the compiler these accesses
potentially allow optimisations.  The non-atomic accesses sit at the bottom of the preference pile.  They are still restricted by the same constraints such as $\Gmm$, and they are restricted by fences,
and as far as the ordering goes they are equivalent to relaxed, ie, there are no additional effects.  However they can be optimised away more strongly than relaxed accesses, which can be encoded by
the optimisations thing *as it currently stands*.

The point is that syntactic dependencies aren't necessarily preserved.

In the specification non-atomic accesses of atomic variables are not considered part of the release sequence unless they are behind some sort of synchronisation, in which case they are not a data race and so
have the typical semantics.  As such we treat everything as if it is DRF.  We do not attempt to mimic how a compiler might detect/determine DRFness and the subsequent things it can do or can't; we consider this
at least equivalent to putting an explicit marker "do not optimise" around things (or, preserved seq deps).

\OMIT{
We need to consider the possibility of ``consecutive'' loads of the same variable being coalesced into a single load from which both take their value.  This has a
flow on effect of potentially allowing reorderings that were possible before the coalesce.

We can achieve this in the semantics by generalising the reordering for guards (loads).  Currently we have
\begin{equation}
	\loadx{v_1} \nro \loadx{v_2}
\end{equation}
However we can generalise this if we consider them to be relations \OMIT{(see \refsect{relations})}.  We can actually state:
\begin{equation}
\labeleqn{ro-guard-new}
	\rocmd{\guard{v_1 = v_2}}{\loadx{v_1}}{\loadx{v_2}}{m}
\end{equation}

\newcommand{\ndgen}[2]{{\displaystyle \bigsqcap_{#1}}\, #2}
\newcommand{\exeval}[2]{{\Meaning {#1}_{#2}}}

Now we show the expansion of two consecutive non-atomic load (assignments).
\begin{equation}
	\ry1 \asgn y \scomp \ry2 \asgn y 
\end{equation}
becomes
\begin{equation}
	\left( \ndgen{v_1}{\load{y}{v_1} \scomp \ry1 \asgn v_1}\right) 
	\scomp 
	\left( \ndgen{v_2}{\load{y}{v_2} \scomp \ry2 \asgn v_2}\right) 
\end{equation}
Now we manipulate:

\begin{derivation}
	\step{
		\left( \ndgen{v_1}{\load{y}{v_1} \scomp \ry1 \asgn v_1}\right) 
		\scomp 
		\left( \ndgen{v_2}{\load{y}{v_2} \scomp \ry2 \asgn v_2}\right) 
	}

	\trans{=}{Distribute nondet choice}
	\step{
		\left( \ndgen{v_2}{ \left( \ndgen{v_1}{\load{y}{v_1} \scomp \ry1 \asgn v_1}\right) 
		\scomp 
		\load{y}{v_2} \scomp \ry2 \asgn v_2}\right) 
	}

	\trans{=}{Associativity; reorder}
	\step{
		\left( \ndgen{v_2}{ \ndgen{v_1}{
			\load{y}{v_1} 
			\scomp 
			\load{y}{v_2} 
			\scomp 
			\ry1 \asgn v_1}
			\scomp 
			\ry2 \asgn v_2}
		\right) 
	}

	\trans{=}{Reorder according to new (generalised) rule}
	\step{
		\left( \ndgen{v_2}{ \ndgen{v_1}{
			\guard{v_1 = v_2}
			\scomp 
			\load{y}{v_1} 
			\scomp 
			\ry1 \asgn v_1}
			\scomp 
			\ry2 \asgn v_2}
		\right) 
	}

	\trans{=}{Meta-level manipulation, clearly the only valid traces are those where the values are the same}
	\step{
		\left( \ndgen{v_1}{
			\load{y}{v_1} 
			\scomp 
			\ry1 \asgn v_1
			\scomp 
			\ry2 \asgn v_1}
		\right) 
	}

	\end{derivation}

Now we have shown that both registers can be serviced by the same ``load''.  

A similar manipulation \emph{may} be possible for write coalescing.

\subsection{Address shifting}

Consider:
\begin{equation}
\labeleqn{load-coalesce-ex}
	r_x \asgn x \scomp \ry1 \asgn y + r_x \scomp \ry2 \asgn y \scomp r_z \asgn z + \ry2 \quad
	\pl \quad z \asgn 1 \scomp \fence \scomp x \asgn 1
\end{equation}
For simplicity here we are assuming expression evaluation is atomic, but an assignment is not.  Using address shifting we can easily remove this assumption.
The load into $\ry2$ is blocked by the load into $\ry1$.  Due to the other dependencies it appears the load of $z$ must happen after the load of $x$.
However on real processors this is not the case; operationally it seems the second load of $y$ is allowed to proceed and the rest of the program is executed
speculatively, dependent on whether the first load reads the same value (actually, the same write).  

However the previous reasoning is not all that useful in the address shifting case.  
The tricky thing at the processor is level is something like \emph{address shifting}, where the load address is
dependent on a register and will not proceed until that is evaluation.  This is a case of atomic expression evaluation.  Let us write $f(a,r)$ for an address
calculated from address $a$ shifted by $r$.  The we have, by a processor-specific definition for $f$,
\begin{equation}
	r_1 \asgn f(y,r_2)
	\sdef
	\ndgen{v,v_y}{\guard{y = v_y \land v = f(v_y,r_2)} \scomp r_1 \asgn v}
\end{equation}
Note that the guard contains references to two variables, $y$ and $r_2$.  Now, seemingly, we have
\begin{equation}
	\rocmd{\guard{w = v_y}}{\guard{y = v_y \land v = f(v_y,r_2)}}{\guard{y = w}}{m}
\end{equation}
	

}
}

\OMIT{
\section{Generalising to relations}
\labelsect{relations}

\newcommand{\mangle}[2]{\fwd{#1}{#2}}

In the work we have considered so far we only have used either single assignments or guards.  These are both simple (and common) instances of relations.
We can work out some function which generalises forwarding so that for relations $R_1$ and $R_2$ we have
\begin{equation}
	\mangle{R_1}{R_2} \comp R_1 \imp 
		R_1 \comp R_2
\end{equation}
This should preserve sequential consistency, and reduce down to the special cases we have given, in particular, \refeqn{ro-guard-new}.

It may not be possible to define $\mangle{R_1}{R_2}$, other than via the above property.
}

\section{Mixing incremental and non-incremental evaluation}
\labelappendix{indivis-vecaca}

In \refsect{incremental-evaluation} 
we gave a semantics for evaluating instructions incrementally, as opposed to treating instructions as indivisible in \refsect{semantics}.
In this section, for completeness, we show how to mix both possibilities within the syntax of \impro.

\newcommand{\seqOf}{\mathbf{seq}~}
\newcommand{\instr}{\iota}
\newcommand{\specinstr}{{\instr^{+}}}
\newcommand{\codeinstr}{s}
\newcommand{\vecinstr}{\vec{\instr}}

\newcommand{\divistag}{\mathsf{divis}}
\newcommand{\indivistag}{\mathsf{indivis}}

We define an instruction $\instr$ to be one of the three basic types of assignment, guard and fence, and an action $\aca$ to 
be a list of instructions (written $\vecinstr$).  Actions are the basic type of a step in the operational semantics.
We define a ``specification instruction'' to pair a basic instruction $\instr$ with a designation
as to whether it is divisible or indivisible, \ie, whether it is to be executed incrementally or as a single indivisible step.
Finally, within the syntax of \impro, instead of a list of actions $\vecaca$ as the base type (\refeqndefn{cmd}), 
we allow a \emph{statement} $\codeinstr$, which is a list of specification instructions.
\begin{equation*}
	\instr \in \Instr
	\quad
	\aca \in Action
	\quad
	\specinstr \in SpecInstr
	\quad
	\codeinstr \in Statement
\end{equation*}
\begin{eqnarray*}
	\instr \attdef x \asgn e \csep \guarde \csep \fencepf
	\\
	\aca \asdef \vecinstr
	\\
	\specinstr \attdef \instr \cross (\mathsf{divisible} | \mathsf{indivisible})
	\\
	\codeinstr \attdef \vec{\specinstr}
\end{eqnarray*}
This gives a significant amount of flexibility in describing the execution mode of composite actions, for instance,
$
	\seqT{
	(\guard{x = y}, \indivistag),
	(x \asgn y + z, \divistag)
	}
$ is a statement that calculates whether $x = y$ in the current state, and then incrementally evaluates $y + z$ before assigning the result to $x$.
Of course, this level of flexibility is not necessary for the majority of cases,
and 
syntactic sugar can be used to cover the commonly occurring cases,
in particular, letting a singleton list of specification instructions be written as a single specification instruction; 
and conventions for distinguishing between divisible and indivisible versions of instructions.

We lift \refeqndefn{indivisact} for indivisible instructions to the new types.
\begin{eqnarray*}
	\indivisact{(\instr, \mathsf{indivisible})} \aeq \True \\
	\indivisact{(\instr, \mathsf{divisible})} \aeq \indivisact{\instr} 
	\\
	\indivisact{\codeinstr} \asdef \forall \specinstr \in \codeinstr @ \indivisact{\specinstr}
\end{eqnarray*}
Any specification instruction tagged $\indivistag$ is indivisible, 
while 
a specification instruction tagged $\divistag$ is divisible if its instruction is,
but is otherwise indivisible.

The relevant operational semantics for specification instructions and statements is as follows.
\begin{gather}
	\Rule{
		\instr \tra{\aca} \instr'
	}{
		(\instr, \mathsf{divisible}) 
		\tra{\aca}
		(\instr', \mathsf{divisible}) 
	}
	\labelrule{spec-instr}
	\\
	\Rule{
		\indivisact{\codeinstr_1}
		\qquad
		\specinstr \tra{\aca} \specinstr'
	}{
		\codeinstr_1 \cat \specinstr \cat \codeinstr_2
			\tra{\aca} 
		\codeinstr_1 \cat \specinstr' \cat \codeinstr_2
	}
	\qquad
	\qquad
	\Rule{
		\indivisact{\codeinstr}
	}{
		\codeinstr \ttra{strip(\codeinstr)} \Nil
	}
	\labelrule{code-instr}
\end{gather}
\refrule{spec-instr}
states that any instruction tagged divisible can take an incremental execution step according to the evaluation rules in \refsect{incremental-evaluation}.
This is used to build \refrule{code-instr} for a statement $\codeinstr$, where specification instructions within $\codeinstr$ 
are executed incrementally from left to right:
the first divisible specification instruction ($\specinstr$) in $\codeinstr$ may take a step, which becomes a step of the statement.
When all instructions within $\codeinstr$ are indivisible a final, single, indivisible step is taken.
This action is formed by simply stripping the $\indivistag/\divistag$ tags from the specification instructions, 
\ie,
$
	strip((\instr, \_)) \sdef \instr
$,
which is lifted to statements by applying $strip$ onto each element,
\ie,
$
	strip(\codeinstr) \sdef map(strip, \codeinstr)
$.

A compare and swap command (\refeqndefn{cas})
can be redefined to incrementally evaluate its arguments before 
executing an indivisible (``atomic'') test-and-set step.
\begin{eqnarray}
	\CASxe
	\asdef
	\seqT{(\guard{x = e}, \divistag), (x \asgn e', \divistag)} \choice (\guard{x \neq e}, \divistag)
\end{eqnarray}
In the successful case first $e$ is evaluated to a value $v$, then $e'$ is evaluated to a value $v'$,
and finally the action $\seqT{\guard{x = v} \tsep x \asgn v'}$ is executed (the $\divistag$ tags are stripped).
This means that $e$ and $e'$ can be incrementally evaluated, but the final test/update remains atomic.

\OMIT{
\section{The \Clang standard}
\labelappendix{c-standard}

\robnote{For PAC readers: this will likely be dropped but may be useful to reference}

\robnote{Mention which document versions were accessed}

See \url{http://en.cppreference.com/w/cpp/atomic/memory_order}

Possibly find the std through
\url{
https://stackoverflow.com/questions/81656/where-do-i-find-the-current-c-or-c-standard-documents
}, giving main(?) links to 
\begin{itemize}
\item
the \Clang committee (WG14),
\url{
http://www.open-std.org/jtc1/sc22/wg14/www/projects#9899
}
which
gives
link to n2731, Oct 18 2021,
\url{
http://www.open-std.org/jtc1/sc22/wg14/www/docs/n2731.pdf
}
(see in \T{C11/Literature})

\item
and the \Cpp committee (WG21),
\url{
http://www.open-std.org/jtc1/sc22/wg21/
}
which
gives
link to n3797, Oct 13 2013,
\url{
http://www.open-std.org/jtc1/sc22/wg21/docs/papers/2013/n3797.pdf
}
(see in \T{C11/Literature})

See overview note from October 2020 \url{
http://www.open-std.org/jtc1/sc22/wg21/docs/papers/2020/p0943r6.html
}

\item
Some discussions by the working groups lined by Sewell (Cerberus project)
\url{https://www.cl.cam.ac.uk/~pes20/cerberus/}

\end{itemize}

Best library reference: 
\Clang:
\url{
https://en.cppreference.com/w/c/atomic
}
\\
Cpp:
\url{
https://en.cppreference.com/w/cpp/header/stdatomic.h
}

Highest level library reference?:
\url{
https://en.cppreference.com/w/cpp/atomic
}

\subsection{Terminology}

Terminology, from ``official'' above (document N2176):%
\footnote{
See  \\
Official(?): \\
\url{https://web.archive.org/web/20180118051003/}
\\
\sloppy
\url{http://www.open-std.org/jtc1/sc22/wg14/www/abq/c17_updated_proposed_fdis.pdf} \\
Unofficial: \\
\url{https://en.cppreference.com/w/c/atomic/memory_order} \\
\url{https://en.cppreference.com/w/cpp/atomic/memory_order} \\
\url{https://en.wikipedia.org/wiki/MOESI_protocol} 
}
\footnote{
There is extra info in \cite{AutomaticallyComparingMCMs} Fig 3, S2.2, e.g., all $\seqcst$ reads have acquire semantics, all $\seqcst$ writes have release semantics, $\seqcst$ fences
are both acquire and release, and ``non-atomic reads access only non-atomic locations'', which appears to allow $\na$-writes to atomic locations, and atomic actions only
access atomic locations.
}
\begin{itemize}
\item
\emph{Conflict; C11 Standard 5.1.2.4 (4)}.
``Two expression evaluations \emph{conflict} if one of them modifies a memory location and the other one
reads or modifies the same memory location.''

Hence any program we are interested in has a ``conflict''.

\item
\emph{Data race; C11 Standard 5.1.2.4 (35)}.
``The execution of a program contains a \emph{data race} if it contains two conflicting actions in different
threads, at least one of which is not atomic, and neither happens before the other. Any such data
race results in undefined behavior.''

Hence to have a ``data race'' one must have a non-atomic access of a shared variable.  We will tend to avoid these,
and hence our examples can never have a data race; thus we are mostly dealing with defined behaviour.
Note that the definition of data race refers to ``happens before'', which I think is close to a circularity; certainly entails some
deep understanding of the memory model \emph{before} deciding whether there is a data race.

\item
\emph{Lack of multicopy atomicity? Sect 5.1.2.4 (Note 4)}
``There is a separate order for each atomic object. There is no requirement that these can be combined into a single
total order for all objects. In general this will be impossible since different threads may observe modifications to different
variables in inconsistent orders.''

\item
\emph{Happens before; Sect 5.1.2.4 (18)}
``An evaluation A happens before an evaluation B if A is sequenced before B or A inter-thread happens
before B. The implementation shall ensure that no program execution demonstrates a cycle in the
“happens before” relation''

Sequenced before is essentially program order, but noting that expression evaluation can be in different orders, and 
there is no particular ordering on sequential composition (see 5.1.2.3(3) and Annex C)

Inter-thread happens before is straightforwardish...

\item
\emph{Release sequence; Sect 5.1.2.4 (10)}.
``A release sequence headed by a release operation A on an atomic object M is a maximal contiguous
sub-sequence of side effects in the modification order of M, where the first operation is A and every
subsequent operation either is performed by the same thread that performed the release or is an
atomic read-modify-write operation.''

Given code like
\begin{equation}
	y \asgn 1 
	\scomp
	\pmodR{x \asgn 1}
	\scomp
	y \asgn 2
	\scomp
	x \asgn 2
	\sspace \pl \sspace
	\pmodA{r_1 \asgn x}
	\scomp
	r_2 \asgn y
\end{equation}
the release sequence includes both $x \asgn 1$ and $x \asgn 2$.  If the acquire in the other thread sees either write to $x$, then
$r_2 \asgn y$ is guaranteed to see $y \asgn 1$, \emph{but may not} see $y \asgn 2$.  This agrees with our usual understanding;
the release sequence concept is a confusing way of saying ordering on $x$ is (always) preserved but other orders may not be (due to \Gmm).

\end{itemize}

See also:
\begin{itemize}
\item
Section 7.17.3, ``Order and consistency''
\item
Annex C, ``Sequence points''.  This gives some order to executions, for instance, seems to imply that if-statement conditions are evaluated
before expressions inside the branches.
\item
Section 7.26.4, ``Mutexes''. This might be a good value-add for our paper

\end{itemize}

Other notes of interest:
\begin{itemize}
\item
\emph{Sect 5.1.2.4 (37, Note 19)}.
``... Reordering of atomic loads in cases in which the atomics in question may alias is also generally precluded, since
this may violate the coherence requirements.''

This seems to relate to the RSW difficulty associated with address shifting/pointer arithmetic.

\end{itemize}

\robnote{See \url{http://plv.mpi-sws.org/gps/rcu/paper.pdf}.
Discussion about second example on page 2 doesn't make sense?  Neither may "win"
}

\robnote{Can we deal with \T{kill\_dependency}, \\
\url{https://en.cppreference.com/w/cpp/atomic/kill_dependency} possibly by stripping the tags?}

\begin{itemize}
\item
``If an operation A that modifies an atomic object M happens before an operation B that modifies M,
then A shall be earlier than B in the modification order of M. [...]
The requirement above is known as “write-write coherence”.''

\item
``If a value computation A of an atomic object M happens before a value computation B of M, and A
takes its value from a side effect X on M, then the value computed by B shall either be the value
stored by X or the value stored by a side effect Y on M, where Y follows X in the modification
order of M. [...]
The requirement above is known as “read-read coherence”.''

\item
``If a value computation A of an atomic object M happens before an operation B on M, then A shall
take its value from a side effect X on M, where X precedes B in the modification order of M.
[...]
The requirement above is known as “read-write coherence”''

\item
``If a side effect X on an atomic object M happens before a value computation B of M, then the
evaluation B shall take its value from X or from a side effect Y that follows X in the modification
order of M.
[...]
The requirement above is known as “write-read coherence”.''
\end{itemize}

\subsection*{Consume}

Interesting quotes about consume:

``In fact, the only known weakly-ordered processor that does not preserve data dependency ordering is the DEC Alpha''.  

``Indeed, RCU served as motivation for adding consume semantics to C++11 in the first place.''

See \\
\url{https://preshing.com/20140709/the-purpose-of-memory_order_consume-in-cpp11/}

There was a bug in gcc wrt consume, see \\
\url{https://gcc.gnu.org/bugzilla/show_bug.cgi?id=59448} \\
This suggests the current spec is hard to work with...

Also, where ``consume'' is found on that page, there is a discussion that the compiler "lost track of" a dependency.  This goes to the overall point.  However the above thread covers 2013-2016

Also of relevance:
\url{
	https://stackoverflow.com/questions/38280633/c11-the-difference-between-memory-order-relaxed-and-memory-order-consume
}
\\
and
\url{
	http://www.open-std.org/jtc1/sc22/wg21/docs/papers/2015/p0098r0.pdf
}.

}

\OMIT{
\section{New rule...}

Silent reordering.
\[
	\Rule{
		c_2 \tra{\acb} c_2'
		\qquad
		\rocmd{\acb'}{c_1}{\acb}{\mm}
	}{
		c_1 \ppseqm c_2 
		\tra{\tau}
		\acb' \bef (c_1 \ppseqm c_2')
	}
\]

Requires '$\bef$' as a primitive.

`Fixes' the monotonicity problem...

but kills associativity

Problem case is
\[
	(\aca \ppseqm \acb) \ppseqm \acc
\]
where $\roabab$, $b \nro c$ but $\acb' \ro \acc$ (and $\aca \ro \acc$).
(With $\acb' = \fwdab$)

Refinement says, and will always say, $\aca \ppseqm \acb \refsto \acb' \ppseqs \aca$
Hence assuming monotonicity,
\[
	(\aca \ppseqm \acb) \ppseqm \acc
	\\ \refsto 
	(\acb' \ppseqs \aca) \ppseqm \acc
	\\ \refsto 
	\acc \ppseqs (\acb' \ppseqs \aca) 
\]
But in the original program, operationally there is no way for $\acc$ to occur first.
The above rule fixes that, since the reordering of b and a can happen silently and in-place.

But now associativity is broken, ie,
\[
	(\aca \ppseqm \acb) \ppseqm \acc
	\neq
	\aca \ppseqm (\acb \ppseqm \acc)
\]
There is again no way for $\acc$ to come first in the RHS, so
\[
	(\aca \ppseqm \acb) \ppseqm \acc
	\refsto
	\aca \ppseqm (\acb \ppseqm \acc)
\]
but not the other way around

Presumably $\fwdab = \acb$, or, stronger, $\nddepab$, fixes this issue.

More generally, when $\roabab \imp \acb' \ro \acc \imp \acb \ro \acc$ (all of which is only relevant if $\aca \ro \acc$ as well)

And these can be lifted to cmd level, which is nice

So, associative if no forwarding

Probably better to have associativity as conditional, rather than a complex notion for monotonicity

Note that the new operational rule allows more behaviours, in particular, now
\[
	x \asgn 1 \ppseqt r_1 \asgn x \ppseqt r_2 \asgn y
	\xtra{r_2 \asgn y}
	r_1 \asgn 1 \ppseqs x \asgn 1 
\]
So now the load can happen earlier.  Was this (refinement) prevented by submitted isabelle theory?
}

\OMIT{
\section{Peterson's lock}

See example in \reffig{petersons},
taken from Doherty et al \cite{DohertyVerifyingC11}.

\renewcommand{\f}[1]{f_#1}
\renewcommand{\c}[1]{crit_#1}
\newcommand{\teq}[1]{turn = #1}

\newcommand{\pmutex}[1]{{\mathsf{pmutex_#1}}}

\newcommand{\LAt}{\modLA{turn}}
\newcommand{\LAfa}{\modLA{\f1}}
\newcommand{\LAfb}{\modLA{\f2}}
\newcommand{\swapcmd}[3]{\modAct{#1}{#3} \asgn #2}
\newcommand{\swapt}[1]{\swapcmd{turn}{#1}{\acqrel}}

\begin{figure}
Initially $ \f1 = false \land \f2 = false \land \teq1$
\\
also $\c1 = \c2 = false$.

\begin{minipage}{0.475\textwidth}
\[
\modX{\f1 \asgn true} 
\\
\swapt{2}
\\
\While \modA{(\f2 =true)} \land \modX{(\teq2)}
\\ \qquad
\Do \Skip
\\
\mbox{\emph{\green{(Critical section)}} } \\
\\
\modR{\f1 \asgn false} 
\]
\end{minipage}
$ \pl $
\begin{minipage}{0.478\textwidth}
\[
\modX{\f2 \asgn true}
\\
\swapt{1}
\\
\While \modA{(\f1 =true)} \land \modX{(\teq1)}
\\ \qquad
\Do \Skip
\\
\mbox{\emph{\green{(Critical section)}} } \\
\\
\modR{\f2} \asgn false
\]
\end{minipage}

\centerline{Original from \cite{DohertyVerifyingC11}}

{
\begin{minipage}{0.45\textwidth}
\[
\qquad \pmutex1 \\
\modX{\f1} \asgn true
\\
\swapt{2}
\\
await \ldots
\\
\c1 \asgn true \ppseqc \c1 \asgn false
\\
\modR{\f1 \asgn false} 
\]
\end{minipage}
$ \pl $
\begin{minipage}{0.47\textwidth}
\[
\qquad \pmutex2 \\
\modX{\f2 \asgn true}
\\
\swapt{1}
\\
await \ldots
\\
\c2 \asgn true \ppseqc \c2 \asgn false
\\
\modR{\f2} \asgn false
\]
\end{minipage}

\[
await \sdef
	\Repeat 
		r_1 \asgn \modA{\f2} 
		\ppseqc
		r_2 \asgn \modX{turn}
	\Until \neg r_1 \land r_2 = 1
\]
Will have a problem with variable names overlapping in both processes

\centerline{Our encoding making critical section explicit}
}

\robnote{The acq on the setting of turn is a good example of acq/rel not ordered in C11?}

\Description{TODO}
\caption{Peterson's algorithm with release-acquire, taken from Doherty et al.}
\label{fig:petersons}
\end{figure}

In comparison to Doherty's we have made the $\relaxed$ annotations explicit, and added the $\c*$ variables to model critical sections.
Doherty et. al use a special $\mathbf swap$ command for $\acqrel$ assignment to $turn$, which models an atomic RMW operation, ie, it reads the value of $turn$ before updating.
However the current value of $turn$ appears to be ignored (by the swap), and hence it is just used because $\acqrel$ only makes sense in C11 on a read/update action.  We don't need
to worry about that so just annotate a normal assignment; the annotation itself enforces the required semantics.
The difference is important when calculating complex axiomatic relations at the top level; we just let the local reorderings deal with it.

\OMIT{
\subsubsection{Eliminating the loop}
The main structural difference is we encode the empty while loop as an abstract ``await'' statement, that is, a guard on the negation of the condition.  This is
an abstraction of the while loop, which simplifies the reasoning as we do not need to deal with loop invariants and possible reorderings between loop iterations.
The traces of our version eliminate finite sequences of loads of $\f2$ and $turn$ where the guard evaluates to $true$, but these do not affect functional
correctness.

In terms of traces, the version with the while loop has a sequence loads of $\f2$ and $turn$, followed by a single evaluation of one of them which results in the
while condition being falsified.  Our version is similar, except eliminating all of the earlier loads.  Loads do not affect the functional correctness so this is fine.

The condition itself is negated so that it can become an await statement, and then we split the resulting disjunction into a choice.  This is only justified if
each resulting guard is at least as restrictive on reorderings as the original.  Hence we must strengthen the $\relaxed$ tag on $turn$ to an $\acquire$ tag.  If
we did not then potentially $\c1 \asgn true$ could be reordered before the $turn$ guard.  
}

\OMIT{
Ian says the rely/guarantees become:
\[
	\teq1 \imp turn' = 1 
	\qquad \mbox{(derivable from $turn' = turn \lor turn' = 1$)}
	\\
	\c1 \imp \f1 \land (\neg\f2 \lor \teq1 \lor \neg\c2)
\]
}
}

\end{document}